\documentclass[11pt,a4paper]{article}
\usepackage{fullpage}
\usepackage{amsmath, amsthm, amsfonts,  amssymb, latexsym}
\usepackage[colorlinks, citecolor=Red, linkcolor=Blue, urlcolor=blue]{hyperref}
\usepackage{caption}
\usepackage[dvipsnames]{xcolor}
\usepackage{tikz}
\usetikzlibrary{patterns}
\usetikzlibrary{shapes.geometric}
\usetikzlibrary{decorations.pathmorphing}
\usetikzlibrary{arrows.meta}
\tikzset{snake it/.style={decorate, decoration=snake}}
\usetikzlibrary{positioning}
\usetikzlibrary{arrows,shapes,positioning}
\usetikzlibrary{decorations.markings}
\tikzstyle arrowstyle=[scale=1]
\tikzstyle directed=[postaction={decorate,decoration={markings,mark=at position .65 with {\arrow[arrowstyle]{stealth}}}}]
\tikzstyle reverse directed=[postaction={decorate,decoration={markings,mark=at position .65 with {\arrowreversed[arrowstyle]{stealth};}}}]

\usepackage{mcite}
\usepackage{enumerate}
\usepackage{mathrsfs}
\usepackage{graphicx}

\usepackage{authblk}

\newtheorem{proposition}{Proposition}[section]
\newtheorem{theorem}[proposition]{Theorem}

\newtheorem{definition}[proposition]{Definition}

\newcommand{\rhs}{r.h.s.\ }

\newcommand{\wrt}{w.r.t.\ }
\newcommand{\cf}{cf.\ }

\newcommand{\ud}{\mathrm{d}}
\newcommand{\del}{\partial}

\newcommand{\R}{\mathbb{R}}
\newcommand{\C}{\mathbb{C}}
\newcommand{\Z}{\mathbb{Z}}

\newcommand{\gA}{\mathfrak{A}}

\newcommand{\order}{O}

\newcommand{\sN}{\mathcal{N}}
\newcommand{\sM}{\mathcal{M}}

\newcommand{\nn}{\nonumber}
\newcommand{\beq}{\begin{equation}}
\newcommand{\eeq}{\end{equation}}

\newcommand{\defeq}{\mathrel{:=}}

\renewcommand{\Im}{{\mathrm{Im}}}
\renewcommand{\Re}{{\mathrm{Re}}}

\DeclareMathOperator{\supp}{supp}

\DeclareMathOperator{\WF}{WF}

\DeclareMathOperator{\csch}{csch}
\DeclareMathOperator{\sign}{sgn}

\newcommand{\cB}{{\mathcal{B}}}
\newcommand{\cO}{{\mathcal{O}}}
\newcommand{\cV}{{\mathcal{V}}}
\newcommand{\cD}{{\mathcal{D}}}
\newcommand{\cA}{{\mathcal{A}}}

\newcommand{\cR}{\mathcal{R}}

\newcommand{\cC}{\mathcal{C}}
\newcommand{\cH}{\mathcal{H}}

\newcommand{\cU}{\mathcal{U}}
\newcommand{\cX}{\mathcal{X}}

\newcommand{\cN}{\mathcal{N}}

\newcommand{\sT}{{\mathscr{T}}}
\newcommand{\sR}{{\mathscr{R}}}

\newcommand{\ws}{{\Omega}}

\newcommand{\ev}{+}
\newcommand{\co}{c}
\newcommand{\ca}{-}
\newcommand{\EH}{{\mathcal{H}}}
\newcommand{\CH}{{\mathcal{CH}}}
\newcommand{\dSH}{\mathcal{H}_\co}
\newcommand{\rin}{{\mathrm{in}}}
\newcommand{\rup}{{\mathrm{up}}}
\newcommand{\rout}{{\mathrm{out}}}
\newcommand{\rI}{{\mathrm{I}}}
\newcommand{\rII}{{\mathrm{II}}}
\newcommand{\rIII}{{\mathrm{III}}}
\newcommand{\rIV}{{\mathrm{IV}}}
\newcommand{\rb}{{\mathrm{b}}}
\newcommand{\rU}{{\mathrm{U}}}
\newcommand{\rF}{{\mathrm{C}}}
\newcommand{\rHH}{{\mathrm{HH}}}

\usepackage[T1]{fontenc}

\begin{document}

\title{Quantum Instability of the Cauchy Horizon in Reissner-Nordstr{\"o}m-deSitter Spacetime}

\author[1]{Stefan Hollands\thanks{stefan.hollands@uni-leipzig.de}}
\author[2]{Robert M.\ Wald\thanks{rmwa@uchicago.edu}}
\author[1]{Jochen Zahn\thanks{jochen.zahn@itp.uni-leipzig.de}}
\affil[1]{Institut f\"ur Theoretische Physik, Universit\"at Leipzig \protect\\ Br\"uderstr.\ 16, 04103 Leipzig, Germany.}
\affil[2]{Enrico Fermi Institute and Department of Physics, The University of Chicago \protect\\ 5640 South Ellis Avenue, Chicago, Illinois 60637, USA.}

\date{\today}

\maketitle

\begin{abstract}
In classical General Relativity, the values of fields on spacetime are uniquely determined by their values at an initial time within the domain of dependence of this initial data surface. However, it may occur that the spacetime under consideration extends beyond this domain of dependence, and fields, therefore, are not entirely determined by their initial data. This occurs, for example, in the well-known (maximally) extended Reissner-Nordstr{\"o}m or Reissner-Nordstr{\"o}m-deSitter (RNdS) spacetimes. The boundary of the region determined by the initial data is called the ``Cauchy horizon.'' It is located inside the black hole in these spacetimes. The strong cosmic censorship conjecture asserts that the Cauchy horizon does not, in fact, exist in practice because the slightest perturbation (of the metric itself or the matter fields) will become singular there in a sufficiently catastrophic way that solutions cannot be extended beyond the Cauchy horizon. Thus, if strong cosmic censorship holds, the Cauchy horizon will be converted into a ``final singularity,'' and determinism will hold. Recently, however, it has been found that, classically this is not the case in RNdS spacetimes in a certain range of mass, charge, and cosmological constant. In this paper, we consider a quantum scalar field in RNdS spacetime and show that quantum theory comes to the rescue of strong cosmic censorship. We find that for any state that is nonsingular (i.e., Hadamard) within the domain of dependence, the expected stress-tensor blows up with affine parameter, $V$, along a radial null geodesic transverse to the Cauchy horizon as $T_{VV} \sim C/V^2$ with $C$ independent of the state and $C \neq 0$ generically in RNdS spacetimes. This divergence is stronger than in the classical theory and should be sufficient to convert the Cauchy horizon into a singularity through which the spacetime cannot be extended as a (weak) solution of the semiclassical Einstein equation. This behavior is expected to be quite general, although it is possible to have $C=0$ in certain special cases, such as the BTZ black hole.

\end{abstract}

\section{Introduction}

The Reissner-Nordstr\" om-deSitter (RNdS) spacetime describes a charged, static and spherically symmetric ``eternal'' 
black hole in deSitter spacetime. It is an exact solution to the Einstein-Maxwell field equations with a positive cosmological constant $\Lambda$. The part of the maximally extended spacetime relevant for the discussions in this paper is drawn in fig.~\ref{fig:1}, see for example \cite{HintzVasy, DafermosShlapentokhRothman18} or sec.~\ref{sec:2} for detailed discussion.

\usetikzlibrary{decorations.pathmorphing}
\tikzset{zigzag/.style={decorate, decoration=zigzag}}
\begin{figure}
\centering
\begin{tikzpicture}[scale=1]
 \draw (0,0) -- (-2,2) node[midway, below, sloped]{$\EH^-$};
 \draw (-2,2) -- (-4,4) node[midway, below, sloped]{$\EH^L$};
 \draw (0,0) -- (2,2) node[midway, below, sloped]{$\dSH^-$};
 \draw (2,2) -- (0,4) node[midway, above, sloped]{$\dSH^L$};
 \draw (-2,2) -- (0,4) node[midway, above, sloped]{$\EH^R$};
 \draw (0,4) -- (-2,6) node[midway, above, sloped]{$\CH^R$};
 \draw (-4,4) -- (-2,6) node[midway, above, sloped]{$\CH^L$};

 \draw[green, thick] (-4,4) .. controls (0,2) .. (4,4);

 \draw[snake it] (0,4) -- (0,8);
 \draw (-2,6) -- (0,8) node[midway, above, sloped]{$\CH^+$};
 \draw[double] (0,4) -- (4,4);
 \draw (2,2) -- (4,4) node[midway, below, sloped]{$\dSH^R$};
 \draw (0,2) node{${\rm I}$};
 \draw (2,3) node{${\rm III}$};
 \draw (-2,4) node{${\rm II}$};
 \draw (-1,6) node{${\rm IV}$};

 \draw[fill=white] (-4,4) circle (2pt);
 \draw[fill=white] (4,4) circle (2pt);
 \draw[fill=white] (0,8) circle (2pt);
 \draw[fill=black] (-2,6) circle (2pt);
 \draw[fill=white] (0,0) circle (2pt) node[below]{$i^-$};
 \draw[fill=black] (-2,2) circle (2pt);
 \draw[fill=white] (0,4) circle (2pt) node[above right]{$i^+$};
 \draw[fill=black] (2,2) circle (2pt);
  
\draw[densely dotted] (0,0) .. controls (-0.8,4) .. (0,5);
\end{tikzpicture}
\caption{Sketch of partially extended RNdS. The wiggled line indicates the curvature singularity and the double line conformal infinity. Filled circles correspond to bifurcation surfaces, whereas open circles indicate singular points and/or points at infinity. Also indicated, as a green line, is a Cauchy surface for the partially extended RNdS, up to the region IV beyond the Cauchy horizon $\CH^R$. To the \rhs of the dotted curve is the outer region of a shell in a typical collapse spacetime.}
\label{fig:1}
\end{figure}
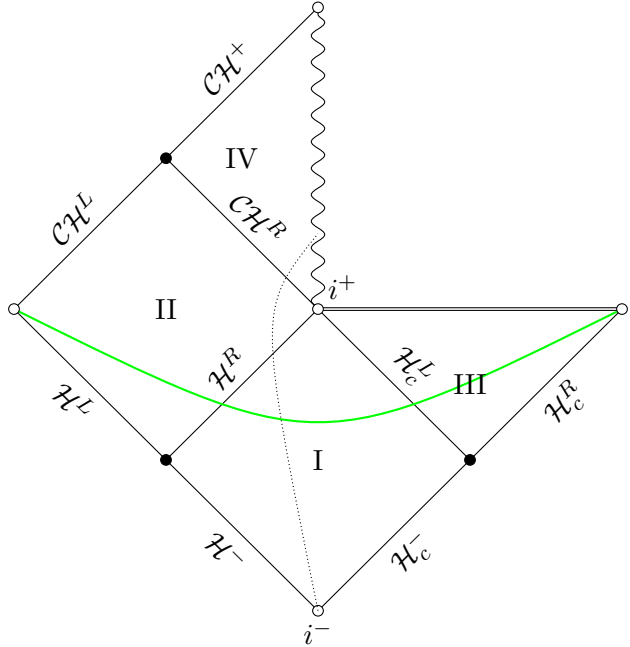

The regions II and IV are inside the black hole, the regions I, III outside, and the green line is a spacelike surface, $\Sigma$. By the usual properties of solutions to hyperbolic field equations such as the covariant Klein-Gordon wave equation, 
\beq
(\square-\mu^2) \Phi=0 \, , 
\label{eq:KGeq}
\eeq
if we prescribe initial data for $\Phi$ on such a surface, then $\Phi$ will be determined uniquely in the ``domain of dependence'' of  $\Sigma$, which in the case at hand is the union of regions I, II, III. However, $\Phi$ is not determined uniquely beyond---i.e., in particular, in region IV---representing a breakdown of determinism. The horizon separating region IV from the rest is called the ``Cauchy horizon''. A similar situation occurs for the more general Kerr-Newman-dS black holes (not discussed in this paper) and for other fields such as the linearized gravitational field. 
 
The maximally extended RNdS spacetime also extends through the surface labeled as $\CH^L$ in fig.~\ref{fig:1} to a region isometric to region IV. The surface $\CH^L$ also comprises part of the full Cauchy horizon of the maximally extended RNdS spacetime. However, an eternal static black hole is physically unrealistic. For a black hole produced by the collapse of a charged, spherically symmetric body, the region exterior to the body would be described by the RNdS spacetime, but the region inside the body would be replaced by a suitable ``interior metric.'' The dotted line in fig.~\ref{fig:1} represents a typical trajectory of the surface of a collapsing body. The region of RNdS spacetime to the right of this line is thus physically relevant for the gravitational collapse spacetime. However, the region to the left of this line would be ``covered up'' by the collapsing body and is not physically relevant. The key point is that the singularity in the region beyond $\CH^L$ will always be ``covered up'' \cite{Boulware73}. Indeed, all of $\CH^L$ and part of $\CH^R$ may be ``covered up,'' as is the case for the trajectory shown in fig.~\ref{fig:1}. But even if $\CH^L$ is not fully covered up, it will no longer play the role of a Cauchy horizon in a gravitational collapse spacetime. On the other hand, a portion of $\CH^R$ ``near'' the event horizon will never be covered up and will correspond to a portion of the Cauchy horizon of the gravitational collapse spacetime. In this paper, we will consider the behavior of a quantum scalar field as one approaches $\CH^R$ in the extended RNdS spacetime shown in fig.~\ref{fig:1}. Our results will be applicable to the portion of the gravitational collapse spacetime outside of the collapsing body.

It was argued many years ago by Penrose \cite{PenroseSCC} that dynamical fields such as $\Phi$ would have a very bad behavior at the Cauchy horizon $\CH^R$, thereby converting the Cauchy horizon to a singularity. The conjecture that the maximal Cauchy evolution of suitable initial data generically yields an inextendible spacetime is known as {\em strong cosmic censorship} (sCC). If sCC holds in general relativity, then there would be no breakdown of determinism. Penrose's original argument\footnote{See \cite{IsraelPoissonMassInflation} for a closely related argument formulated in terms of ``mass-inflation''.} for the instability of $\CH^R$, phrased for Reissner-Nordstr\"om (RN) spacetime where $\Lambda=0$, is that a timelike observer staying in the external region, such as A in fig.~\ref{fig:2}, will reach future infinity $i^+$ at infinite proper time, while an infalling observer such as B in fig.~\ref{fig:2} will reach the Cauchy horizon in finite proper time. Thus, if A sends periodic (according to her time) light signals into the black hole, illustrated by the green rays in fig.~\ref{fig:2}, these will arrive more and more frequently at B according to his time. Similarly, source free solutions $\Phi$ that oscillate only moderately near $\mathcal{I}^+$ will oscillate extremely rapidly near $\CH^R$, i.e., there will be an infinite ``blueshift effect'' as one approaches $\CH^R$. Such a rapidly oscillating field would have an infinite transversal derivative at $\CH^R$, thus resulting in a singular stress tensor. In the full theory, one might expect this behavior to be sufficiently singular to render the Einstein equations ill defined at $\CH^R$, thus ``solving'' the problem of indeterminism -- or rather relegating it to the domain of quantum gravity taking over near this singularity. 
 
 \begin{figure}
\centering
\begin{tikzpicture}[scale=1.0]

\draw[green, directed, thick] (1.2,1.2) -- (-2.8,5.2);
\draw[green, directed, thick] (1.6,1.6) -- (-2.4,5.6);
\draw[green, directed, thick] (1.8,1.8) -- (-2.2,5.8);
\draw[green, directed, thick] (1.9,1.9) -- (-2.1,5.9);
\draw[green, directed, thick] (1.95,1.95) -- (-2.05,5.95);

\draw[white,fill=white] (0,0) .. controls (-0.8,4) .. (0,5) -- (-2,6) -- (-4,4) -- (0,0);
\draw[white,fill=white] (0,0) .. controls (0.3,2.3) .. (0,4) -- (2,2) -- (0,0);
\draw[blue,directed, ultra thick] (0,0) .. controls (-0.8,4) .. (0,5) node[midway,above left]{${\rm B}$};
\draw[red,directed, ultra thick] (0,0) .. controls (0.3,2.3) .. (0,4) node[midway,above right]{${\rm A}$};

 \draw (0,0) -- (-2,2) node[midway, below, sloped]{$\EH^-$};
 \draw (-2,2) -- (-4,4) node[midway, below, sloped]{$\EH^L$};
 \draw (0,0) -- (2,2) node[midway, below, sloped]{$\mathcal{I}^-$};
 \draw (2,2) -- (0,4) node[midway, above, sloped]{$\mathcal{I}^+$};
 \draw (-2,2) -- (0,4) node[midway, above, sloped]{$\EH^R$};
 \draw (0,4) -- (-2,6) node[midway, above, sloped]{$\CH^R$};
 \draw (-4,4) -- (-2,6) node[midway, above, sloped]{$\CH^L$};

 \draw[snake it] (0,4) -- (0,8);
 \draw (-2,6) -- (0,8) node[midway, above, sloped]{$\CH^+$};
 \draw (0,2) node{${\rm I}$};
 \draw (-2,4) node{${\rm II}$};
 \draw (-1,6) node{${\rm IV}$};

 \draw[fill=white] (-4,4) circle (2pt);
 \draw[fill=white] (0,8) circle (2pt);
 \draw[fill=black] (-2,6) circle (2pt);
 \draw[fill=white] (0,0) circle (2pt);
 \draw[fill=black] (-2,2) circle (2pt);
 \draw[fill=white] (0,4) circle (2pt) node[above right]{$i^+$};
 \draw[fill=black] (2,2) circle (2pt);

\end{tikzpicture}
\caption{Illustration of the blueshift effect in Reissner-Nordstr\"om spacetime.}
\label{fig:2}
\end{figure}
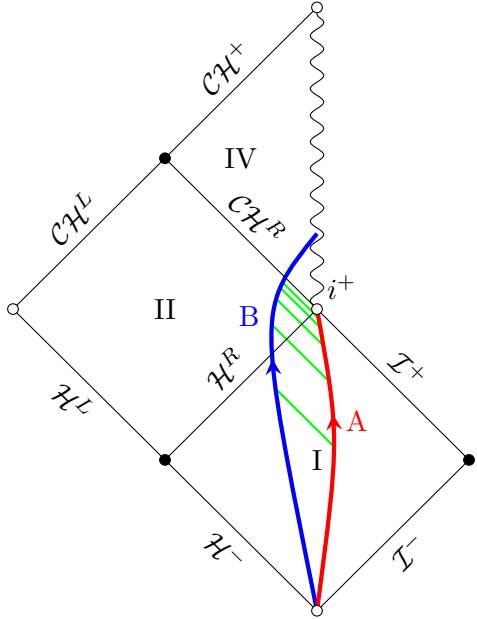

 On the other hand, as observed also some time ago \cite{MellorMoss90, MellorMoss92, BradyMossMyers}, in the presence of a positive cosmolgical constant, 
 the amplitude of a dynamical field $\Phi$ is expected to decay exponentially in region I (and on $\EH^R$) due to a similar ``red-shift effect'' related to the cosmological expansion. In principle, this decay could counterbalance the blue-shift effect, thus leading only to a mild -- if any -- singularity of the stress tensor at the Cauchy horizon, implying a potential violation of sCC. The earlier investigations of  \cite{MellorMoss90, MellorMoss92, BradyMossMyers} remained to a certain extent inconclusive, partly because the precise conclusions depend on the regularity of the initial data for $\Phi$ (see e.g. \cite{DiasEtAlRough} for a detailed discussion and \cite{DafermosShlapentokhRothman18} who advocate less regular initial data).
More recently, the decay of fields near the Cauchy horizon was analyzed in a mathematically rigorous manner both in RN and RNdS spacetime, see \cite{McNamara1978, Dafermos2003, Dafermos2005, DafermosShlapentokhRothman18, DafermosLukI, Dafermos2014, HintzVasy, SbierskiPhD, LukOh2015, CGNSa, CGNSb, CGNSc, Franzen2014, CostaFranzen2016}. Concretely, it was proven by Hintz and Vasy \cite{HintzVasy} (partly building on earlier work by \cite{BarretoZworski1997, BonyHaefner2008, Dyatlov2011, WunschZworski2011, Vasy2013, Dyatlov2014, Dyatlov2015, NonnenmacherZworski2015, HintzVasy2018}) that, starting from smooth initial data e.g.\ on $\Sigma$, the solution $\Phi$ has Sobolev regularity $H^{1/2+\beta-\epsilon}$ for any small $\epsilon>0$ locally near $\CH^R$. Here the parameter governing regularity at the Cauchy horizon is 
\beq
\label{eq:betadef}
\beta = \frac{\alpha}{\kappa_{\ca}} .
\eeq
In this expression, $\kappa_\ca$ is the surface gravity of the Cauchy horizon and $\alpha$ the spectral gap, i.e., the infimum 
\beq
\label{eq:aldef}
\alpha={\rm inf} \{- \Im (\omega_{n \ell m}) \}, 
\eeq 
with $\omega_{n \ell m}$
the so-called quasi-normal frequencies \cite{Nollert} of region I. Together with their results on ``conormal regularity'' at the Cauchy horizon, this implies via the Sobolev embedding 
theorem that the field $\Phi$ is H\" older-continuous ($C^{0,\beta-\epsilon}$) across the Cauchy horizon $\CH^R$. Taking the field $\Phi$ as a proxy for a dynamical metric in the full Einstein-Maxwell equations, this would correspond to a rather weak singularity in the sense that an observer crossing $\CH^R$ would not be crushed/stretched by an infinite amount\footnote{A similar conclusion for RN, where $\beta=0$, was reached a long time ago in the context of the Israel-Poisson mass inflation scenario \cite{IsraelPoissonMassInflation,OriPRL67_1991} and 
later in a mathematically rigorous fashion by \cite{DafermosLukI}.}.

Taking derivatives, the Sobolev regularity
translates into a behavior not more singular than  
 \beq
 \label{eq:ClassResult}
T_{VV} \sim D |V |^{-2+2\beta} \quad \text{as $V \to 0^-$}
\eeq
 for the classical stress-energy component $T_{VV}=(\partial_V \Phi)^2$ near the Cauchy horizon $\CH^R$ at $V=0$, where $V$ denotes a Kruskal-type null coordinate. (Here $D$ depends on the solution $\Phi$.) This result is roughly consistent with the results of \cite{BradyMossMyers}. The spectral gap, $\alpha$, which vanishes for $\Lambda = 0$, has recently been analyzed in detail by Cardoso et al. \cite{CardosoEtAl18,CardosoEtAl19} for $\mu=0$ in this context. They have given strong numerical evidence that there are black hole parameters for which, e.g.\  $\beta> \frac{1}{2}$. This would correspond to a singularity that is sufficiently weak that a solution could be extended in different ways as a weak solution across the Cauchy horizon; see e.g.\ \cite{DafermosLukI} for a more detailed discussion of such matters. Thus, it seems that, from the classical viewpoint, 
RNdS black holes can violate sCC for certain ranges of the black hole parameters. 
 
In this work, we therefore ask whether quantum effects can change this classical picture in an essential way. For this, we consider 
a free, real Klein-Gordon quantum field eq.~(\ref{eq:KGeq}) propagating in RNdS. First, we set up a class of quantum states that are regular, in a natural sense, near the initial data surface such as $\Sigma$ in fig.~\ref{fig:1} -- technically, they are of ``Hadamard type''.  This condition is necessary and sufficient in order to 
have a smoothly varying, and in particular finite, expectation value of any ``composite field'' (i.e.\ a suitably renormalized \cite{HollandsWaldWick}
monomial in covariant derivatives of $\Phi$) such as the stress tensor in the regions I, II, III, but not necessarily at the Cauchy horizon $\CH^R$. In this sense, Hadamard 
states are the quantum analog of classical solutions that are smooth on an initial data surface $\Sigma$. 
They are thus the natural class of states to work with because -- at least in our opinion -- 
the statement of sCC becomes obscure if we work with states having singularities already to begin with\footnote{
Our arguments rest on the assumption of regularity across the cosmological horizon, or the weaker condition that $T_{vv}$ vanishes on the cosmological horizon, with $v$ the ``Killing parameter''. This condition has a natural limit in RN, that of absence of radiation infalling from spatial infinity. It seems natural to discuss the behavior at the Cauchy horizon in terms of such asymptotic conditions, as without asymptotic conditions also the classical consideration of strong cosmic censorship becomes meaningless, as discussed for example in the Introduction to \cite{DafermosLukI}: The domain of dependence of a compact achronal set in any globally hyperbolic spacetime has a Cauchy horizon, across which it is of course smoothly extendible.
}, i.e.\ states which fail to be Hadamard near $\Sigma$ already. 

Our main result is that the expectation value of the stress tensor operator component $T_{VV}$ in such a state $\Psi$
generically diverges as 
\beq
\label{eq:MainResult}
\langle T_{VV} \rangle_\Psi \sim C |V |^{-2} + t_{VV} \quad \text{as $V \to 0^-$}
\eeq
in terms of the Kruskal type coordinate $V$ adapted to the Cauchy horizon (see sec.~\ref{sec:2} for details), where $C$ is some function that, in the limit $V \to 0^-$, goes to a constant, depending only on the black hole parameters $M,Q,\Lambda$, but {\em not} on the state $\Psi$. In the following, we will identify $C$ with its limit and treat it as a constant, but it should be kept in mind that there also may be subleading divergences associated with the term $C|V|^{-2}$ in \eqref{eq:MainResult}. By contrast, $t_{VV}$ is a piece depending on the state $\Psi$, 
diverging, however, no more strongly than the stress tensor of a classical solution with smooth initial data, so $t_{VV}$ roughly behaves as \eqref{eq:ClassResult}, i.e.\ $t_{VV} \sim  |V |^{-2+2\beta}$. 
Our investigations, which partly rely on numerical calculations, indicate that $C \neq 0$ for generic values of the black hole parameters, 
consistent with recent results of \cite{Zilberman:2019buh} for the case of RN, who have also determined the constant $C$ in that case; see also \cite{Sela18} for related results. We expect a similar behavior 
in Kerr-Newman-dS.\footnote{The methods of our paper also apply to this more general case in principle.}

The proof of eq.~\eqref{eq:MainResult} can be given quite easily for the case of a massless scalar field in 2 spacetime dimensions. In this case, the classical divergence \eqref{eq:ClassResult} of the stress tensor component $T_{VV}$ at $\CH^R$ becomes simply $\ T_{VV}  \sim D |V|^{-2+2\frac{\kappa_\co}{\kappa_\ca}}$. As we shall show in sec.~\ref{sec:2d}, the behavior of the quantum stress tensor is determined by the trace anomaly and the continuity equation alone. We find that eq.~\eqref{eq:MainResult} holds, with $C \sim \kappa_\co^2-\kappa_\ca^2$. Our arguments 
are technically rather similar to considerations made a long time ago
by \cite{ChristensenFulling77,BirrellDavies78,MarkovicUnruh}, although there are some minor differences in how we set up the states. By contrast, our arguments for eq.~\eqref{eq:MainResult} in 4 dimensions are considerably more involved. Our analysis in 4 dimensions involves arguments from microlocal analysis (to estimate $t_{VV}$) and a mode scattering computation in regions I,II (to obtain a mode sum expression for $C$, which is then evaluated numerically). The scattering coefficients arising in the mode sum are evaluated only in the conformally coupled case in which $\mu^2=\frac{1}{6}R=\frac{2}{3}\Lambda$, and the numerical evaluation is done only in this case as well. Nevertheless, the fact that we find $C \neq 0$ in this case provides strong evidence that $C \neq 0$ generically.\footnote{As we shall discuss in sec.~\ref{sec:BTZ}, this divergent behavior differs significantly from the results found for the BTZ-black hole spacetimes by \cite{Dias:2019ery}, since $C=0$ in that case.}

That the divergence of the state dependent term $t_{VV}$ in \eqref{eq:MainResult} is not worse than that of the stress tensor of a classical solution is proven in subsection \ref{sec:Reduction} for the range $\beta > \frac{1}{2}$. This is not a serious restriction, as this is precisely the range in which sCC is potentially violated classically, see e.g.\ \cite{DafermosLukI}. It is remarkable that, provided $C \neq 0$, the strength of the singularity is state independent, so that the blow-up at the Cauchy horizon is not restricted to ``generic initial data'' as in classical formulations of sCC. Restoring units, one checks that $C$ is of $O(\hbar)$, so that the  leading divergence \eqref{eq:MainResult} is caused entirely by quantum effects.

The divergence \eqref{eq:MainResult} is sufficiently strong that, if backreaction of the quantum stress tensor on the metric were considered via the semiclassical Einstein equation, it would stretch/crush observers on geodesics crossing the Cauchy horizon into region IV. 
With a stress-energy tensor of the form \eqref{eq:MainResult}, the Ricci-tensor should behave as $R_{\lambda\lambda} :=
R_{\mu\nu} \dot \gamma^\mu \dot \gamma^\nu \sim C\lambda^{-2}$ for a null geodesic $\gamma$ crossing the Cauchy horizon at the value $\lambda=0$ 
of its affine parameter. In contrast with the behavior found in the classical theory by \cite{IsraelPoissonMassInflation,OriPRL67_1991} and 
\cite{McNamara1978, Dafermos2003, Dafermos2005, DafermosShlapentokhRothman18, DafermosLukI, Dafermos2014, HintzVasy, SbierskiPhD, LukOh2015, CGNSa, CGNSb, CGNSc, Franzen2014, CostaFranzen2016}, this would correspond to 
a ``strong'' rather than a ``weak'' singularity, e.g. in the sense of \cite{TiplerPhysicsLetters64A_1977,Clarke_CambridgeUniversityPressLectureNotes1993}. In more detail, let $Z_1, Z_2$ be two Jacobi fields normal to the null geodesic $\gamma$. Denote by $n$ a null vector field 
along $\gamma$ such that $n_\mu \dot \gamma^\mu = 1$ and $0=Z_1^\mu n_\mu=Z_2^\mu n_\mu$.
We may use these to define an area element $\epsilon = \frac{1}{2} Z_1 \wedge Z_2$ at each point of $\gamma$. Furthermore, the area of this element relative to that 
induced by the metric is $A$ where $6 n_{[\mu} \dot \gamma_{\nu} \epsilon_{\alpha\beta]} = A\eta_{\mu\nu\alpha\beta}$. Writing $|A|=x^2$ we find that
the geodesic deviation equation implies in the usual way that
\beq
\frac{\ud}{\ud \lambda} (x^2 \sigma^A{}_B) = -x^2 C^A{}_{\lambda B \lambda}, \quad
\frac{\ud^2}{\ud \lambda^2} x = -\frac{1}{2}(2\sigma^A{}_{B} \sigma^B{}_A + R_{\lambda\lambda})x,
\eeq
with $\sigma^A{}_{B}, A,B=1,2$, being the components of the shear tensor 
in a parallely propagated frame and $C^A{}_{\lambda B \lambda}$ being the components of the Weyl tensor (see, e.g., \cite{WaldGR} for details). Since $R_{\lambda\lambda} \sim C \lambda^{-2}$,
the functions $x,\sigma^{A}{}_B$ must clearly behave badly as $\lambda \to 0$ unless $C=0$, which as we have argued, should not be the generic case. 
In fact, if $C>0$, we immediately learn from the second equation that $x \to 0$ as the Cauchy horizon is reached, which corresponds to an 
infinite amount of crushing. (A similar conclusion can be reached for time-like geodesics.) On the other hand, when $C<0$, to avoid an infinite stretching 
(i.e.\ $x\to \infty$ as $\lambda \to 0$), one eigenvalue of $\sigma^{A}{}_B$ should diverge as $\lambda^{-1}$. Then by the first equation, if $x$ stays away from 
zero as $\lambda \to 0$, at least one eigenvalue of $C^A{}_{\lambda B \lambda}$ must diverge as $\lambda^{-2}$. Then, by the Jacobi equation for 
$Z^A_i$, at least one of the Jacobi fields must go to infinity, thus resulting in an infinite stretching. This kind of behavior presumably would convert the Cauchy horizon into a singularity through which the spacetime could not be extended as a (weak) solution of the semiclassical Einstein equation. If so, quantum effects can be said to rescue sCC, at least in the case of RNdS.

This paper is organized as follows: In section \ref{sec:2}, we introduce RNdS spacetime and some standard coordinate systems. In sec.~\ref{sec:2d}, we analyze the 2 dimensional case. 
In sec.~\ref{sec:3}, we recall the basic setup for linear quantum field theory, including
the construction of Hadamard states by prescribing data on suitable null surfaces (``null quantization''). We construct the Unruh state in RNdS spacetime and show it is Hadamard in the union of regions I, II, and III.
In sec.~\ref{sec:5}, we treat the 4 dimensional case. This requires (i) showing that the difference in the stress-energy of the Unruh state and an arbitrary Hadamard state behaves as in the classical case and (ii) performing a mode sum analysis to estimate the behavior of the Unruh state stress-energy. Numerical evaluation of the mode sum is required to argue that generically, $C \neq 0$ in \eqref{eq:MainResult}. In sec.~\ref{sec:BTZ} we compare our results to the case of the BTZ black hole recently analyzed by \cite{Dias:2019ery}, where, in fact, $C=0$.
In sec.~\ref{sec:Conclusion}, we draw our conclusions. 

\section{RNdS spacetime and coordinate systems}
\label{sec:2}

The metric of RNdS spacetime is
\beq
\label{eq:metric}
 g = - f(r) \ud t^2 + f(r)^{-1} \ud r^2 + r^2 \ud \Omega^2,
\eeq
where $\ud \Omega^2$ is the line element of the unit 2-sphere, and where
\beq
\label{frnds}
 f(r) = 1 - \frac{2 M}{r} + \frac{Q^2}{r^2} - \frac{\Lambda}{3} r^2,
\eeq
with $M > 0$, $Q$ and $\Lambda >0$ having the interpretation of the black hole mass, charge and the cosmological constant, respectively.
For physically reasonable\footnote{I.e., solutions representing a black hole hidden from the outside by an event horizon, as opposed to a representing a naked singularity visible from the outside.} parameters, assumed throughout, this function has three positive roots $r_\ca < r_\ev < r_\co$, corresponding to the Cauchy, the event, and the cosmological horizon, respectively. In terms of these, one may express $f$ as
\beq
 f(r) = \frac{(r-r_\ca) (r- r_\ev) (r_\co - r) (r - r_o)}{r^2 (r_o^2 - r_\ca r_\co - r_\ev r_\ca - r_\ev r_\co) }
\eeq
with 
\beq
\label{eq:def_r_o}
 r_o = - (r_\co + r_\ev + r_\ca)
\eeq
being the the fourth, negative, root. The quantities
\beq
\label{eq:def_kappa}
 \kappa_X = \frac{1}{2} | f'(r_X) |, \qquad X \in \{ \ca, \ev, \co, o \}
\eeq
are the corresponding surface gravities. One has $\kappa_\ca \geq \kappa_\ev$, with equality in the extremal case $r_\ca = r_\ev$. Far from extremality, we have $\kappa_\ca, \kappa_\ev > \kappa_\co$, but in the extremal limit $r_\ca \to r_\ev$, one has $\kappa_\ca, \kappa_\ev \to 0$ while $\kappa_\co$ stays finite. We will refer to the parameter range where $\kappa_\ca < \kappa_\co$ as the near extremal regime.

We frequently use the tortoise coordinate $r_*$, defined up to an integration constant by $\ud r_* = f(r)^{-1} \ud r$. For later convenience\footnote{This becomes relevant when 
matching the phases of the various mode functions across $r=r_\ev$ in sec.~\ref{sec:modes1}.}, we choose the integration constant in the interior region such that, near $r=r_\ev$,
\beq
\label{eq:r*IntegrationConstant}
 r_* =  \frac{1}{2 \kappa_\ev} \log | r-r_\ev | + D + \order(r-r_\ev),
\eeq
with the {\it same} integration constant $D$ for positive and negative $r-r_\ev$. In terms of the tortoise coordinate, the metric reads
\beq
 g = f(r) (- \ud t^2 + \ud r_*^2 ) + r^2 \ud \Omega^2.
\eeq
In the exterior region, $r_\ev < r < r_\co$, the coordinate $r_*$ diverges to $- \infty$ towards the event horizon and to $+ \infty$ towards the cosmological horizon. Choosing further
\beq
\label{eq:def_u_v}
 u \defeq t - r_*, \qquad v \defeq t + r_*,
\eeq
leads to
\beq
 g = - f(r) \ud u \ud v + r^2 \ud \Omega^2.
\eeq
With these definitions, $u$ diverges to $+ \infty$ at the event horizon $\EH^R$, while $v$ diverges to $+ \infty$ on the cosmological (de Sitter) horizon $\dSH^L$, \cf Figure~\ref{fig:ConformalDiagram}. Choosing
\beq
 U \defeq - e^{- \kappa_\ev u},
\eeq
one can extend $U$ (and the metric) analytically from $(-\infty, 0)$ to $(-\infty, \infty)$ over the event horizon $\EH^R$ into the interior of the black hole. In this interior region, $r_*$ goes from $-\infty$ at $\EH^R$ to $+\infty$ at $\CH^L$.

Furthermore, choosing
\beq
\label{eq:vco}
 V_\co \defeq - e^{- \kappa_\co v},
\eeq
one can extend $V_\co$ (and the metric) analytically from $(-\infty, 0)$ to $(-\infty, \infty)$ over the cosmological horizon $\dSH^L$. Similarly, choosing
\beq
\label{eq:V_-}
 V_\ca \defeq - e^{- \kappa_\ca v},
\eeq
one can extend $V_\ca$ (and the metric) analytically from $(-\infty, 0)$ to $(-\infty, \infty)$ over the Cauchy horizon $\CH^R$.
Figure~\ref{fig:ConformalDiagram} gives a sketch of the geometry and indicates the coordinates and their range. It is sometimes useful to also consider a $u$ coordinate in the interior region II, where it can be defined as
\beq
 u = - \frac{1}{\kappa_\ev} \log U.
\eeq
It diverges to $+\infty$ on the event horizon $\EH^R$ and to $- \infty$ on $\CH^L$, \cf also Figure~\ref{fig:ConformalDiagram}. Due to \eqref{eq:r*IntegrationConstant}, this is consistent with \eqref{eq:def_u_v}, also in region $\rII$.
Finally, a coordinate adapted to $\CH^L$, i.e., that allows for an extension through $\CH^L$, can be defined as
\beq
\label{eq:U_-}
 U_\ca = e^{\kappa_\ca u} = U^{-\frac{\kappa_\ca}{\kappa_\ev}}.
\eeq

\begin{figure}
\centering
\begin{tikzpicture}[scale=1]
 \draw (0,0) -- (-2,2) node[midway, below=0.1cm, sloped]{$\EH^-$};
 \draw (-2,2) -- (-4,4) node[midway, below=0.1cm, sloped]{$\EH^L$};
 \draw (0,0) -- (2,2) node[midway, below=0.1cm, sloped]{$\dSH^-$};
 \draw[dashed] (2,2) -- (0,4) node[midway, above, sloped]{$\dSH^L$};
 \draw[dashed] (-2,2) -- (0,4) node[midway, above, sloped]{$\EH^R$};
 \draw[dashed] (0,4) -- (-2,6) node[midway, above, sloped]{$\CH^R$};
 \draw (-4,4) -- (-2,6) node[midway, above=0.1cm, sloped]{$\CH^L$};
 \draw[snake it] (0,4) -- (0,8);
 \draw (-2,6) -- (0,8) node[midway, above=0.05cm, sloped]{$\CH^+$};
 \draw[double] (0,4) -- (4,4);
 \draw (2,2) -- (4,4) node[midway, below=0.05cm, sloped]{$\dSH^R$};
 \draw[green] (-4,4) .. controls (0,2) .. (4,4);

 \draw[fill=white] (-4,4) circle (2pt);
 \draw[fill=black] (-2,6) circle (2pt);
 \draw[fill=white] (0,0) circle (2pt);
 \draw[fill=black] (-2,2) circle (2pt);
 \draw[fill=white] (0,4) circle (2pt);
 \draw[fill=black] (2,2) circle (2pt);
 \draw[fill=white] (0,8) circle (2pt);
 \draw[fill=white] (4,4) circle (2pt);
 
 \draw[-Latex,blue] (0.1,-0.1) -- (2.1,1.9) node[below right]{$v$};
 \draw[-Latex,red] (0.05,-0.05) -- (4.05,3.95) node[below right]{$V_\co$};
 \draw[-Latex,blue] (-0.1,-0.1) -- (-2.1,1.9) node[below left]{$u$};
 \draw[-Latex,blue] (-4.1,3.9) -- (-2.1,1.9);
 \draw[-Latex,red] (-0.05,-0.05) -- (-4.05,3.95) node[below left]{$U$};
 \draw[-Latex,blue] (-4.1,4.1) -- (-2.1,6.1) node[above left]{$v$};
 \draw[-Latex,red] (-4.05,4.05) -- (-0.05,8.05) node[above left]{$V_\ca$};
\end{tikzpicture}
\caption{Same as fig.~\ref{fig:1}, but with an indication of the coordinates used.}
\label{fig:ConformalDiagram}
\end{figure}
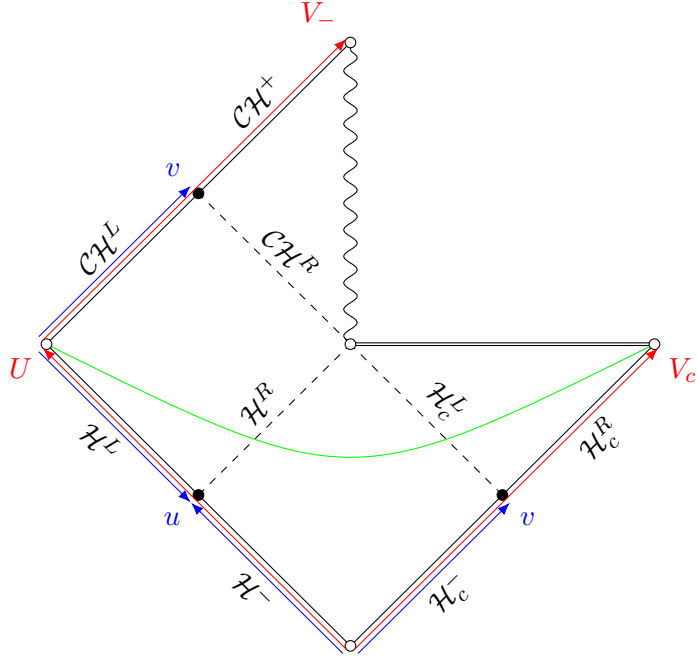

For our arguments in sections~\ref{sec:UnruhFinal} and \ref{sec:Reduction}, we will need estimates on the asymptotic behavior of classical solutions near $i^+$ (or analogously at $i^-$).
This has been analyzed in the present context by Hintz and Vasy \cite{HintzVasy} and is best expressed resolving $i^+$ ($i^-$) by a blow up procedure also introduced in \cite{HintzVasy}. We first concentrate on region $\rI$ where $r_\ev \le r \le r_\co$. In this region we can define, following \cite{HintzVasy}, a new time coordinate (basically the Kerr-star coordinate) by
\beq
\label{eq:d_t_*}
 \ud t_* = \ud t \pm [f(r)^{-1}+c(r)] \ud r.
\eeq
Here $c(r)$ is a smooth function on the interval $[r_\ev, r_\co]$, up to the boundary, such that 
\beq
 c(r) \begin{cases} = - f(r)^{-1} & r \in ( r_\ev + \delta, r_\co - \delta ), \\ \in (-2 f(r)^{-1}, 0) & r \in [r_\ev, r_\ev + \delta / 2) \cup ( r_\co - \delta / 2, r_\co]. \end{cases}
\eeq
Note that near $r_\ev$ and $r_\co$, the function $c(r)$ is not fully prescribed, but only restricted to a ($r$ dependent) interval.
The $+$ sign in \eqref{eq:d_t_*} is chosen near $r_\ev$ and the $-$ sign at $r_\co$. For $r \in  ( r_\ev + \delta, r_\co - \delta )$, $\ud t_*$ coincides with $\ud t$, and the modifications near $r_\ev$, $r_\co$ ensure that $t_*$ can be extended as a time-like coordinate to $r < r_\ev$, $r > r_\co$. In particular, in the coordinate system $(t_*, r, \ws)$ thus set up 
for $r \in [r_\ev, r_\co]$, the metric takes the form \eqref{eq:metric} for $r \in (r_\ev+\delta, r_\co-\delta)$, and 
\beq
g = -f(r) \ud t_*^2 \pm 2[1 +  c(r)f(r)] \ud r \ud t_*  - c(r)[2 + c(r)f(r)] \ud r^2 + r^2 \ud \Omega^2 
\eeq
for $r \in [r_\ev, r_\ev+\delta/2) \cup  (r_\co-\delta/2,r_\co]$. The coordinate transformations has achieved that the ``point'' $i^+$ has been blown up to the corridor $(r_\ev, r_\co)_r \times S^2_\ws$ which is represented  by the lower part of the boundary at the coordinate location $\tau=0$ of $\tau = e^{-t_*}$ of a rectangle $(t_0, \infty)_{t_*} \times (r_\ev, r_\co)_r \times S^2_\ws$ representing the region $\rI$ near $i^+$. A similar procedure can be carried through with the regions $\rII$, $\rIII$, $\rIV$; see sec.~2.1 of  \cite{HintzVasy}. The 
blow up of $i^+$ is shown in fig.~\ref{fig:BlowUp}. The vertical lines in the domain $\cR$ represent (parts of) $\CH^R$, $\EH^R$, $\dSH^L$, respectively, and the 
horizontal boundary at $\tau=0$ represents a blow up of $i^+$.\footnote{In \cite{HintzVasy}, also a further artificial horizon beyond the Cauchy horizon is introduced. We refer to \cite{HintzVasy} for details.}

As discussed in the Introduction, only parts of the partially extended RNdS would be present in a more realistic spacetime representing a collapsing spherical object. However, even in the presence of a collapsing star, a portion of 
$\CH^R$ near $i^+$ is always part of the spacetime, and our arguments would apply to that part. To make our analysis simple, we 
will ignore the collapsing star, which is inessential to this problem, and work with the portions $\rI$, $\rII$, $\rIII$, $\rIV$ of RNdS spacetime as described above.

\begin{figure}
\centering
\begin{tikzpicture}[scale=1]
 \draw[dashed,thick,fill=gray!30] (-2.5,0) -- (-2.5,-2) .. controls (-1,-2.3) and (1,-1.8) .. (2.5,-2.2) -- (2.5,0) -- (-2.5,0);
 \draw (-3,0) -- (3,0) node[midway,above]{$i^+$};
 \draw (-2,0) -- (-2,-3) node[below]{$\CH^R$};
 \draw (0,0) -- (0,-3) node[below]{$\EH^R$};
 \draw (2,0) -- (2,-3) node[below]{$\dSH^L$};
\end{tikzpicture}
\caption{Sketch of the ``blow-up'' of $i^+$ and of a sequence of Cauchy surfaces as described in the text. The shaded region is called $\cR$.}
\label{fig:BlowUp}
\end{figure}
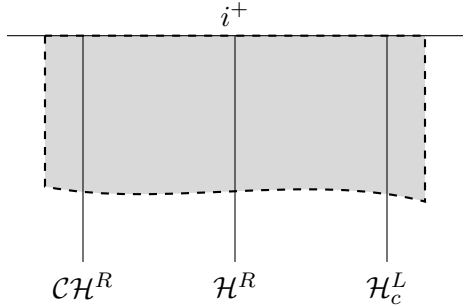

\section{Quantum stress-energy in 2 dimensional RNdS spacetime}
\label{sec:2d}

In this section, we consider a massless, real scalar field $\Phi$ satisfying 
\beq
\label{eq:kg2d}
\square \Phi=0
\eeq
on the 2 dimensional spacetime obtained from \eqref{eq:metric} by suppressing the angular coordinates, i.e.,
\beq
\label{eq:metric2d}
 g = - f(r) \ud t^2 + f(r)^{-1} \ud r^2 \, ,
\eeq
with $f$ given by eq.~(\ref{frnds}).

\subsection{Behavior of the classical and quantum stress-energy tensor near $\CH^R$}

We first consider the behavior of the stress-energy component $T_{VV}$ near $\CH^R$ for a {\em classical} field \eqref{eq:kg2d}. This can be determined immediately from the facts that (i) the stress-energy tensor is conserved, $\nabla^\nu T_{\mu \nu} = 0$, and (ii) the classical stress-energy tensor is traceless, $T^\mu_{\ \mu} = 0$. Thus, in the null coordinates $u,v$ of \eqref{eq:def_u_v}, we have $T_{uv} = 0$. The $\mu = v$ component of the continuity equation $\nabla^\nu T_{\mu \nu} = 0$ then yields
\beq
 \del_u T_{vv} = - f \del_v ( f^{-1} T_{uv} ) = 0.
\eeq
Thus, $T_{vv}$ is constant along any ``left moving'' null geodesic, such as the blue line in Figure~\ref{fig:IntegrationPath}, i.e., we have $T_{vv} (U,v) = T_{vv} (U_0,v)$. We now take the limit as $v \to \infty$ with $U_0 < 0$ and $U > 0$, so that $(U,v)$ approaches $\CH^R$ and $(U_0, v)$ approaches $\dSH^L$ (see Figure~\ref{fig:IntegrationPath}). We transform to regular coordinates $V=V_\ca$ and $V=V_\co$ near $\CH^R$ and $\dSH^L$, respectively, (see eqs.~\eqref{eq:vco} and \eqref{eq:V_-}). We obtain
\beq
\label{eq:T_V_-_V_-_2d_classical}
T_{V_\ca V_\ca}  = \frac{\kappa_\co^2}{\kappa_\ca^2} ( - V_\ca )^{2 \frac{\kappa_\co}{\kappa_\ca} - 2} \, T_{V_\co V_\co} .
\eeq
If the classical solution is smooth on $\dSH^L$, then \eqref{eq:T_V_-_V_-_2d_classical} shows that if $T_{V_\co V_\co} \neq 0$ on $\dSH^L$ and if $\kappa_\co < \kappa_\ca$, then  $T_{V_\ca V_\ca}$ blows up on $\CH^R$. However, in all cases, if $T_{V_\co V_\co}$ has a finite limit on $\dSH^L$, then $V^2_\ca T_{V_\ca V_\ca} \to 0$ on $\CH^R$.

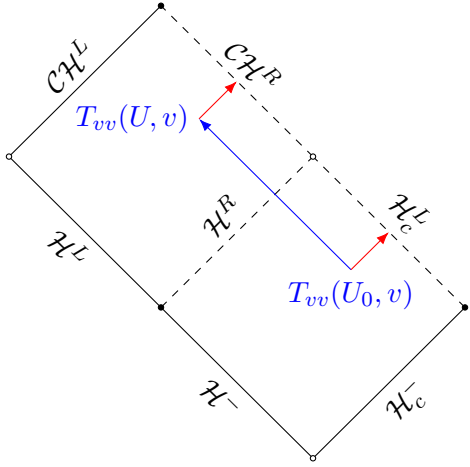
\begin{figure}
\centering
\begin{tikzpicture}[scale=1.0]
 \draw (0,0) -- (-2,2) node[midway, below, sloped]{$\EH^-$};
 \draw (-2,2) -- (-4,4) node[midway, below, sloped]{$\EH^L$};
 \draw (0,0) -- (2,2) node[midway, below, sloped]{$\dSH^-$};
 \draw[dashed] (2,2) -- (0,4) node[midway, above, sloped]{$\dSH^L$};
 \draw[dashed] (-2,2) -- (0,4) node[midway, above, sloped]{$\EH^R$};
 \draw[dashed] (0,4) -- (-2,6) node[midway, above, sloped]{$\CH^R$};
 \draw (-4,4) -- (-2,6) node[midway, above, sloped]{$\CH^L$};
 \draw[blue,-Latex] (0.5,2.5) node[below]{$T_{vv}(U_0,v)$} -- (-1.5,4.5) node[left]{$T_{vv}(U,v)$};
 \draw[red,-Latex] (0.5,2.5) -- (1,3);
 \draw[red,-Latex] (-1.5,4.5) -- (-1,5);
 \draw[fill=white] (-4,4) circle (1pt);
 \draw[fill=black] (-2,6) circle (1pt);
 \draw[fill=white] (0,0) circle (1pt);
 \draw[fill=black] (-2,2) circle (1pt);
 \draw[fill=white] (0,4) circle (1pt);
 \draw[fill=black] (2,2) circle (1pt);
\end{tikzpicture}
\caption{The blue arrow indicates the integration path and the red arrows the limit $v \to \infty$ in the discussion of the 2d case.}
\label{fig:IntegrationPath}
\end{figure}

We now consider the corresponding behavior of the quantum stress-energy tensor near $\CH^R$. The only difference in our analysis from the classical case is that,
as is well-known (see e.g.\ \cite{ChristensenFulling77} for a thorough discussion) in any conformally invariant field theory in 2 dimensions, we have a trace anomaly of the form
\beq
 g^{\mu \nu} T_{\mu \nu} = a R,
\eeq
where $a=1/24\pi$ for the case of a massless scalar field. For the metric \eqref{eq:metric2d}, we have $R = - f''$, where the prime denotes derivative \wrt $r$, so
\beq
\label{eq:TraceAnomaly2d}
 T_{u v} = \frac{a}{4} f f''.
\eeq
Therefore, the $\mu = v$ component of the continuity equation now reads
\beq
\label{eq:ContinuityEq_2d}
 \del_u T_{vv} = - f \del_v ( f^{-1} T_{uv} ) = - \frac{a}{8} f^2 f'''.
\eeq
This can be integrated to yield
\beq
 T_{vv}(U, v) = \frac{a}{8} \left[ 2 f f'' - (f')^2 \right]_{r(U_0, v)}^{r(U, v)} + T_{vv}(U_0, v).
\eeq
We choose $U_0 < 0$ and $U > 0$ and take the limit as $v \to \infty$. Since
\beq
 \lim_{v \to \infty} r(U, v) = \begin{cases} r_\co & U < 0, \\ r_\ca & U > 0. \end{cases}
\eeq
we obtain
\beq
\label{eq:Tvv_ta_2d}
 \lim_{v \to \infty} T_{vv}(U, v) = \lim_{v \to \infty} T_{vv}(U_0, v) - \frac{a}{2} ( \kappa_\ca^2 - \kappa_\co^2 ) . 
\eeq
This differs from the classical case by the step of $- a( \kappa_\ca^2 - \kappa_\co^2 )/2$ that occurs at $U=0$. 

When we transform \eqref{eq:Tvv_ta_2d} to regular coordinates $V=V_\ca$ and $V=V_\co$ near $\CH^R$ and $\dSH^L$, respectively, the first term will give rise to the behavior \eqref{eq:T_V_-_V_-_2d_classical} found in the classical case. However, the second term will give rise to singular behavior of the form
\beq
\label{eq:TVV_ta_2d}
T_{V_\ca V_\ca}  \sim \frac{a}{2} \frac{\kappa_\co^2 - \kappa_\ca^2}{\kappa_\ca^2} \frac{1}{V_-^2}.
\eeq
Unless $\kappa_\co = \kappa_\ca$, this behavior is more singular than the classical behavior. Furthermore, in the quantum case, the $V_\ca^{-2}$ divergence of $T_{V_\ca V_\ca}$ at the Cauchy horizon $\CH^R$ is present whenever $\kappa_\co \neq \kappa_\ca$ even if $T_{V_\co V_\co}$ vanishes on the cosmological horizon $\dSH^L$. 

Similar arguments for the RN case were previously given in \cite{BirrellDavies78}. Apart from not assuming stationarity, the main difference of our argument is that we impose the ``initial condition'' of nonsingularity at $\dSH^L$ instead of the event horizon, as in \cite{BirrellDavies78}, which seems physically much better motivated.

\subsection{The Unruh state in 2 dimensional RNdS}

The main shortcoming of the analysis of the previous subsection is that it does not generalize to higher dimensions, as the trace anomaly does not give sufficient information to determine $T_{vv}$. To make an argument, one has to invoke the -- very plausible -- absence of numerical coincidence of a cancellation between a geometric (from the trace anomaly) and a state-dependent (from the tangential pressure) contribution to the integral of $\del_u T_{vv}$ \cite{BirrellDavies78}.\footnote{Another argument for a divergence at the Cauchy horizon that does not invoke absence of numerical coincidences was put forward in \cite{Hiscock80}: Under the assumptions of spherical symmetry, stationarity, and a non-vanishing total flux of radiation, the stress tensor diverges either at $\CH^R$ or $\CH^L$. However, as discussed in the Introduction, a divergence at $\CH^L$ is irrelevant in general.}

In order to make contact with the analysis that will be given in the following sections for the 4 dimensional case, it is useful to analyze the properties of the Unruh state, defined in the usual way (see e.g.\ \cite{WaldQFTCS} or \eqref{eq:statedef} below) by the positive frequency mode solutions $\square \psi_k^{\rU, \rin/\rup}=0$  
given by
\begin{align}
\label{eq:MU_modes}
 \psi^{\rU, \rin}_{k} & = \frac{1}{\sqrt{2 \pi}} \frac{1}{\sqrt{2 k}} e^{- i k V_\co}, &
 \psi^{\rU, \rup}_{k} & = \frac{1}{\sqrt{2 \pi}} \frac{1}{\sqrt{2 k}} e^{- i k U}, 
\end{align}
We write the metric in the form
\beq
g = G \ud U \ud V_\co
\eeq
The expectation value of the stress tensor in the state defined by these modes is given by \cite{DaviesFullingUnruh}
\begin{align}
\label{eq:T_UU_T_VV_general}
\langle T_{V_\co V_\co} \rangle_\rU & = -\frac{1}{12 \pi} G^{\frac{1}{2}} \del^2_{V_\co} G^{- \frac{1}{2}}, &
\langle T_{U U} \rangle_\rU & = - \frac{1}{12 \pi} G^{\frac{1}{2}} \del^2_{U} G^{- \frac{1}{2}}.
\end{align}
One easily verifies that $\langle T_{V_\co V_\co} \rangle_\rU$ is finite at the cosmological horizon $\dSH^L$ and that near the Cauchy horizon $\CH^R$
\beq
\label{eq:T_V_-_V_-_2d}
\langle T_{V_\ca V_\ca} \rangle_\rU = - \frac{1}{48 \pi} \left( 1 - \frac{\kappa_\co^2}{\kappa_\ca^2} \right) V^{-2}_\ca,
\eeq 
consistent with \eqref{eq:TVV_ta_2d} and the conformal anomaly coefficient $a = \frac{1}{24 \pi}$ of the scalar field. 

We remark that in the stationary state defined by the modes \eqref{eq:MU_modes}, also $\langle T_{UU} \rangle_\rU$ diverges as $U^{-2}$ at $\CH^L$ in adapted coordinates $U=U_\ca$. However, following \cite{MarkovicUnruh}, it is not difficult to see that it is possible to choose a nonstationary state that is regular on $\CH^L$. To do so, 
choose a point $(u_0,v_0)$ in the exterior region of RNdS and define $\tilde U(u)$ and $\tilde V(v)$, respectively, as the affine parameter of ingoing and outgoing null geodesics starting at $(u_0, v_0)$, \cf Figure~\ref{fig:MU}. Defining the state by positive frequency mode solutions $(4 \pi k)^{-\frac{1}{2}} e^{- i k \tilde U}$, $(4 \pi k)^{-\frac{1}{2}} e^{- i k \tilde V}$, one finds that $\langle T_{UU} \rangle$ is finite on $\CH^L$, but $\langle T_{V_\ca V_\ca} \rangle$ is still singular at $\CH^R$. The state is Hadamard also in an extension beyond $\CH^L$. See Figure~\ref{fig:MU} for a sketch of the setup and the domain in which the state is Hadamard. However, since regularity at $\CH^L$ is inessential for us, but stationarity is quite convenient for computational purposes, we shall work with the stationary Unruh state defined by the modes \eqref{eq:MU_modes}.

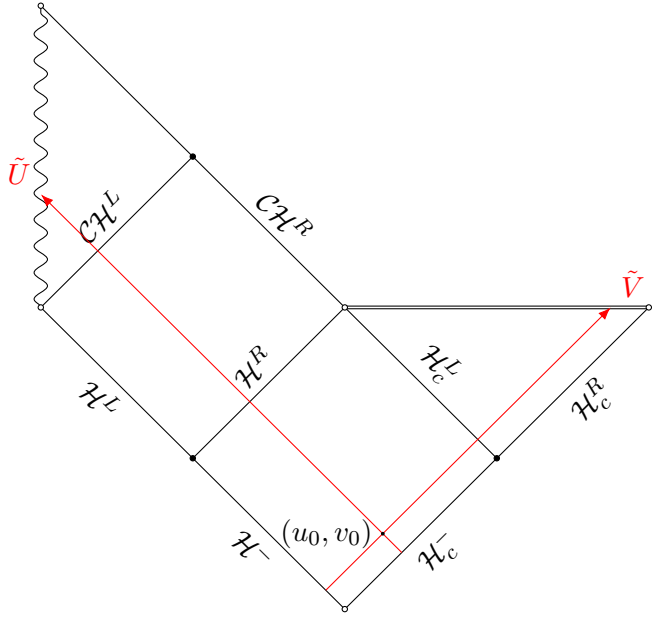
\begin{figure}
\centering
\begin{tikzpicture}[scale=1]
 \draw (0,0) -- (-2,2) node[midway, below, sloped]{$\EH^-$};
 \draw (-2,2) -- (-4,4) node[midway, below, sloped]{$\EH^L$};
 \draw (0,0) -- (2,2) node[midway, below, sloped]{$\dSH^-$};
 \draw (2,2) -- (0,4) node[midway, above, sloped]{$\dSH^L$};
 \draw (-2,2) -- (0,4) node[midway, above, sloped]{$\EH^R$};
 \draw (0,4) -- (-2,6) node[midway, above, sloped]{$\CH^R$};
 \draw (-4,4) -- (-2,6) node[midway, above, sloped]{$\CH^L$};
 \draw[snake it] (-4,4) -- (-4,8);
 \draw (-2,6) -- (-4,8);
 \draw[double] (0,4) -- (4,4);
 \draw (2,2) -- (4,4) node[midway, below, sloped]{$\dSH^R$};
 \draw[red,-Latex] (-0.25,0.25) -- (3.5,4) node[above right]{$\tilde V$}; 
 \draw[red,-Latex] (0.75,0.75) -- (-4,5.5) node[above left]{$\tilde U$}; 
 \draw [fill=black] (0.5,1) circle (0.5pt) node[left]{$(u_0,v_0)$};
 \draw[fill=white] (-4,4) circle (1pt);
 \draw[fill=white] (4,4) circle (1pt);
 \draw[fill=white] (-4,8) circle (1pt);
 \draw[fill=black] (-2,6) circle (1pt);
 \draw[fill=white] (0,0) circle (1pt);
 \draw[fill=black] (-2,2) circle (1pt);
 \draw[fill=white] (0,4) circle (1pt);
 \draw[fill=black] (2,2) circle (1pt);
\end{tikzpicture}
\caption{The domain in which the Markovic-Unruh state is Hadamard, \cf the description in the text.}
\label{fig:MU}
\end{figure}

Note that a single ingoing Unruh-mode $\psi^{\rU, \rin}_k$ has a stress tensor at the 
Cauchy horizon $\CH^R$
\beq
\label{eq:T_VV_Psi_in}
 T_{V_\ca V_\ca}[\psi^{\rU, \rin}_k] = \del_{V_\ca} \psi^{\rU, \rin}_k \del_{V_\ca} \overline{ \psi^{\rU, \rin}_k} = \frac{1}{4 \pi} k \frac{\kappa_\co^2}{\kappa_\ca^2} V_\ca^{2 \frac{\kappa_\co}{\kappa_\ca} -2},
\eeq
with a divergence which, of course, is the same as the one found in \eqref{eq:T_V_-_V_-_2d_classical} for the classical stress tensor and is weaker than that found for the full quantum stress tensor. The divergence \eqref{eq:T_V_-_V_-_2d} of the quantum stress tensor is thus a UV effect, i.e.\ it stems from integration over modes of arbitrarily high frequency. This can be seen more explicitly by---instead of renormalizing the stress tensor via Hadamard point-splitting, i.e.\ using \eqref{eq:T_UU_T_VV_general}---computing the difference of the expectation value of $T_{V_\ca V_\ca}$ evaluated in the state defined by the modes \eqref{eq:MU_modes} and a ``comparison'' state which is Hadamard across the Cauchy horizon $\CH^R$ (so that the latter has $T_{V_\ca V_\ca}$ finite at $\CH^R$). Such a state can be defined by using, instead of the modes $\psi^{\rU, \rin}_k$, the modes $\square \psi^{\rF, \rout}_k=0$, with
\beq
 \psi^{\rF, \rout}_k = \frac{1}{\sqrt{2 \pi}} \frac{1}{\sqrt{2k}} e^{- i k V}
\eeq
where now $V=V_\ca$. 

This difference in expectation values can be easily computed directly, with a point-split procedure. Having in mind our later discussion of the 4 dimensional case, we perform this calculation differently. Restricting to the in-modes, we consider the ``Boulware'' mode solutions (just as a computational device, not changing the definition of the state!)
\begin{subequations}
\label{eq:LeftMover_2d}
\begin{align}
\label{eq:LeftMoverI_2d}
 \psi^{\rin, \rI}_{\omega}(V_\co) & = \begin{cases} (2 \pi)^{-\frac{1}{2}} (2\omega)^{-\frac{1}{2}} e^{- i \omega v} & V_\co < 0 \\ 0 & V_\co > 0 \end{cases}, \\
\label{eq:LeftMoverII_2d}
 \psi^{\rin, \rIII}_{\omega}(V_\co) & = \begin{cases} 0 & V_\co < 0 \\  (2 \pi)^{-\frac{1}{2}} (2\omega)^{-\frac{1}{2}} e^{- i \omega v} & V_\co > 0 \end{cases},
\end{align}
\end{subequations}
where for $V_\co > 0$, we defined
\beq
 v = \kappa_\co^{-1} \log V_\co.
\eeq
Here $\rI$ stands for the original exterior region and $\rIII$ for the region beyond the cosmological horizon $\dSH^L$, where $V_\co > 0$, \cf fig.~\ref{fig:1}.
As usual \cite{ChristensenFulling77}, the formal mode integral over the Unruh modes $\psi^{\rU, \rin}_k$ can also be expressed as
\beq
 \langle T_{v v} \rangle_{\rU} = \int \ud \omega \ \coth (\pi \omega \kappa_c^{-1}) T_{v v}[\psi^{\rin, \rI}_\omega]
\eeq
with
\beq
\label{eq:T_vv_phi_I}
 T_{v v}[\psi^{\rin, \rI}_{\omega}] = \del_v \psi^{\rin, \rI}_{\omega} \del_v \overline{\psi^{\rin, \rI}_{\omega}} = \frac{1}{4 \pi} \omega 
\eeq
the evaluation of $T_{vv}$ in the solution $\psi^{\rin, \rI}_{\omega}$. For the difference between the expectation value of $T_{vv}$ evaluated in the state defined by the modes $\psi^{\rU, \rin}_k$ and the modes $\psi^{\rF, \rout}_k$ defining the comparison state $\langle \ . \ \rangle_\rF$ (the latter being Hadamard across $\CH^R$), one thus obtains
\beq
\label{eq:ModeSum2dBoulware}
 \langle T_{vv} \rangle_{\rU} - \langle T_{vv} \rangle_{\rF} = \frac{1}{4 \pi} \int_0^\infty \ud \omega \ \omega \left( \coth \frac{\pi \omega}{\kappa_\co} - \coth \frac{\pi \omega}{\kappa_\ca} \right) = \frac{1}{48 \pi} (\kappa_\co^2 - \kappa_\ca^2),
\eeq
which coincides with the result  \eqref{eq:Tvv_ta_2d}, for the anomaly coefficient $a = (24 \pi)^{-1}$ of the scalar field. We note that in terms of the modes $\psi^{\rin, \rI}_{\omega}$, the stress tensor for a single mode \emph{does} give the correct prediction for the degree of singularity of the stress tensor at the Cauchy horizon, \cf \eqref{eq:T_vv_phi_I}. There is no contradiction to the previous discussion in terms of the $\psi^{\rU, \rin}_k$ modes, as a single $\psi^{\rin, \rI}_{\omega}$ mode contains arbitrarily high frequencies \wrt $V_\co$. The point is that for a single $\psi^{\rin, \rI}_{\omega}$ mode the stress tensor is divergent both at the cosmological and the Cauchy horizon, but at the cosmological horizon the Boltzmann weights of these modes are exactly such that their contributions cancel upon renormalization. At the Cauchy horizon, on the other hand, the Boltzmann weights are then incommensurate, unless $\kappa_\co = \kappa_\ca$, so that a finite value of $T_{vv}$ (and, thus, a singular stress-energy tensor) at the Cauchy horizon is obtained.

\section{Construction of the Unruh state in 4 dimensional RNdS spacetime}
\label{sec:3}

The analysis just given for the 2 dimensional case is not adequate to treat the 4 dimensional case because (a) the trace anomaly and 
continuity equation do not give enough information to determine the behavior of the stress energy tensor near the Cauchy horizon and (b) in the 2 dimensional case the field is an exact linear superposition of left- and right-movers, whereas the ``backscattering of modes'' occurs in 4 dimensions. In order to make our arguments in 4 dimensions with the necessary clarity and precision, we introduce in this section some notions and results from QFT in curved spacetime. The main result of this section is the construction of the Unruh state as a stationary Hadamard state in regions I, II, and III of fig.~\ref{fig:1}.
The reader not interested in the general discussion of QFT in curved spacetime or in the details of the Unruh state construction may skip to sec.~\ref{sec:5}.

\subsection{Quantum fields and Hadamard states}
\label{sec:4.1}

Usually, one thinks of observables, such as the (smeared) stress tensor, as operators on a Hilbert space.
However, properly speaking, it is problematical to fix from the outset a fixed Hilbert space representation in which the states live as vectors or 
statistical operators (density matrices). 
This is because (a) frequently states must be considered that are not expressible as density matrices on the fixed Hilbert space, and (b)
the usual Hilbert space quantization is tied closely to the existence of a Cauchy surface (to define canonical fields), which is not available in 
our situation. 

Therefore, it is formally much more convenient -- and essentially requires no work (!) -- to set up the quantum theory in an algebraic manner for the purpose of the discussion. 
Recall that a spacetime $(\sN, g)$ is called ``globally hyperbolic''
if it has a Cauchy surface $\Sigma$, i.e.\ a spacelike or null codimension-one surface such that every inextendible causal curve hits 
$\Sigma$ precisely once, see \cite{WaldGR} for details. In such a situation, the Klein-Gordon wave operator $\square-\mu^2$ possesses a unique 
retarded and advanced propagator, called $E^\pm$. We can think about the propagators as integral operators 
$E^\pm: C^\infty_0(\sN) \to C^\infty(\sN)$, which are uniquely determined by the properties 
\beq
(\square-\mu^2) E^\pm = 1, \quad {\rm supp}(E^\pm f) = J^\pm({\rm supp } f), 
\eeq
where ``supp'' is the support of a function $f$ (the closure of the set where it is not zero), and where $J^\pm(\cO)$ are the causal future/past of a set $\cO$, see e.g.\ 
\cite{WaldGR} for the detailed definition of these notions. 
The existence and uniqueness property of $E^\pm$ express the fact that the initial value problem 
for $(\square -\mu^2)\Phi=0$ has a unique solution, i.e.\ they express determinism. 
The ``commutator function'' is defined as $E=E^+-E^-$, and it is often identified with a distributional kernel $E(x_1, x_2)$ on $\sN \times \sN$
(this identification implicitly depends on the choice $\ud \eta$ of the integration element, which we take to be that given by the metric $g$). 
Associated with $(\sN,g)$, we then 
define an algebra, $\gA(\sN, g)$, whose generators are ``smeared'' field observables $\Phi(f)$, their formal adjoints\footnote{Here, the * is an anti-linear operation satisfying the usual rules for the adjoint, namely $A^* B^*=(BA)^*$.} $\Phi(f)^*$ and an identity $1$, where $f$ is a complex-valued $C^\infty$-function on $\sN$
with compact support, and whose relations are:

\begin{enumerate}
\item[A1)] For any $f_1, f_2$ and any complex numbers $c_1, c_2$, we have 
$\Phi(c_1 f_1 + c_2 f_2) =  c_1 \Phi(f_1) + c_2 \Phi(f_2)$. 

\item[A2)]  For any $f$ we have $\Phi[(\square -\mu^2)f]=0$.

\item[A3)] For any $f$ we have $\Phi(f)^*= \Phi(\bar f)$.

\item[A4)] For any $f_1,f_2$, we have $\Phi(f_1)\Phi(f_2)-\Phi(f_2)\Phi(f_1) = iE(f_1,f_2)1$.
\end{enumerate}  

Item A1) essentially says that $\Phi(f)$ is an operator-valued distribution. By analogy with the conventions in distribution theory, we 
write 
\beq
\Phi(f) = \int_\sN \Phi(x) f(x) \ud \eta(x)
\eeq
 as if $\Phi(x)$ was a function, 
and work from now with the formal point-like object $\Phi(x)$. Item A2) says that 
$(\square -\mu^2)\Phi(x)=0$ in the sense of distributions and item A3) that $\Phi(x)=\Phi(x)^*$ in the sense of distributions, i.e.\ we have a hermitian field. 
Item A4) gives the usual commutation relation in covariant form $[\Phi(x_1),\Phi(x_2)] = iE(x_1,x_2)1$, with $E(x_1,x_2)$ the distributional kernel of $E$.

In RNdS, see fig.~\ref{fig:1}, the regions I, II, III separately are globally hyperbolic, as is their union. However, region IV is not globally hyperbolic. Since we somehow want to talk about the field in this region, we must find a way to generalize this construction. This can be done in the following cheap way. Let $\sM$ be a non globally hyperbolic spacetime such as the the union $\rI \cup \rII \cup \rIII \cup \rIV$ depicted in fig.~\ref{fig:1}. We consider all globally hyperbolic subregions $\sN \subset \sM$ whose causal structure coincides with that of $\sM$, i.e.\ $J^\pm_{\sN}(p) = J^\pm_{\sM}(p) \cap \sN$ for all $p \in \sN$, where $J^\pm_{\sN}(p)$ denotes the causal future/past of $p$ in $\sN$.\footnote{The issue that is circumvented with this restriction was pointed out in \cite{Kay1992}.} For two such globally hyperbolic subregions $\sN_1$ and $\sN_2$ such that $\sN_1 \subset \sN_2$, one then has $\gA(\sN_1) \subset \gA(\sN_2)$. In a first step, let us define a (too large)
algebra $\bigvee_{\sN \subset \sM} \gA(\sN,g)$ {\em freely}
generated by the algebras $\gA(\sN,g)$ of such 
globally hyperbolic subregions of $\sM$ (i.e.\ no relations imposed between generators of different $\gA(\sN_i,g)$). 
Then we divide out the relation $\Phi_{\sN_1}(f) = \Phi_{\sN_2}(f)$ for all $f$, $\sN_1$, $\sN_2$ such that $\supp f \subset \sN_1, \sN_2$,
This construction
gives a satisfactory way of defining the observables in a spacetime like RNdS which is not globally hyperbolic, by imposing all relations that can be obtained in globally hyperbolic subregions.\footnote{Imposing ``boundary conditions'' at the singularity in region $\rIV$ would in effect render the equation hyperbolic on all of $\sM$. In this case an algebra $\gA(\sM)$ could be defined directly. It would correspond to a quotient of the algebra defined previously by additional relations.} 

Of course, to do physics, we not only need an algebra but also states. A state in this framework is simply a positive, normalized, linear
functional on $\gA(\sM,g)$, where ``positive'' means a functional $\gA(\sM,g) \owns A \mapsto \langle A \rangle_\Psi \in \C$ such that $\langle 
A^* A \rangle_\Psi \ge 0$ for all $A$, and where ``normalized'' means $\langle 1 \rangle_\Psi = 1$. Since $\gA(\sM,g)$ is presented in 
terms of generators and relations, a state is given once we know its correlation functions $\langle \Phi(x_1) \cdots \Phi(x_n) \rangle_\Psi$. 
Among all states, we will focus on ``Gaussian states'', which are determined uniquely in terms of their 2-point correlation (``Wightman''-)
function $\langle \Phi(x_1) \Phi(x_2) \rangle_\Psi$, see e.g.\ \cite{HollandsWaldWick} for details. From the definitions, we must have
\begin{enumerate}
\item[S1)] 
(Commutator) $\langle \Phi(x_1) \Phi(x_2) \rangle_\Psi - \langle \Phi(x_2) \Phi(x_1) \rangle_\Psi = iE(x_1, x_2)$ whenever $x_1,x_2$ are contained in some globally hyperbolic portion $\sN \subset \sM$.

\item[S2)] 
(Wave equation) $(\square-\mu^2)_{x_1} \langle \Phi(x_1) \Phi(x_2) \rangle_\Psi = (\square-\mu^2)_{x_2} \langle \Phi(x_1) \Phi(x_2) \rangle_\Psi = 0$.

\item[S3)] 
(Positive type) 
$\int \langle \Phi(x_1) \Phi(x_2) \rangle_\Psi f(x_1) \overline f(x_2) \ud \eta(x_1) \ud \eta(x_2) \ge 0$ for any smooth, compactly supported 
function $f$ on $\sM$. 
\end{enumerate}

Conversely, any distribution $\langle \Phi(x_1) \Phi(x_2) \rangle_\Psi$ with these properties defines a Gaussian state. One may always 
represent the field algebra $\gA(\sM,g)$ on some Hilbert space such that the expectation functional $\langle \ . \ \rangle_\Psi$ corresponds to 
some ``vacuum'' vector in that Hilbert space, although for a generic state the terminology ``vacuum'' has no physical meaning. At any rate, it will be fully sufficient for this paper to work with the expectation functionals. 

While the above three conditions characterize states in general, to obtain physically reasonable states (in a sense explained below), one should impose more stringent conditions on the short-distance behavior of the 2-point function $\langle \Phi(x_1) \Phi(x_2) \rangle_\Psi$. 
One such condition is the ``Hadamard condition''. To state this condition, one introduces the notion of a 
convex normal neighborhood, $\cO$, in the total spacetime $\sM$, which is a globally hyperbolic 
sub-spacetime $\cO \subset \sM$ such that any pair of points $x_1,x_2$ from $\cO$ can be connected by a unique geodesic. 
In such a neighborhood, we can define uniquely the signed squared geodesic distance $\sigma(x_1,x_2)$ for $x_1,x_2$ from $\cO$, 
and we can, non-uniquely, define a time function $T(x)$. E.g., in RNdS, $T(x)=t$ in region I, or $T(x)=r$ in region II.

\begin{definition} (Hadamard condition, see, e.g.\ \cite{KayWald91})
For $(x_1, x_2) \in \cO \times \cO$, where $\cO$ is any convex normal globally hyperbolic neighborhood, the 2-point function has the general form
\beq
\label{Hadamard}
\langle \Phi(x_1) \Phi(x_2) \rangle_\Psi = \frac{1}{4\pi^2} \left(
\frac{U(x_1,x_2)}{\sigma+i0T} + V(x_1, x_2) \log (\sigma+i0T) + W_\Psi(x_1,x_2)
\right) .
\eeq
Here $U=\Delta^{1/2}$ 
is the square root of the VanVleck determinant \cite{Poisson03}, $T=T(x_1)-T(x_2)$,
$V$ is determined by the Hadamard-deWitt transport equations \cite{DeWittBrehme} and is thereby, as $U,\sigma$, entirely determined by the local geometry
within $\cO$. The state dependence is contained in $W_\Psi$, which is required to be a $C^\infty$ function on $\cO \times \cO$. 
\end{definition}

Roughly speaking, the Hadamard condition states that the singular part of the 2-point function $\langle \Phi(x_1) \Phi(x_2) \rangle_\Psi$
is entirely determined by the local geometry. In any Hadamard state, one can define the $n$-point correlation functions of ``renormalized composite fields''. Such fields are classically given by polynomials of the field $\Phi$ and its covariant derivatives $\nabla_{\mu_1} \dots
\nabla_{\mu_k} \Phi$ such as the stress tensor\footnote{There is an ambiguity in the definition of the classical stress tensor, as one could interpret the mass term also as a curvature coupling.} 
\beq
T_{\mu\nu} = \nabla_\mu \Phi \nabla_\nu \Phi - \tfrac{1}{2} g_{\mu\nu} (\nabla^\sigma \Phi \nabla_\sigma \Phi + \mu^2 \Phi^2) .
\eeq
At the quantum level, these are defined (non-uniquely), as operator valued distributions in some larger algebra containing $\gA(\sM,g)$
\cite{HollandsWaldWick,HollandsWaldTO}. For us, it is only important how their 1-point function is defined. This is easiest to explain if there are no derivatives and if the polynomial is quadratic, i.e.\ for $\Phi^2$. Then one defines
\beq
\label{eq:Phi2_LocalRenormalization}
\langle \Phi^2(x) \rangle_\Psi = \lim_{x_1,x_2\to x} W_\Psi(x_1, x_2).
\eeq
As the singular part of \eqref{Hadamard} -- which is effectively subtracted -- is covariant, this ``{\it point-split renormalization} prescription''\footnote{The ``point-split'' prescription is 
equivalent to saying that $\Phi^2$ is the first non-trivial, i.e.\ aside from the identity, operator in the OPE, see \cite{HollandsWaldAxiomatic}.} is covariant -- in a sense it is the ``same'' on all spacetimes, see 
\cite{HollandsWaldWick, BrunettiFredenhagenVerch} for the precise meaning of this statement.
The 1-point functions $\langle \nabla_{\mu_1} \dots
\nabla_{\mu_k} \Phi \nabla_{\nu_1} \dots
\nabla_{\nu_l} \Phi (x) \rangle_\Psi$ would be defined in the same way, except that we take the derivatives before the coincidence limit. 
There are certain ambiguities in the definition of the composite quantum fields $\nabla_{\mu_1} \dots
\nabla_{\mu_k} \Phi \nabla_{\nu_1} \dots \nabla_{\nu_l} \Phi (x)$, in the sense that the above point-split scheme is not the only one which is ``the same on all spacetimes''. 
These ambiguities are fully characterized in a mathematically precise manner in \cite{HollandsWaldWick}. They
would correspond in the case of $\Phi^2$ to adding a constant multiple of the scalar curvature and the mass, $\Phi^2 \to \Phi^2 + c_1 R + c_2 \mu^2$, and in the case of 
more general quadratic composite fields, to adding linear combinations of higher curvature invariants of the correct ``dimension''. In the case of 
$T_{\mu\nu}$ this freedom must be used to ensure that $\nabla^\mu T_{\mu\nu}=0$. These ambiguities, while important in general, are however of no consequence 
for our investigation here, since they only give corrections to the expectation values that are smooth functions on spacetime, whereas we are interested 
in their (potential) divergent part as we approach a Cauchy horizon.

For us, the following properties of Hadamard states will be relevant:

\begin{enumerate}
\item[H1)] For any Hadamard state, $\langle \nabla_{\mu_1} \dots
\nabla_{\mu_k} \Phi \nabla_{\nu_1} \dots
\nabla_{\nu_l} \Phi (x) \rangle_\Psi$ is a $C^\infty$ function of $x$ on $\sM$.

\item[H2)] For any pair $\Psi, \Psi'$ of Hadamard states, $\langle \Phi(x_1) \Phi(x_2) \rangle_\Psi -
\langle \Phi(x_1) \Phi(x_2) \rangle_{\Psi'}$ is a $C^\infty$ function on $\sN \times \sN$ on any 
globally hyperbolic portion $\sN$ of $\sM$.

\item[H3)] For any pair $\Psi, \Psi'$ of Hadamard states and $x \in \sM$, we have
\beq\label{eq:ptsplit}
\langle \Phi^2(x) \rangle_\Psi - \langle \Phi^2(x) \rangle_{\Psi'}= \lim_{x_1,x_2\to x} ( \langle \Phi(x_1) \Phi(x_2) \rangle_\Psi - \langle \Phi(x_1) \Phi(x_2) \rangle_{\Psi'}), 
\eeq
with a similar formula for general quadratic composite fields $\nabla_{\mu_1} \dots
\nabla_{\mu_k} \Phi \nabla_{\nu_1} \dots \nabla_{\nu_l} \Phi (x)$. 
\end{enumerate}

Properties H1) and H3) are immediate consequences of the definition since $W$ is locally smooth. Property H2) is remarkable, because 
the Hadamard form only refers to an arbitrarily small neighborhood $\cO$, while H2) makes a statement about arbitrarily large globally hyperbolic subsets of $\sM$, such as the union of regions $\rI$, $\rII$, $\rIII$ in RNdS. Its proof combines nontrivially the commutator, field equation, and positive type 
property of Hadamard states \cite{Radzikowski}. It can only be given using ``microlocal'' techniques, which we will briefly describe in the next subsection. States which are not Hadamard typically have  infinite fluctuations, e.g.\ $\langle \Phi^2(f_1) \Phi^2(f_2) \rangle_\Psi=\infty$, of smeared composite fields, even if their expectation value should be finite $\langle \Phi^2(f) \rangle_\Psi < \infty$ (which need not be the case either, see e.g.\ \cite{BrunettiFredenhagenHollands}). Furthermore, interacting quantum field theories (e.g.\ adding a $\Phi^4$-interaction) require Hadamard states \cite{BrunettiFredenhagenScalingDegree,HollandsWaldTO}. Reasonable states in flat spacetime such as the vacuum, thermal, finite particle number, steady states are Hadamard, and if a state is Hadamard near any Cauchy surface $\Sigma$ in a globally hyperbolic spacetime, i.e.\ ``initially'', 
then it remains Hadamard \cite{RadzikowskiVerch}. 

Together, these properties strongly suggest that Hadamard states are the reasonable class to 
consider as ``regular states'' (analogous classically to $C^\infty$-functions) on globally hyperbolic spacetimes. In this paper, we will study in a sense whether a state which is initially Hadamard (in regions $\rI$, $\rII$, $\rIII$)
remains Hadamard across the Cauchy horizon $\CH^R$ of RNdS.

\subsection{Hadamard states from null surfaces}

Having defined Hadamard states, we must ask whether (a) such states exist, at least on globally hyperbolic spacetimes, and (b) how to 
construct/characterize concretely Hadamard states on given spacetime representing given physical setups. (a) has been established by several rigorous methods, see \cite{FNW81,JunkerSchrohe,GerardWrochna,HollandsPhD,DappiaggiMorettiPinamontiBook}. In particular, given any one Hadamard state $\Psi$, we may go to a representation 
of the field algebra by operators on a Hilbert space, in which the state is represented by a vector. Then, applying any product $\Phi(f_1) \cdots \Phi(f_N)$ to this vector, 
we get a new vector giving a new expectation functional. It can be shown that this is again a Hadamard state. Thus, in globally hyperbolic spacetimes, there is an abundance of Hadamard states.
(b) In the present context, it is particularly natural to define particular Hadamard states with a concrete interpretation by ``prescribing positive frequency modes'' on suitable null surfaces. This idea goes back already to the beginnings of the subject \cite{Hawking1975}. The Hadamard property of 
such a state was first established by \cite{HollandsPhD} in the specific context of null-cones in curved space, and later independently by \cite{DappiaggiMorettiPinamonti_Unruh} in the case of the ``Unruh state'' in Schwarzschild. 

\subsubsection{``Local vacuum states'' defined on null-cones.} To motivate the constructions in RNdS, we first consider the case of a massless scalar field on Minkowski spacetime. 
There, we have of course the global vacuum state with the usual Wightman 2-point function $\langle \Phi(x_1) \Phi(x_2) \rangle_0
= \frac{1}{4\pi}\, [(x_1-x_2)^2 + i0(x_1^0-x_2^0)]^{-1}$. If we restrict attention to the causal future $J^+(0)$ of the origin $0$, we can alternatively write 
this state in terms of modes with characteristic initial data on the null boundary, $\dot J^+(0)$, of this cone. For this, we let $\psi_{\omega \ell m}$
be the mode solutions $\square \psi_{\omega \ell m}=0$ such that 
\beq
\psi_{\omega \ell m}(x) =  (2 \pi)^{-\frac{3}{2}} (2 \omega)^{-\frac{1}{2}} Y_{\ell m}(\theta, \phi) r^{-1} e^{- i \omega r}
\quad \text{on $\dot J^+(0)$.}
\eeq 
Here, $r$ is viewed as an affine parameter along the null generators of $\dot J^+(0)$. As shown in \cite{HollandsPhD}, the vacuum 2-point 
function can be written inside the future cone as 
\beq
\label{eq:statedef}
\langle \Phi(x_1) \Phi(x_2) \rangle_0 = \int_0^\infty \sum_{\ell,m} \overline{\psi_{\omega \ell m}(x_1)}\psi_{\omega \ell m}(x_2) \ud \omega. 
\eeq
This formula can be rewritten as follows. Let $\Phi(f)=\int \Phi(x) f(x) \ud \eta(x)$ be the smeared field, 
and let $f_1,f_2$ be $C^\infty$-functions supported inside the future light cone $J^+(0)$. Let $E=E^+-E^-$ be the commutator function (retarded minus advanced fundamental solution) and let $F_1 = E^-f_1 |_{\dot J^+(0)}$, with a similar formula for $F_2$. 
Then we can write alternatively, with $\Omega=(\theta,\phi)$ and $\ud^2 \Omega = \ud \theta^2 + \sin^2 \theta \ud \phi^2$,
\beq
\label{eq:Adef}
\langle \Phi(f_1) \Phi(f_2) \rangle_0 = -\frac{1}{\pi} \int_0^\infty \! \! \int_0^\infty \! \! \int_{S^2} 
\frac{r_1 r_2 F_1(r_1, \Omega) F_2(r_2, \Omega)}{(r_1-r_2-i0)^2}  \ud^2 \Omega \ud r_1 \ud r_2. 
\eeq
As argued in \cite{HollandsPhD}, and later independently in \cite{DappiaggiPinamontiPorrmann}, one can generalize these constructions to locally define states associated with lightcones in 
curved spacetime $(\sM,g)$ for a conformally coupled field $(\square-\tfrac{1}{6} R)\Phi=0$,
in the following manner. 

First, we pick a reference point $p \in \sM$, and we let $J^+(p)$ be its causal future. The boundary $\dot J^+(p)$ will in general 
not have the structure of an embedded hypersurface far away from $p$ due to caustics, but defining $\dot J^+(p) \cap \cO$, 
with $\cO$ a convex normal neighborhood of $p$ with smooth boundary, $\dot J^+(p) \cap \cO$ is diffeomorphic to a future lightcone in 
Minkowski space cut off by a spacelike plane, \cf Fig.~\ref{fig:LightCone}. We can introduce coordinates $(r, \ws)$ on $\dot J^+(p) \cap \cO$ analogous to the above as follows: Choose a smooth assignment $S^2 \ni \ws \mapsto l(\ws) \in T_p \sM$ such that each null ray in $T_p \sM$ is intersected exactly once by this mapping. Then $(r, \ws) \mapsto \exp_p( r l(\ws))$ is a diffeomorphism between a neighborhood of the origin in $\R_+ \times S^2$ and $\dot J^+(p) \cap \cO$.
Then, we can define\footnote{A minor subtlety 
is to prove that the integrals converge at the tip of the lightcone, $p$, see \cite{HollandsPhD}.} a local vacuum state by 
\beq
\label{eq:LVS}
\langle \Phi(f_1) \Phi(f_2) \rangle_p = -\frac{1}{\pi} \int_0^\infty \! \! \int_0^\infty \! \! \int_{S^2} 
\frac{r_1 r_2 (\Delta^{-\frac{1}{2}} F_1)(r_1,\Omega) (\Delta^{-\frac{1}{2}} F_2)(r_2, \Omega)}{(r_1-r_2-i0)^2}  \ud^2 \Omega \ud r_1 \ud r_2,
\eeq
for testfunctions $f_1, f_2$ supported in $J^+(p) \cap \cO$, where now 
$F_1 = Ef_1 |_{\dot J^+(p) \cap \cO}$, and similarly for $F_2$. The presence of the VanVleck determinant $\Delta(p, \ . \ )$, which is 
equal to $1$ in Minkowski spacetime, 
ensures that the definition is independent of the arbitrary choice of coordinates $(r,\Omega)$ \cite{HollandsPhD}.

\begin{figure}
\centering
\includegraphics[width=0.4\textwidth]{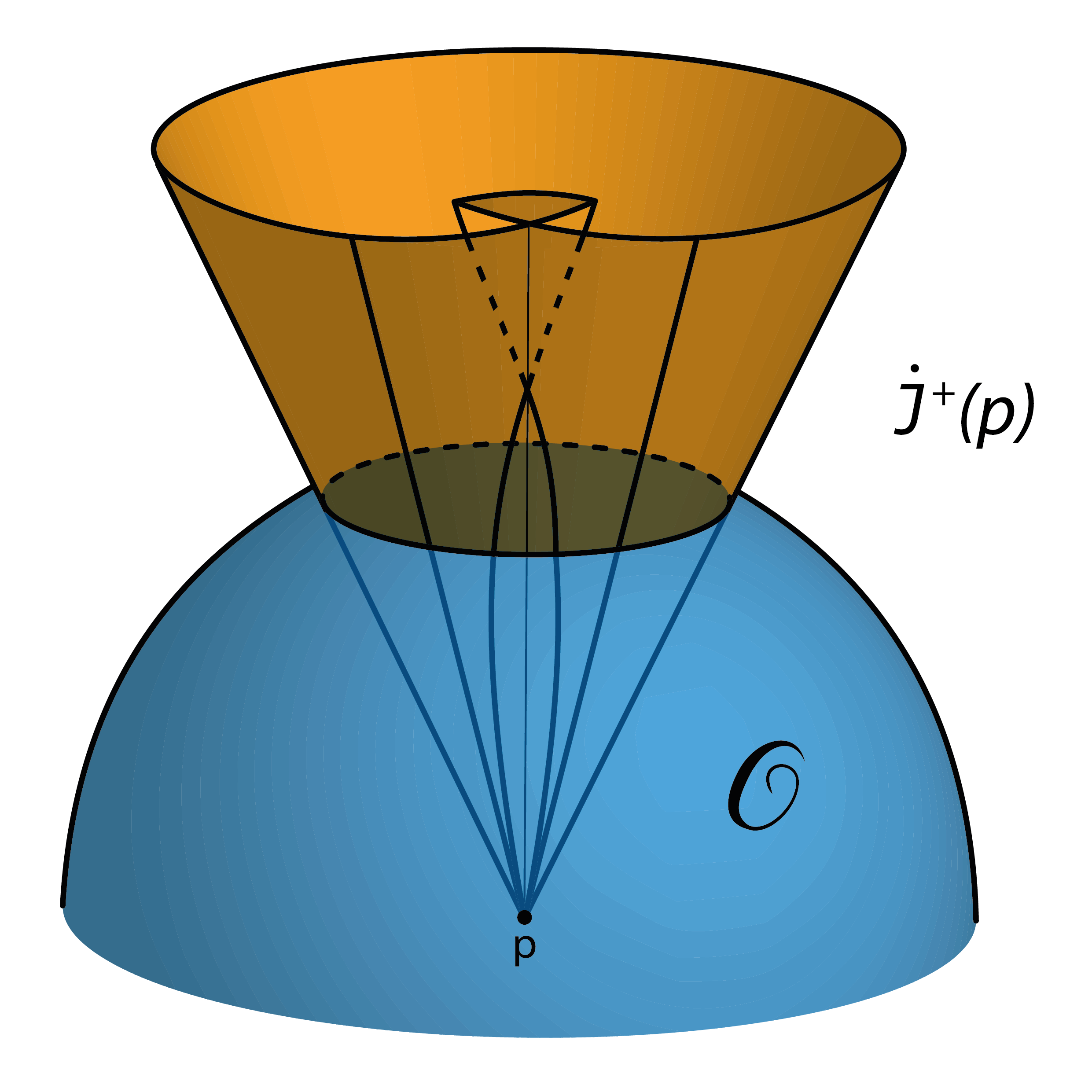}
\caption{Sketch of the of the subset $\dot J^+(p) \cap \cO$ of the light cone on which one can define a state by \eqref{eq:LVS}.}
\label{fig:LightCone}
\end{figure}

By construction, the 2-point function \eqref{eq:LVS} defines a state within $J^+(p) \cap \cO$ because it can be seen to satisfy the commutator, field equation, and 
positivity requirements \cite{HollandsPhD}. Furthermore, it can be viewed as a ``local vacuum state'', in the sense \cite{Kay2001} that its construction only depends on the local geometry within $J^+(p) \cap \cO$, but not on arbitrary choices of coordinates. But is it also Hadamard inside $J^+(p) \cap \cO$?
The answer is yes \cite{HollandsPhD}, but the proof cannot be obtained simply by checking the definition. Instead, one must use the methods of microlocal analysis. Since we will use a similar argument below for RNdS spacetime, we here present the derivation. For details on distribution theory and functional analysis, we refer to \cite{HoermanderI}. 

If $\varphi$ is a smooth function on $\R^n$, then its Fourier transform $\hat \varphi(k)$ decays faster than any inverse power $|k|^{-n}$ as $|k| \to \infty$. 
If $\varphi$ is distributional, we say that a non-zero $k \in \R^n$ is a regular direction at $x$ if there exists a smooth cutoff function $\chi$ of compact 
support such that $\chi(x) \neq 0$ and an open cone $\Gamma$ containing $k$ such that for all $N >0$ there is a constant $C_N$ for which
\beq
|\widehat{\chi \varphi}(p)| \le C_N (1 + | p |)^{-N} \quad \forall p \in \Gamma.
\eeq
The ``wave-front set'', $\WF_x(\varphi)$ of $\varphi$ at point $x$ is, loosely speaking, the set of directions $k$ such that $\chi \varphi(x)$ fails to have a 
rapidly decreasing Fourier transform. More precisely, it is the complement of the regular directions at $x$, with $k=0$ by definition never in $\WF_x(\varphi)$. It can be seen that under a diffeomorphism $\psi$ of a neighborhood of $x$, $\WF_x(\varphi \circ \psi) = \psi^* \WF_{\psi(x)}(\varphi)$. 
Thus, on an $n$-dimensional manifold $\cX$,  $\WF_x(\varphi)$ should be viewed as a dilation invariant subset of $T^*_x \cX \setminus o$, where $o$ denotes the zero section. 
One also sets $\WF(\varphi)= \bigcup_{x \in \cX} \WF_x(\varphi) \subset T^*\cX \setminus o$. In the present context, we take 
$\cX = \sN \times \sN$, where $\sN$ is some globally hyperbolic subset of a spacetime $\sM$.

By an important result of Radzikowski \cite{Radzikowski}, a state is Hadamard if and only if its 2-point correlation function 
$G_\Psi(x_1,x_2) = \langle \Phi(x_1) \Phi(x_2) \rangle_\Psi$ satisfies
\beq
\label{eq:WF_condition}
\WF'(G_\Psi) = \cC^+,
\eeq
where
\beq
 \cC^\pm \defeq \left\{ (x_1, \xi_1, x_2, \xi_2) \in T^*(\sN \times \sN) \setminus o \mid
 (x_1, \xi_1) \sim (x_2, \xi_2), \, 
\pm \xi_1 \triangleright 0 \right\}.
\eeq
Here, the notation means $\xi_1 \in T^*_{x_1} \sN$ etc. and $(x_1, \xi_1) \sim (x_2, \xi_2)$ means that $x_1$ and $x_2$ can be 
joined by a null geodesic $\gamma$ such that $\xi_1, \xi_2$ (viewed as vectors using the metric $g$)
are tangent to $\gamma$. $\xi \triangleright 0$ means that that the corresponding vector is future-pointing.
If $K(x_1,x_2)$ is a distributional kernel, the primed wave front set is defined as $\WF'(K) = \{(x_1, \xi_1, x_2, \xi_2) \mid
(x_1, \xi_1, x_2, -\xi_2) \in \WF(K)\}$. With this characterization 
one can show:

\begin{theorem}
\label{eq:LightconeHadamard}
The local vacuum state \eqref{eq:LVS} is Hadamard in $J^+(p) \cap \cO$.
\end{theorem}

\begin{proof}
The full proof is given in \cite{HollandsPhD}.
Apart from a minor subtlety arising from tip of the lightcone at $p$ $(r=0)$,
 it uses only straightforward techniques from microlocal analysis. The ingredients are as follows.
\begin{itemize}
\item The wave front set of the Schwartz kernel $K_E$ of the commutator function 
satisfies $\WF'(K_E) = \cC^+ \cup \cC^-$; this is a restatement of the ``propagation of singularities theorem'' \cite{DuistermaatHoermanderII}.
\item The wave front set of the distribution kernel $K_A(r_1, \ws_1, r_2, \ws_2)=\delta(\ws_1, \ws_2)/(r_1-r_2-i0)^2$ satisfies 
$\WF'(K_A) = \{(\ws, r, \xi_\ws, \xi_r; \ws, r, \xi_\ws, \xi_r) \mid \xi_r>0\}$ away from the tip $r=0$ by direct computation.
So, microlocally, the corresponding operator $A$ ``removes the negative frequencies''.
\item The kernel $K_R$ of the restriction operator $RF=F|_{\dot J^+(p) \cap \cO}$ from $C^\infty$-functions on 
$\cO$ to $C^\infty$-functions on $\dot J^+(p) \cap \cO$, 
parameterized by $(r,\ws)$ via a diffeomorphism $\psi: (0,r_0) \times S^2 \to \dot J^+(p) \cap \cO \setminus \{p\}$ away from the tip
satisfies:
\beq
\WF'(K_R) \subset \left\{ (r, \ws, \xi_r, \xi_\ws; y, \eta) \mid y = \psi(r, \ws), \text{$(\ud \psi)^t(r, \ws)\eta = (\xi_r, \xi_\ws)$ if $r>0$} \right\} 
\eeq
by the restriction theorem (thm.~8.2.4 of \cite{HoermanderI}).
\end{itemize}
Looking at the definition of the local vacuum state, we see that the 2-point function is a composition of the kernels of $E,A,R$. 
These results are then combined with a standard result about the composition of distributional kernels. Let $B: \cD(\cX_1) \rightarrow \cD'(\cX_2)$ and 
$A: \cD(\cX_2) \rightarrow \cD'(\cX_3)$ be linear continuous maps. 
By the Schwartz Kernel theorem these correspond to 
distribution kernels $K_B \in \cD'(\cX_2 \times \cX_1)$ 
and $K_A \in \cD'(\cX_3 \times \cX_2)$.
If $K \in \cD'(\cX_1 \times \cX_2)$, one defines 
\beq
\WF(K)_{\cX_2} \defeq \{ (x_2, \xi_2) \mid (x_1, 0; x_2, \xi_2)
\in \WF(K)\}.
\eeq
Then one has (thm.~8.2.14 of \cite{HoermanderI}):
If $K_B$ has proper support, and 
\beq
\WF'(K_A)_{\cX_2} \cap \WF'(K_B)_{\cX_2} = \emptyset
\eeq
then the composition $A \circ B$ is well defined and
\begin{eqnarray}
\WF'(K_{A\circ B}) &\subset& \WF'(K_A) \circ \WF'(K_B) \cup \\ 
&& ((\cX_1 \times \{0\}) \times \WF(K_B)_{\cX_3}) 
\cup (\WF(K_A)_{\cX_1} \times (\cX_3 \times \{0\})), \nn
\end{eqnarray} 
where 
\begin{eqnarray}
\WF'(K_A) \circ \WF'(K_B) & \defeq & 
\{(x_3, \xi_3, x_1, \xi_1) \mid \text{there exist $\xi_2 \neq 0$ and $x_2$ s.t.} \\
&& \text{$(x_2, \xi_2, x_1, \xi_1) \in \WF'(K_B)$ and 
$(x_3, \xi_3, x_2, \xi_2) \in \WF'(K_A)$} \}. \nn 
\end{eqnarray}

Combining this information, one then finds that $\WF'(G_p)$ is contained in the set $\cC^+$. The tip of the lightcone in 
effect plays no role for this argument, as backward null geodesics starting in the interior of $J^+(p) \cap \cO$ never 
reach it\footnote{The tip of the lightcone only needs to be examined a little more carefully for the proof that 
the commutator property holds, and plays a role similar to $i^-$ in the subsequent construction in RNdS.}. The 
theorem is proven if we can show that both sets are in fact equal. 
To see this, introduce $G_p^t(f_1, f_2) = G_p(f_2, f_1)$.
Then, by the definition of the wave front set, 
\begin{eqnarray*}
\WF(G_p^t) = \WF(G_p)^t \defeq \{
(x_1, \xi_1, x_2, \xi_2) \mid (x_2, \xi_2, x_1, \xi_1) \in \WF(G_p) \}, 
\end{eqnarray*}
from which it follows that $\WF'(G_p^t) \subset \cC^{-}$. 
Therefore, by the commutator property of the 2-point function,
\begin{eqnarray*}
\cC^+ \cup \cC^- = \WF'(E) = \WF'(G_p-G_p^t) \subset 
\WF'(G_p) \cup \WF'(G_p^t) \subset \cC^+ \cup \cC^-,
\end{eqnarray*}
and the above inclusions must in fact be equalities.
Since $\cC^+ \cap \cC^- = \emptyset$, this means that
in fact $\WF'(G_p) = \cC^+$, which is the statement of 
the theorem.
\end{proof}  

\subsubsection{The ``Unruh'' and a ``comparison'' state in RNdS.}
\label{sec:UnruhFinal}
We now define two states by a similar construction in RNdS which play a role for the subsequent argument. 

\medskip
\noindent
{\bf Unruh state:} The first is an analogue 
of the ``Unruh state'' \cite{Unruh:1976db} in Schwarzschild. We first define (see fig.~\ref{fig:1}) 
\begin{align}
\label{eq:def:H_c_H}
 \dSH & = \dSH^- \cup \dSH^R, &
 \EH & = \EH^- \cup \EH^L,
\end{align}
so $\dSH \cup \EH$ is formally, i.e.\ apart from the ``point'' $i^-$ at infinity, a characteristic surface for the union $\sN$ 
of the regions $\rI$, $\rII$, $\rIII$, where the Unruh state will be defined. Heuristically, the Unruh state is specified as in \eqref{eq:statedef}
by the ``positive frequency modes ($k \ge 0$)'' $\psi_{k\ell m}^{\rU, \rin/\rup}$, which are solutions $(\square-\mu^2) \psi_{k\ell m}^{\rU, \rin/\rup}=0$ defined on $\sN$ in terms of their initial data
 on $\dSH \cup \EH$:\footnote{This is to be understood in the sense of wave packets, \cf sec.~\ref{sec:modes1} for how to translate this into asymptotic data of the corresponding radial functions.}
\begin{subequations}
\label{eq:Psi_MU_4d}
\begin{align}
\label{eq:Psi_MU_4d_in}
\psi^{\rU, \rin}_{k \ell m} (x) & = 
\begin{cases} 
(2 \pi)^{-\frac{3}{2}} (2 k)^{-\frac{1}{2}} Y_{\ell m}(\theta, \phi) r_\co^{-1} e^{- i k V} 
& \text{ on } \dSH, \\ 
 0
 & \text{ on } \EH, 
 \end{cases} \\
\label{eq:Psi_MU_4d_up}
\psi^{\rU, \rup}_{k \ell m} (x) & = 
\begin{cases} 
0 
& \text{ on } \dSH, \\ 
 (2 \pi)^{-\frac{3}{2}} (2 k)^{-\frac{1}{2}} Y_{\ell m}(\theta, \phi) r_\ev^{-1} e^{- i k U} 
 & \text{ on } \EH, 
 \end{cases}
\end{align}
\end{subequations}
where $V=V_\co$ is the Kruskal type (affine) coordinate adapted to the cosmological horizon, and where $U$ is the 
Kruskal type (affine) coordinate adapted to the event horizon; see sec.~\ref{sec:2}.

In order to see that the Unruh state is well-defined, and to understand some of its basic properties, 
we formally rewrite the definition \eqref{eq:statedef} (with the modes $\psi_{k\ell m}^{\rU, \rin/\rup}$) in a similar way as in  \eqref{eq:LVS}, 
noting that the VanVleck determinant vanishes on a Killing horizon. 
Let  $f_1, f_2 \in C^\infty_0(\sN)$, and define for $X \in \{ \ev, \co \}$ the restrictions of the corresponding solutions, 
$F_1^X=E^-f_1 |_{\cH_X}, F_2^X=E^-f_2 |_{\cH_X}$. Then the precise definition of the Unruh state is:

\begin{definition}
The Unruh state is defined by the 2-point function on $\sN \times \sN$ given by:
\beq
\label{eq:US}
\begin{split}
\langle \Phi(f_1) \Phi(f_2) \rangle_{\rU} =& -\frac{1}{\pi} \int_\R \int_\R \int_{S^2} 
\frac{F_1^\ev(U_1,\Omega) F_2^\ev(U_2, \Omega)}{(U_1-U_2-i0)^2}  \ud^2 \Omega \ud U_1 \ud U_2 \\
&-\frac{1}{\pi} \int_\R \int_\R \int_{S^2} 
\frac{F_1^\co(V_1,\Omega) F_2^\co(V_2, \Omega)}{(V_1-V_2-i0)^2}  \ud^2 \Omega \ud V_1 \ud V_2
\end{split}
\eeq
where null-generators of the event and cosmological horizon are parameterized by $U$ resp.\ $V$, 
and $\ud^2 \Omega$ is the induced integration element on the spheres 
of constant $V$ resp.\ $U$. 
\end{definition}

This definition is actually still formal as it stands, because the convergence of the integrals over $U$ and $V$ 
has not been ensured. Also, one needs to show the commutator, field equation, and positivity properties S1) -- S3), and one would also like to show that the Unruh state is Hadamard in the union $\sN$ of the regions $\rI$, $\rII$, $\rIII$ where it is defined. These properties have been proven in the case of Schwarzschild spacetime ($\Lambda=Q=0$) by 
\cite{DappiaggiMorettiPinamonti_Unruh}. A similar proof can be given in the present context, taking into account the features of the dynamical evolution special to RNdS.

First, the positivity property S3) is obvious -- assuming \eqref{eq:US} is well defined -- because $-(x-i0)^{-2}$ is a kernel of positive type (its Fourier
transform is non-negative). The field equation S2) is also obvious, because the commutator function $E$ is a bi-solution. 
For a Cauchy surface $\Sigma$ as in fig.~\ref{fig:1}, $E$ can be written in the form 
\cite{Dimock1980}:
\beq
\label{eq:Erel}
E(f_1,f_2) = \int_\Sigma n^\mu (F_1 \nabla_\mu F_2 - F_2 \nabla_\mu F_1) \ud \eta_\Sigma,  
\eeq
with $F_1 = Ef_1, F_2=Ef_2$. Since $(\square-\mu^2) F_1=(\square -\mu^2)F_2=0$, Gauss' theorem shows that 
the expression is independent of the chosen Cauchy surface for $\sN$. Now, if we could formally deform $\Sigma$ to 
the pair of null surfaces $\dSH \cup \EH$, then we would have, formally 
\beq
\label{eq:com}
E(f_1,f_2) = \int_\R \int_{S^2}  (F_1^\ev \partial_U F_2^\ev - F_2^\ev \partial_U F_1^\ev) \ud^2 \Omega \ud U +
\int_\R \int_{S^2}  (F_1^\co \partial_V F_2^\co - F_2^\co \partial_V F_1^\co) \ud^2 \Omega \ud V . 
\eeq
Although this argument is not rigorous because it ignores a potential contribution from $i^-$, let us proceed for the moment and assume that the integrals 
in \eqref{eq:US} converge in a suitable sense. Then using the formula $\Im (x-i0)^{-2} = - \pi \delta'(x)$, on both terms, one formally gets
\beq
\label{eq:com1}
\begin{split}
&\langle \Phi(f_1) \Phi(f_2) \rangle_\rU - \langle \Phi(f_2) \Phi(f_1) \rangle_\rU = 2i \Im \langle \Phi(f_1) \Phi(f_2) \rangle_\rU \\
=& i\int_\R \int_{S^2}  (F_1^\ev \partial_U F_2^\ev - F_2^\ev \partial_U F_1^\ev) \ud^2 \Omega \ud U
+i\int_\R \int_{S^2}  (F_1^\co \partial_V F_2^\co - F_2^\co \partial_V F_1^\co) \ud^2 \Omega \ud V .
\end{split}
\eeq
Combining with the previous equation, this would show the commutator property S1). 

In order to make this argument rigorous, we should first check 
\eqref{eq:com}, and for that, we need the asymptotic behavior of $F_1, F_2$ near $i^-$.
This asymptotic behavior has been analyzed by Hintz and Vasy \cite{HintzVasy} and is best expressed 
resolving $i^-$ by a blow up procedure sketched in sec.~\ref{sec:2}. The results of \cite{HintzVasy} express the regularity of the forward solution $F=E^+ f$ near $i^+$ in terms of certain ``variable order'' Sobolev spaces. By time reflection symmetry, these analogously hold for $F = E^- f$ near $i^-$. The variable orders of (fractional)
differentiability express the different regularity properties of the solution near the various horizons, i.e.\ the different locations of $r$, see app.~A of \cite{HintzVasy}
for the precise definitions. The variable forward order functions $s(r)$ considered by these authors are such that $s(r)$ is constant for $r < r_\ca + \epsilon$ and $r > r_\ca + 2 \epsilon$, with
\beq
s(r)
\begin{cases}
< 1/2 + \alpha/\kappa_\ca & \text{for $r<r_\ca + \epsilon$}\\
\ge 1/2 + {\rm max}\{\alpha/\kappa_\co, \alpha/\kappa_\ca\} &   \text{for $r > r_\ca + 2\epsilon$}
\end{cases}
\eeq 
and $s'(r)\ge 0$. The parameter $\alpha$ of main interest is the spectral gap, but forward order functions can in principle be considered for any $\alpha \in \R$.
Based on a forward order function, a certain Sobolev space  $H^s_{\rb, +}(\cR)$ of functions supported towards the future of the Cauchy surface $\Sigma$, 
is then defined, see app.~B of \cite{HintzVasy} for details, and fig.~\ref{fig:BlowUp} for the definition of the domain $\cR$. Further, \cite{HintzVasy} define the spaces $H^{s,\alpha}_{\rb, +}(\cR) = \tau^\alpha H^s_{\rb, +}(\cR)$ of functions
with an exponential decay towards $t_* \to \infty$, i.e.\ towards $i^+$, see sec.~\ref{sec:2} for the definition of $t_* = - \log \tau$. Their main result can be phrased as follows:

\begin{theorem}\label{thm:HV}
If\footnote{The restriction to $\mu^2 > 0$ is made here simply to avoid complications arising from constant solutions in the case $\mu^2 = 0$.} $\mu^2 > 0$ and $f \in C^\infty_0(\sN)$, then $F=E^+ f \in H^{s,\alpha}_{\rb, +}(\cR)$ for any forward order function of weight $\alpha < \min (\kappa_\ca, \inf \{ - \Im(\omega_{n \ell m}) \})$, and furthermore
$\| F \|_{H^{s,\alpha}_{\rb, +}(\cR)} \lesssim \| f \|_{C^m(\sN)} $ for some $m$ depending on $s$.  
\end{theorem}

\begin{proof}
The proof of this theorem is essentially contained in that of prop.~2.17 of \cite{HintzVasy}. Here we only somewhat pedantically go through the norm estimates provided by \cite{HintzVasy}.
\cite{HintzVasy} first modify the wave operator to $P = \square - \mu^2 -iQ$, where
$Q$ is a suitable
formally self-adjoint pseudo-differential operator having its support in the region $r<r_\ca$ inside the black hole behind the Cauchy horizon $\CH^L$. The
construction of $Q$ is roughly made in such a way that singularities propagating into
the interior of the black hole get ``absorbed'' by $Q$. Then they show (thm.~2.13)
that with a suitable choice of $-\epsilon<\tilde \alpha<0$ and of $Q$, the map $P$ is a Fredholm operator between their weighted ``forward'' Sobolev spaces $H^{\tilde s,\tilde \alpha}_{\rb,+}(\cR)
\to H^{\tilde s-1,\tilde \alpha}_{\rb,+}(\cR)$, with $\tilde s$ a forward order function for
$\tilde \alpha$, and with $\cR$ a domain of the shape indicated in fig.~\ref{fig:BlowUp}. More precisely, $\cR$ should be extended to the left beyond the Cauchy horizon $r<r_\ca$
by effectively modifying the metric function $f(r)$ such that the singularity at $r=0$ gets replaced with the origin of polar coordinates or another horizon. By the standard theory of Fredholm operators, the inverse $P^{-1}$ is defined as a bounded operator on a subspace of $H^{\tilde s-1,\tilde \alpha}_{\rb,+}(\cR)$ of finite co-dimension. In fact, it is shown in \cite{HintzVasy}, lemma 2.15 that the projector onto this subspace may be chosen as $Rf=f-\sum_{i=1}^N \phi_i v_i(f)$, with $\phi_i$ smooth and supported in $r<r_\ca$ and $v_i$ distributions. Thus, if $f \in C^\infty_0(\sM)$, it follows that the forward
solution of $PF = f$ has
\beq
\| F \|_{H^{\tilde s,\tilde \alpha}_{\rb,+}(\cR)} \lesssim \| f \|_{C^m(\sM)}
\eeq
for a sufficiently large $m$ depending on the choice of $\tilde s$.
Then, by the propagation of singularities theorem, prop.~2.9 of \cite{HintzVasy}, one has
\beq
\| F \|_{H^{s,\tilde \alpha}_{\rb,+}(\cR)} \lesssim \| F \|_{H^{\tilde s,\tilde \alpha}_{\rb,+}(\cR)}
+ \| P F \|_{H^{\tilde s,\tilde \alpha}_{\rb,+}(\cR)} \lesssim \|f\|_{C^m(\sM)} ,
\eeq
with $s\le \tilde s +1$ a forward order function for any weight $\alpha>0$ such that
$\alpha/\kappa_\ca <1$. By \cite{HintzVasy}, sec.~2 and \cite{Vasy2013}, lemma~3.1 and remark~3.4,
(essentially a Fourier-Laplace-transform argument in the variable $t_*=-\log \tau$), $F$
can be developed in an asymptotic expansion
\beq
 F(x,t_*) = \sum_j t_*^{m_j} e^{\sigma_j t_*} a_j(x) + F'(x,t_*)
\eeq
with $-\alpha<\sigma_j<-\tilde \alpha$ the resonances of the Fourier-Laplace-transformed operator
$\widehat P(\sigma)$, $m_j$ their multiplicity, and with a norm bound
\beq
\| F' \|_{H^{s,\alpha}_{\rb,+}(\cR)} \lesssim \|f\|_{C^m(\sM)}.
\eeq
By \cite{HintzVasy}, lemma~2.14, $F' = F$ in the region $r>r_\ca$ if $\alpha$ is chosen so that there are no quasi-normal frequencies for the scattering problem in region I with $- \tilde \alpha > \Im(\omega_{n\ell m}) > - \alpha$ -- in fact they argue that any resonance of $\widehat P(\sigma)$ not equal to such a quasi-normal mode\footnote{Note that $e^{\sigma t_*} a(x)$ satisfies QNM boundary conditions for $\sigma < 0$ due to the definition of $t_*$ in sec.~\ref{sec:2}.} has $a_j(x)=0$ for $r>r_\ca$.

In the case at hand, we have $\mu^2>0$. It follows that there can be no such quasi-normal frequencies on the real line (unlike for the massless wave operator considered by \cite{HintzVasy}, who have to consider separately the constant mode), as one can see by a simple argument involving the Wronskian of the radial equation, \cf \eqref{eq:psi_omega_ell_wave_eqn} below, for real $\omega$. Thus, we may put $\alpha$ to be the value \eqref{eq:aldef}.
\end{proof}

These results will now be used in order to investigate the properties of the Unruh state. As already mentioned, thm.~\ref{thm:HV} analogously holds for the behavior of the backward solution $F = E^- f$ near $i^-$. We are interested in its behavior in region $\rI$. But, as no lower bound is imposed on the forward order function $s(r)$ in the region $r_\ev \le r \le r_\co$, using Sobolev embedding, there follows a bound
\beq
\label{eq:estimate}
| \partial^N F(t_*,r,\ws) | \lesssim e^{\alpha t_*}  \| f \|_{C^m(\sN)} \quad \text{for $r_\ev \le r \le r_\co$,}
\eeq
where $\partial^N$ is a derivative in $(t_*,r,\ws)$ of order $N>0$ and $m$ depends only on $N$. By choosing in \eqref{eq:Erel} a sequence of Cauchy surfaces as depicted in fig.~\ref{fig:3} moving downwards towards the horizontal boundary, it immediately follows by writing \eqref{eq:Erel} out in the coordinates $(t_*, r, \ws)$ that in the limit only the contributions from the two vertical 
boundaries $\EH^-, \dSH^-$ contribute near $i^-$, but not the interpolating surface depicted in green. Hence, in this limit, \eqref{eq:com} holds true. 

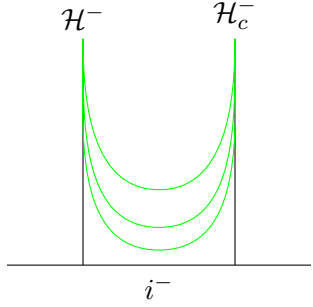
\begin{figure}
\centering
\begin{tikzpicture}[scale=1]
 \draw (-1,0) -- (3,0) node[midway,below]{$i^-$};
 \draw (0,0) -- (0,3) node[above]{$\EH^-$};
 \draw (2,0) -- (2,3) node[above]{$\dSH^-$};
 \draw[green] (0,3) to[out=270, in=180] (1,1);
 \draw[green] (2,3) to[out=270, in=0] (1,1);
 \draw[green] (0,3) to[out=270, in=180] (1,0.5);
 \draw[green] (2,3) to[out=270, in=0] (1,0.5);
 \draw[green] (0,3) to[out=270, in=180] (1,0.2);
 \draw[green] (2,3) to[out=270, in=0] (1,0.2);
\end{tikzpicture}
\caption{Sketch of the ``blow-up'' of $i^-$ and of a sequence of Cauchy surfaces as described in the text.}
\label{fig:3}
\end{figure}

We may also see immediately that the integrals in \eqref{eq:US} converge when $f_1, f_2$ are functions 
in $C^\infty_0(\sN)$, making the Unruh state actually well defined. For this, we note that 
$e^{\alpha t_*} \sim |U|^{-\alpha/\kappa_\ev}$ for $U \to -\infty$ on $\EH$ and $e^{\alpha t_*} \sim |V|^{-\alpha/\kappa_\co}$ for $V \to -\infty$ on $\dSH$. 
In combination with \eqref{eq:estimate}, this easily implies convergence of the integrals \eqref{eq:com1}. We therefore conclude:

\begin{proposition}
The Unruh state is well-defined on the algebra $\cA(\sN), \sN= \rI \cup \rII \cup \rIII$ generated by $\Phi(f)$ with $f \in C^\infty_0(\sN)$. 
\end{proposition}

We now verify that the Unruh 2-point function $G_\rU(x_1, x_2) = \langle \Phi(x_1) \Phi(x_2) \rangle_\rU$ is Hadamard in the union of regions $\rI$, $\rII$, $\rIII$ (but not necessarily across $\CH^R$!). Let us begin by recalling basic results from the analysis of partial differential equations. For any distribution $\varphi$ on $\cX$ and (pseudo-) differential operator $A$, it is known that 
\beq
\label{eq:8_3_1}
 \WF(\varphi) \subset \WF(A\varphi) \cup {\rm char}(A),
\eeq
where ${\rm char}(A)=\{(x,\xi) \in T^* \cX \setminus o \mid a(x,\xi)=0\}$, and where $a(x,\xi)$ is the principal symbol of $A$, see Thm.~8.3.1 of \cite{HoermanderI}. Furthermore, the propagation of singularities theorem \cite{DuistermaatHoermanderII} states that if the differential operator $A$ has real principal symbol $a$ and if 
 $A\varphi$ is smooth then $\WF(\varphi)$ is contained in the intersection of ${\rm char}(A)$ and the Hamiltonian orbits of 
 the symbol $a(x, \xi)$ on ``phase space'' $T^* \cX \setminus o$. In the case of $A=\square - \mu^2$, the principal symbol is 
 $g^{\mu\nu}(x) \xi_\mu \xi_\nu$, and the Hamiltonian orbits in the characteristic set ${\rm char}(\square - \mu^2)$ are precisely the 
 null geodesics with initial condition $(x,\xi)$. We denote one such an orbit, called bicharacteristic, as $\cB(x,\xi)$. Based on this definition, we distinguish three cases:

\medskip
\noindent
{\it 1) Neither the bicharacteristic $\cB(x_1,\xi_1)$ nor $\cB(x_2,\xi_2)$ enters region $\rI$:} In this case, the claim follows by exactly the same argument as given in the previous section for lightcones, because every backward null geodesic emanating from these regions will eventually hit $\EH \cup \dSH$. 

\medskip
\noindent
{\it 2) The bicharacteristics $\cB(x_1,\xi_1)$ and $\cB(x_2,\xi_2)$ enter region $\rI$:} By propagation of singularities, it suffices to consider $x_1, x_2 \in \rI$. In region $\rI$, past-directed null geodesics trapped on the photon ring (see e.g.\ \cite{WaldGR}), will not come from $\EH \cup \dSH$ but begin their lives in the ``point'' $i^-$. Thus, the type of argument made in the case of lightcones -- where this cannot occur -- does not work, and one needs another argument worked out in detail in the case of Schwarzschild by 
\cite{DappiaggiMorettiPinamonti_Unruh}. The argument is based on the fact that the Unruh state is static, i.e. invariant under translations in 
the coordinate $t$. A simplified version of the argument which also works in the present case runs as follows. 

First, for an arbitrary but fixed open region $\cO$ whose closure is contained in region I (i.e. not touching the horizons $\EH \cup \dSH$),  we introduce the maps, $X \in \{ \co, \ev \}$,
\beq
\label{eq:KX}
K^X : C^\infty_0(\cO) \to L^2(\R_\omega \times S^2_\ws) ,
\quad K^X f (\omega, \ws) = r_X \left(
\frac{\omega e^{2\pi r_X \omega}}{\sinh (2\pi r_X \omega)}\right)^{\frac{1}{2}} \int_\R F^X(s,\ws) e^{i\omega s} \ud s
\eeq
where $F^\ev(u,\ws) = Ef |_{\EH^-}(u,\ws)$ and $F^\co(v,\ws) = Ef |_{\dSH^-}(v,\ws)$. Next, 
we let $Kf := K^\ev f \oplus K^\co f \in L^2 \oplus L^2$, and then we can the write the Unruh 2-point function as \cite{KayWald91}
\beq
\label{eq:US_K}
\langle \Phi(f_1) \Phi(f_2) \rangle_\rU = \langle Kf_1, Kf_2 \rangle_{L^2 \oplus L^2}, 
\eeq
which follows from \eqref{eq:US} by Fourier transformation and the relations $U=-e^{-\kappa_\ev u}, V=-e^{-\kappa_\co v}$. 
The decay and regularity results, thm. \ref{thm:HV},
together with the Sobolev embedding theorem 
and the relation between $t_*$ and $u$ resp.\ $v$ on $\EH^-$ resp. $\dSH^-$
imply that $|\partial_s^N  F^X(s,\ws)|  \le C_N e^{-\alpha |s|}$, where the 
constant $C_N$ is controlled by some $C^m$ norm of $f$ by \eqref{eq:estimate}. As a consequence, 
$K = K^\ev \oplus K^\co$ is shown to be a distribution in $\cO$ with values in the Hilbert space 
$L^2 \oplus L^2$: Since $|\partial_s^N  F^X(s,\ws)|  \lesssim \| f \|_{C^m(\cO)} e^{-\alpha |s|}$ for all $f \in C^\infty_0(\cO)$, we find $\| Kf \|_{L^2 \oplus L^2} \lesssim \|f\|_{C^m(\cO)}$, which is 
the continuity required from a distribution. 

Next, invariance of $E$ under translations of $t$ in 
the coordinates $(t,r_*,\ws)$ and the exponentially decreasing prefactor of order $O(e^{2\pi r_X \omega})$ 
in \eqref{eq:KX} for $\omega \to -\infty$ imply that 
the distribution $K$ has a holomorphic extension to the strip $\{ t + is \mid s \in (0, 2\pi r_\ev ) \}$ in the $t$-coordinate. 
Moreover it is the boundary value of this holomorphic extension, in the sense of distributions, as $s \to 0^+$, 
 for all test-functions from $C^\infty_0(\cO)$. By a very slight modification of the proof of thm.~3.1.14 of \cite{HoermanderI} we then get, for some $m$, 
 the inequality 
 \beq 
 \|K(t+is, r, \ws)\|_{L^2 \oplus L^2}
 \lesssim s^{-m} \quad \text{ for $(t,r,\ws) \in \cO$ and $s \in (0,\pi r_\ev)$}. 
 \eeq
 In view of thm.~8.4.8 of \cite{HoermanderI}, we then further get 
 \beq
 \WF(K) |_\cO \subset \{ (x, \xi) \in T^* \cO \setminus o \mid \langle \xi, \partial_t \rangle > 0\}. 
 \eeq
Combining \eqref{eq:US_K} and the rules for the wave front set under the composition of distributions, we finally get 
\beq
\WF'(G_\rU)|_{\cO \times \cO} \subset \{ (x_1, \xi_1, x_2, \xi_2) \in T^*(\cO \times \cO) \setminus o \mid \langle \xi_i, \partial_t \rangle > 0 \} . 
 \eeq
 Next, using the relationship between $t$ and $u,v$, it follows immediately from \eqref{eq:KX}, \eqref{eq:US_K} that the 
 2-point function $G_\rU$ is invariant under translation of $t$ in both arguments. In infinitesimal form, 
 $(\partial_{t_1}+\partial_{t_2})G_\rU=0$. 
From \eqref{eq:8_3_1}, it follows that in fact 
 \beq\label{eq:WF1}
\WF'(G_\rU)|_{\cO \times \cO} \subset \{ (x_1, \xi_1, x_2, \xi_2) \in T^*(\cO \times \cO) \setminus o \mid \langle \xi_1, \partial_t \rangle =
\langle \xi_2, \partial_t \rangle > 0 \} . 
 \eeq
 If we assume that $\cO$ is 
 the entire (open) region $\rI$, viewed as a globally hyperbolic spacetime in its own right with Cauchy-surface $\Sigma_\cO$, 
 then it follows from the propagation of singularities theorem that
 \beq\label{eq:WF2}
\WF'(G_\rU)|_{\cO \times \cO} \subset \{ \cB(x_1, \xi_1) \times \cB(x_2, \xi_2) \setminus o \mid x_i \in \Sigma_\cO, g(\xi_i, \xi_i) =0 \} . 
 \eeq
 
 We now distinguish the cases a) $\cB(x_1, \xi_1) = \cB(x_2, \xi_2)$ and b) $\cB(x_1, \xi_1) \neq \cB(x_2, \xi_2)$. In case a), we obtain, from \eqref{eq:WF1} and \eqref{eq:WF2}, that $(x_1, \xi_1, x_2, \xi_2) \not \in \WF'(G_\rU)$, unless $\xi_1 \sim \xi_2$ and $\xi_1 \triangleright 0$. In case b), we may, without loss of generality, assume that $x_1, x_2 \in \Sigma_\cO$ are spacelike separated. We may then choose real-valued 
 cutoff functions $\chi_1, \chi_2$ supported in a coordinate chart near $x_1, x_2$ such that 
 $\chi_i(x_i) \neq 0$, with $\supp(\chi_1)$ remaining spacelike to ${\rm \supp}(\chi_2)$.
 Since $G_\rU$ is a distribution of positive type, 
 it follows in view of the Cauchy-Schwarz inequality that 
 \beq\label{eq:WF3}
 |G_\rU(\chi_1 e_{-k_1}, \chi_2 e_{k_2} )| \le 
 |G_\rU(\chi_1 e_{-k_1}, \chi_1 e_{k_1})|^{\frac{1}{2}}
 |G_\rU(\chi_2 e_{-k_2}, \chi_2 e_{k_2})|^{\frac{1}{2}}, 
 \eeq
where $e_{k}(x) \defeq e^{- i k x}$. Now let 
 \beq
 \label{eq:Def_V_i_pm}
  \cV_i^\pm=\{ (x, \xi) \in T^* \cO \setminus o \mid x \in {\rm supp} \chi_i, g(\xi,\xi)=0, \langle \xi, \partial_t \rangle \gtrless 0\}.
 \eeq 
 It follows from  the wave front conditions \eqref{eq:WF1}, \eqref{eq:WF2} applied to the right side of \eqref{eq:WF3}
 that the left side of \eqref{eq:WF3} is decaying rapidly in $(k_1, k_2)$ (i.e.\ is bounded by 
 $\le C_N (1+|k_1|+|k_2|)^{-N}$ for any $N$) in any conic neighborhood not intersecting $\cV_1^+ \times \cV_2^+$. 
 However, since the arguments are spacelike, the commutator property of 2-point functions also gives
 \beq
  G_\rU(\chi_1 e_{-k_1}, \chi_2 e_{k_2}) = G_\rU(\chi_2 e_{k_2}, \chi_1 e_{-k_1}).
 \eeq
 Applying the 
 same reasoning to $G_\rU(\chi_2 e_{k_2}, \chi_1 e_{- k_1})$, it follows that the left side of \eqref{eq:WF3} is decaying as fast in $(k_1, k_2)$ in any conic neighborhood not intersecting $\cV_1^- \times \cV_2^-$. For $\supp \chi_{1/2}$ small enough, $\cV_1^+ \times \cV_2^+ \cap \cV_1^- \times \cV_2^- = \emptyset$, so that $(x_1, \xi_1, x_2, \xi_2) \not \in \WF'(G_\rU)$ in case b). Combining this information, we have
  \beq\label{eq:WF4}
  \begin{split}
\WF'(G_\rU)|_{\cO \times \cO} \subset & \{ (x_1, \xi_1, x_2, \xi_2) \in T^*(\cO \times \cO) \setminus o \mid \text{$x_1,x_2$ connected by null geodesic $\gamma$,} \\
&\quad \xi_1 = c_1g( \ . \ ,  \dot \gamma_1), \xi_2 = c_2 g( \ . \ , \dot \gamma_2),
\langle \xi_1, \partial_t \rangle = \langle \xi_2, \partial_t \rangle > 0 \} .   
\end{split}
 \eeq
 The set on the right hand side is easily seen to be equal to $\cC^+$, thus we learn $\WF'(G_\rU)|_{\cO \times \cO} \subset \cC^+$, and 
 we actually get equality by the commutator argument used at the end of the proof of thm.~\ref{eq:LightconeHadamard}.

 \medskip
\noindent
{\it 3) One bicharacteristic, say $\cB(x_1,\xi_1)$, will enter $\rI$, but $\cB(x_2,\xi_2)$ not:} We can argue as in subcase b) above, i.e.\ using that the right hand side of the inequality \eqref{eq:WF3} is decaying rapidly in $(k_1, k_2)$, due to 1), 2).

\medskip
\noindent
{\bf Comparison state:} 
For the sake of the discussion in the next section, it is also convenient to have another state that is stationary and 
that is manifestly Hadamard in a two-sided open neighborhood of the Cauchy horizon $\CH^R$, see fig.~\ref{fig:1}. This state is not of particular
physical interest and introduced purely in order to make our arguments with greater ease. In order to define it, we modify the metric function 
$f(r)$ beyond the Cauchy horizon for $r<r_\ca -\delta$ with a small $\delta > 0$ in such a way that it smoothly interpolates between $f(r)$ and the constant function $1$ in a neighborhood of $r = 0$. We call the new function $f_*(r)$. Evidently, 
it defines the RNdS metric in regions $\rI$, $\rII$, $\rIII$ and an open neighborhood of $\CH^R$, but replaces the singularity in region $\rIV$ with 
a smooth origin of polar coordinates $r=0$, see region $\rIV'$ in fig.~\ref{fig:5}.

\begin{figure}
\centering
\begin{tikzpicture}[scale=1]
 \draw (-2,2) -- (-4,4) node[midway, below, sloped]{$\EH^L$};
 \draw (-2,2) -- (0,4) node[midway, above, sloped]{$\EH^R$};
 \draw (0,4) -- (-2,6) node[midway, below, sloped]{$\CH^R$};
 \draw (-4,4) -- (-2,6) node[midway, above, sloped]{$\CH^L$};

 \draw[dashed] (0,4) -- (0,8);
 \draw (-2,6) -- (0,8) node[midway, above, sloped]{$\CH^+$};
 \draw (-2,4) node{${\rm II}$};
 \draw (-1,6) node{${\rIV'}$};

 \draw[fill=white] (-4,4) circle (1pt);
 \draw[fill=white] (0,8) circle (1pt);
 \draw[fill=black] (-2,6) circle (1pt);
 \draw[fill=black] (-2,2) circle (1pt);
 \draw[fill=white] (0,4) circle (1pt) node[above right]{$i^+$};
  
  \draw[dotted] (0,8) .. controls (-1.8,6) .. (0,4);
\end{tikzpicture}
\caption{Sketch of the spacetime in which the comparison state is Hadamard. The dashed line indicates the origin of coordinates. The region beyond the dotted line is region where $f(r)$ is modified.}
\label{fig:5}
\end{figure}
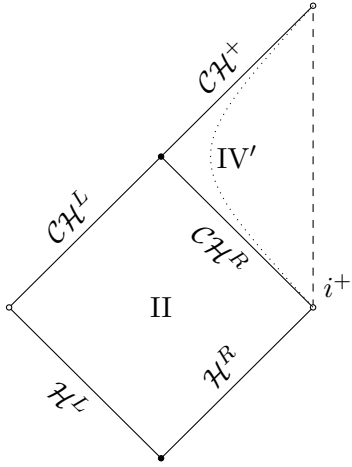

In this -- made-up unphysical -- spacetime, we now define a vacuum state called $\langle \, . \, \rangle_\rF$ for regions $\rII$, $\rIV'$ by the ``positive frequency modes ($k \ge 0$)'' $\psi_{k\ell m}^{\rF, \rout}$, which are solutions $(\square-\mu^2) \psi_{k\ell m}^{\rF,\rout}=0$ defined on the union of $\rII$, $\rIV'$ in terms of their final data
 on $\CH = \CH^+ \cup \CH^L$: 
\beq
\label{eq:Psi_F_4d}
\psi^{\rF, \rout}_{k \ell m} (x)  =
(2 \pi)^{-\frac{3}{2}} (2 k)^{-\frac{1}{2}} Y_{\ell m}(\theta, \phi) r_\co^{-1} e^{- i k V} 
\qquad \text{ on $\CH$} ,
\eeq
where $V=V_\ca$ is the Kruskal type (affine) coordinate adapted to this Cauchy horizon, 
see fig.~\ref{fig:5}. Again, it can be shown by exactly the same methods as in the lightcone case that these modes define 
via \eqref{eq:statedef} a Hadamard 2-point function $\langle \Phi(x_1) \Phi(x_2) \rangle_\rF$ in regions $\rII$, $\rIV'$, which we call the ``comparison state''.

\medskip
We have shown:
 
 \begin{theorem}
 \label{thm:hadaU}
 The Unruh state has a 2-point function of Hadamard form in the union of region $\rI$, $\rII$, $\rIII$.  The comparison state is Hadamard in the union of region $\rII$ and an open (two-sided) neighborhood of the Cauchy horizon $\CH^R$.
 \end{theorem}

\section{Behavior of the quantum stress tensor at the Cauchy horizon in 4 dimensions}
\label{sec:5}

In this section, we will prove our main result, \eqref{eq:MainResult}, concerning the behavior of the stress-energy tensor near the Cauchy horizon, $\CH^R$, of 4 dimensional RNdS spacetime for an arbitrary state $\Psi$ that is Hadamard in the union of regions I, II, and III. We will proceed by writing $\langle T_{\mu\nu} \rangle_\Psi$ as
 \beq
\label{eq:Tdecomp}
\langle T_{\mu\nu} \rangle_\Psi = 
(\underbrace{\langle T_{\mu\nu} \rangle_\Psi -\langle T_{\mu\nu} \rangle_\rU}_{= t_{\mu\nu} \text{ as in prop.~\ref{thm:hadadiff}} }) \quad
+ \quad 
(\underbrace{\langle T_{\mu\nu} \rangle_\rU -\langle T_{\mu\nu} \rangle_\rF}_{\text{calculated in sec.~\ref{sec:Numerics}}}) \quad
+ \underbrace{\langle T_{\mu\nu} \rangle_\rF}_{\text{smooth at } \CH^R} . 
\eeq
In subsection \ref{sec:Reduction}, we will show that the first term behaves like the stress-energy tensor of a classical field. Thus, it has behavior not more singular than \eqref{eq:ClassResult}. The last term is smooth by construction of the state $\langle \ . \ \rangle_\rF$. Most of our effort will be spent in calculating the middle term, which will be done via a mode expansion that will be carried out in subsections \ref{sec:modes1}-\ref{sec:Numerics}. 

In the 2 dimensional case, we found that the limiting behavior of $\langle T_{vv} \rangle_\rU - \langle T_{vv} \rangle_\rF$ on $\CH^R$ was given by \eqref{eq:ModeSum2dBoulware}. Since $\langle T_{\mu \nu} \rangle_\rF$ is smooth on $\CH^R$, this implies that 
$\langle T_{vv} \rangle_\rU \to (\kappa_\co^2 - \kappa_\ca^2)/48 \pi$ on $\CH^R$. Thus $\langle T_{vv} \rangle_\rU$ goes to a finite limit on $\CH^R$, which is nonvanishing unless $\kappa_\co = \kappa_\ca$. As discussed in sec.~\ref{sec:2d}, for $\kappa_\co \neq \kappa_\ca$, this implies that $\langle T_{V_\ca V_\ca} \rangle_\rU$ blows up on $\CH^R$ as $1/V^2_\ca$, a worse singularity than classical behavior. In 4 dimensions, the logical possibilities are that (i) $\langle T_{vv} \rangle_\rU$ does not approach a limit at $\CH^R$, (ii) $\langle T_{vv} \rangle_\rU$ approaches a nonzero limit at $\CH^R$, or (iii) $\langle T_{vv} \rangle_\rU \to 0$ on $\CH^R$. In case (ii), we again would have that $\langle T_{V_\ca V_\ca} \rangle_\rU$ blows up on $\CH^R$ as $1/V^2_\ca$. In case (i) the behavior of $\langle T_{V_\ca V_\ca} \rangle_\rU$ would be more singular than a $1/V^2_\ca$ blow up. Only in case (iii) would the quantum stress-energy possibly behave in a manner similar to the classical case. 

In subsection \ref{sec:modes1}, we will obtain a mode sum expression for $\langle T_{vv} \rangle_\rU - \langle T_{vv} \rangle_\rF$ in terms of scattering coefficients (see \eqref{eq:stint}). The scattering coefficients are calculated in subsection \ref{sec:modes2} for the conformally coupled case, $\mu^2 = \frac{1}{6} R = \frac{2}{3}\Lambda$. The mode sum is then evaluated numerically in subsection \ref{sec:Numerics}. We find that the mode sum converges to a nonzero value, i.e., case (ii) occurs and the behavior of the stress-energy tensor is given by \eqref{eq:MainResult} with $C \neq 0$. Although our evaluation of $C$ is done only in the conformally coupled case $\mu^2 = \frac{2}{3}\Lambda$ and at particular black hole parameters, we would expect $C$ to vary analytically with $\mu$ and the black hole parameters, in which case we would have $C \neq 0$ generically.

\subsection{Reduction to the Unruh state}
\label{sec:Reduction}

In this subsection, we show that $\langle T_{\mu\nu} \rangle_\Psi -\langle T_{\mu\nu} \rangle_\rU$ has the same behavior near $\CH^R$ as the stress-energy tensor of a classical field.
Let $\langle \ . \ \rangle_\Psi$ be an arbitrary state that is Hadamard near the Cauchy surface $\Sigma$ as drawn in fig.~\ref{fig:1}. By the propagation of 
singularities \cite{RadzikowskiVerch}, it remains Hadamard throughout regions $\rI$, $\rII$, $\rIII$, but it is not defined (at least not uniquely) in region IV, just as if we were to specify 
initial data for the classical wave equation on $\Sigma$. Since the state is Hadamard in $\rI$, $\rII$, $\rIII$, the expectation value $\langle T_{\mu\nu} \rangle_\Psi$ is finite and smooth in $\rI$, $\rII$, $\rIII$, but may diverge\footnote{The coordinate indices of course should refer to a coordinate system regular there such as $(U,V,\theta,\varphi)$, since we are not interested in 
artificial singularities created by a bad choice of coordinates.} as we move towards the Cauchy horizon, $\CH^R$. We would like to see how it diverges, if at all. The 
first step is to relate the behavior of $\langle T_{\mu\nu} \rangle_\Psi$ to behavior of the corresponding quantity $\langle T_{\mu\nu} \rangle_\rU$ for the Unruh state, 
which we have a better way of calculating.

Since we are interested in the local behavior, we pick a point  $p$ on the Cauchy horizon $\CH^R$ contained in a coordinate chart, $(x^\mu)=(V,y^i)$, where 
$V=V_\ca$ is the Kruskal-type coordinate and $y^i, i=1,2,3$ are local coordinates parameterizing  $\CH^R$. So, $V=0$ locally defines  $\CH^R$, and the intersection of this chart with 
region $\rII$ locally corresponds to $V<0$. Let $\cU$ be a small open neighborhood contained in this coordinate chart. Then the following proposition
basically follows from thm.~\ref{thm:HV}.

\begin{proposition}
\label{thm:hadadiff}
Let $\Psi, \Psi'$ be Hadamard states near the Cauchy surface $\Sigma$ as drawn in fig. \ref{fig:1}, and assume $\beta > 1/2$, see \eqref{eq:betadef}. Then for any $\beta'$ such that $1/2 < \beta' < \min(\beta,1)$,
\begin{itemize} 
\item
each component of
$
t_{\mu\nu}:=\langle T_{\mu\nu} \rangle_{\Psi} - \langle T_{\mu\nu} \rangle_{\Psi'},
$
is smooth in $y$, and viewed as a function of $V<0$, is locally in $L^p(\R_-)$, 
where $1/p = 2-2\beta'$ (thus $p>1$), with norm uniformly bounded in $y$ within $\cU$.
\item
viewed as a function of $V$, $t_{\mu\nu}$ has $s=\beta' - 1/2$ derivatives locally in $L^p(\R_-)$ for $1/p = 3/2 - \beta'$.
\end{itemize}
\end{proposition}

\medskip
\noindent
{\bf Remarks:} 1) For a classical field with smooth initial data on $\Sigma$, the stress tensor is a function known to have the same regularity as $t_{\mu\nu}$, so 
the result is that the difference between the quantum stress tensors in two states that are initially Hadamard behaves basically like the classical 
stress tensor near the Cauchy horizon.
2) The restriction to $\beta>1/2$ is of technical nature. For $\beta \le 1/2$, already a classical solution $\phi$ would generically 
fail to be of Sobolev class $H^1$ near $\CH^R$, thus making it impossible to extend the solution through $\CH^R$ already at the classical level, 
see \cite{DafermosLukI} for a discussion of such matters. Thus, $\beta > 1/2$ is the interesting case for the discussion of sCC.
3) Numerically, it is found that $\beta < 1$ \cite{CardosoEtAl18}, so that the restriction on the range of $\beta'$ can actually be written as $1/2 < \beta' < \beta$.

\begin{proof}
Let $\sN = \rI \cup \rII \cup \rIII$. Since $\sN$ is globally hyperbolic and since
both states are Hadamard near $\Sigma$, it follows by propagation of singularities that 
the difference between the corresponding two point functions $W(x_1,x_2) = \langle \Phi(x_1) \Phi(x_2) \rangle_\Psi -  \langle \Phi(x_1) \Phi(x_2) \rangle_{\Psi'}$
is in $C^\infty(\sN \times \sN)$.  It has been shown in \cite{Verch1994}, lemma 3.7 that given a causal normal neighborhood $\tilde \cN$ of $\Sigma$ (see \cite{KayWald91} for the definition), and any open subset $\cO$ with compact closure contained in $\tilde \sN$, there exists a smooth function 
$B$ with support in $\cO \times \cO$ such that 
\beq
W(f_1, f_2) = \int_{\cO \times \cO} (Ef_1)(x_1) (Ef_2)(x_2) B(x_1,x_2) \ud \eta(x_1) \ud \eta(x_2) , 
\eeq
for any testfunctions $f_1, f_2$ with support in the domain of dependence $D(\cO)$. By making $\cO$ sufficiently large, we may assume that 
the neighborhood $\cU$ of interest near the Cauchy horizon is contained in $D(\cO)$. Furthermore, by the algebraic relations A3,A4) and the positivity property S3) of states, 
we have $B(x_1, x_2) = \overline{B(x_1,x_2)} = B(x_2,x_1)$.

Let $m$ be a fixed natural number. Using the result \cite{Verch1994}, Appendix~B, and the symmetry properties of $B$, one can see 
that there exist $C^m_0$ functions $\{b_j\}$ on $\cO$ 
such that 
\beq
\label{Bnuclear}
B(x_1,x_2) = \sum_j c_j \overline{b_j(x_1)} b_j(x_2), \quad \sum_j \| b_j \|^2_{C^m(\cO)} < \infty, \quad c_j \in \{\pm 1 \}. 
\eeq
Now let $\psi_j = E^+ b_j$, with $E^+$ the retarded fundamental solution. Then $(\square - \mu^2) \psi_j = b_j$, and furthermore, 
\beq
\label{eq:Wnuclear}
W(x_1, x_2) = \sum_j c_j \overline{\psi_j(x_1)} \psi_j(x_2) \quad \text{for $x_1, x_2 \in \cU$.} 
\eeq
To estimate $\psi_j$, we can apply the results of \cite{HintzVasy} already mentioned in thm.~\ref{thm:HV}. 
First, we define new local coordinates $(z,y^i)$ in the following way. First we locally parameterize the orbit $S^2_p$ of $p \in \CH^R$
by angles $\Omega$. Then we transport these coordinates away from $\CH^R$ along $\partial/\partial V$ and let $z$ be the parameter
along this curve. Finally, we transport the coordinates $(z,\Omega)$ along the Killing vector field $\partial/\partial t$, and then $(y^i)=(t,\Omega)$. 

If we apply an arbitrary product $\prod_{j=1}^N X^{(a_j)}$ of the vector fields $X^{(a)}, a=1,2,3,4$ generating the 
symmetry group $\R \times SO(3)$ of RNdS to the defining equation $(\square - \mu^2) \psi_j = b_j$ (with retarded initial conditions) and 
commute these through $\square - \mu^2$, the function $\prod_{a=1}^N X^{(a)} \psi_j$ is seen to satisfy an equation of the same form and initial conditions 
with a new source. By thm.~\ref{thm:HV} and the construction of $(z,y^i)$, this gives
\beq
\label{eq:SobolevCM}
\|\partial_{y}^N \psi_j\|^2_{H^{1/2+\beta'}(\cU)} \lesssim   \| b_j \|^2_{C^m(\cO)}
\eeq 
for an implicit constant independent of $j$ and a sufficiently large $m$ depending on $N$. 
We can also use our expression for the difference of the expected stress tensor in two Hadamard states, see \eqref{eq:ptsplit}.
This shows in combination with \eqref{eq:Wnuclear} that 
\beq
t_{\mu\nu}(x) = \sum_j c_j \{ 
\partial_\mu \overline{\psi_j(x)} \partial_\nu \psi_j(x)
-
\tfrac{1}{2} g_{\mu\nu} (\partial_\gamma \overline{\psi_j(x)} \partial^\gamma \psi_j(x) + \mu^2 \overline{\psi_j(x)}  \psi_j(x))
\}.
\eeq
A standard argument based on the Fourier transform in $y$, Cauchy's inequality and Parseval's theorem 
shows that for $\varphi$ supported in the coordinate patch covered by $(z,y^i)$,
\beq
\label{eq:fourier}
\begin{split}
|\varphi(z, y)| =& \int \ud^3 \xi \tilde \varphi(z, \xi) e^{iy \xi}  \\
\le& \bigg( \int \ud^3 \xi (1+\xi^2)^{-N/2} \bigg)^{1/2} \bigg( \int \ud^3 \xi (1+\xi^2)^{N/2} |\tilde \varphi(z, \xi)|^2 \bigg)^{1/2} \\
\lesssim & \bigg\| (-\partial_y^2+1)^{N/2} \varphi(z, -) \bigg\|_{L^2(\R_y^3)},
\end{split}
\eeq
taking $N>3$. In combination with the Sobolev embedding theorem, $p$ as in the statement of the proposition,
\beq
\|  \varphi(-, y) \|_{L^{2p}(\R_-)} 
\lesssim
\| (-\partial_z^2+1)^{(\beta'-1/2)/2} \varphi(-, y) \|_{L^{2}(\R_-)} 
\lesssim  
\| (-\partial_y^2+1)^{N/2} \varphi \|_{H^{\beta'-1/2}(\cU)} ,
\eeq
uniformly in $y$ within $\cU$. We now take $\varphi=\partial_\mu \psi_j$ or $\varphi=\psi_j$ and combine this with \eqref{eq:SobolevCM},
thereby obtaining 
\beq
\| \partial_\mu \psi_j(-, y) \|_{L^{2p}(\R_-)} \lesssim  \| b_j \|_{C^m(\cO)}.
\eeq
Therefore, using the Cauchy-Schwarz inequality,
\beq
\begin{split}
&\| t_{\mu\nu} (-, y) \|_{L^{p}(\R_-)}\\
\lesssim & \sum_j \bigg( \sum_{\gamma=1}^4 \| \partial_\gamma \psi_j (-, y) \|_{L^{2p}(\R_-)} + \|  \psi_j (-, y) \|_{L^{2p}(\R_-)} \bigg)^2\\
 \lesssim & \sum_j  \| b_j \|^2_{C^m(\cO)} < \infty.
 \end{split}
\eeq
By considering $\varphi=\partial_\mu \partial_y^M \psi_j$, we can similarly obtain a corresponding bound for $\| \partial_y^M t_{\mu\nu} (-, y) \|_{L^{p}(\R_-)}$.
Transforming from the coordinates $(z,y^i)$ to $(V,y^i)$ gives the first statement.
 
To prove the second result, we use that \eqref{eq:fourier} also implies 
\beq
\| D^{s/2}_z \varphi(-, y) \|_{L^{2 p}(\R_-)} \lesssim  \| (-\partial_y^2+1)^{N/2} \varphi\|_{H^{s}(\cU)}
\eeq
for $(-\partial_z^2+1)^{s/2} = D^s_z$, $s$ and $p$ as indicated in the second statement, and $\varphi$ of compact support in $\cU$. With $\varphi=\partial_\mu \psi_j$, 
we get
\beq
\| D^{s/2}_z \partial_\mu \psi_j(-, y) \|_{L^{2 p}(\R_-)} \lesssim  \|\partial_{y}^N \psi_j\|^2_{H^{s+1}(\cU)}, 
\eeq
for $s = \beta'-1/2$. 
We have classical commutator estimates of the type \cite{KenigPonceVega}, Thm.~A.8,
\beq
\| D^s (f_1 f_2) - D^s f_1 f_2 - f_1 D^s f_2 \|_{L^p(\R^n)} \lesssim \| D^{s'} f_1 \|_{L^{p'}(\R^n)} \| D^{s''} f_1 \|_{L^{p''}(\R^n)}
\eeq
for $0<s=s'+s'', s', s''>0$, $0<1/p=1/p'+1/p''<1$. These basically allow us to distribute the fractional derivatives $D^s_z$ on a product $\partial_\mu \bar \psi_j \partial_\nu \psi_j$
such as in $t_{\mu\nu}$ up to an error term which is controlled in a lower Sobolev norm using $\|\partial_{y}^N \psi_j\|^2_{H^{s+1}(\cU)} \lesssim \| b_j \|_{C^m(\cO)}$. The second statement then follows
taking $s'=s''=s/2$, $p'=p''=2p$. This concludes the proof of the proposition.
\end{proof}

The proposition says that for the purposes of calculating $\langle T_{\mu\nu} \rangle_\Psi$ near the Cauchy horizon in a state $\Psi$ which 
starts out as a Hadamard state near the initial time surface $\Sigma$, we may work with the Unruh state $\langle T_{\mu\nu} \rangle_\rU$
because, as we have shown in thm.~\ref{thm:hadaU}, it is also a Hadamard state on the initial time surface $\Sigma$. In doing this, we must  
accept an error $t_{\mu\nu}$ which has the regularity described in prop.~\ref{thm:hadadiff} at the Cauchy horizon, and which, 
as we have hinted, is the behavior of the classical stress tensor. This is acceptable since, as we will show, the 
behavior of $\langle T_{\mu\nu} \rangle_\rU$ is more singular than this at the Cauchy horizon $\CH^R$. 

We turn, now to the calculation of the middle term in \eqref{eq:Tdecomp}, namely $\langle T_{\mu\nu} \rangle_\rU - \langle T_{\mu\nu} \rangle_\rF$. By stationarity of the Unruh and the comparison state, 
it suffices to perform the computation in a neighborhood of $\CH^L$. The calculation is outlined in subsection \ref{sec:modes1}, with the details of the mode calculations 
relegated to subsection \ref{sec:modes2}.

\subsection{Outline of mode calculation for Unruh state}
\label{sec:modes1}

The mode calculation is a straightforward effective 1-dimensional scattering construction. It proceeds in two steps:
\begin{itemize}
\item We solve a scattering problem in region I, with asymptotic conditions determined by our choice of initial state 
on the surfaces $\EH^-, \dSH^-$. 
\item We subsequently solve a scattering problem in region II, with asymptotic conditions determined by our choice of initial state 
on $\EH^L$ and by the data from the previous step on $\EH^R$. 
\end{itemize}
Because similar constructions have been described in the literature previously for the case of RN \cite{Sela18}, we will be brief and focus on the 
novel aspects.  Actually, the calculation is not more difficult for the operators $(\partial_V^N \Phi)^2$, which are equal to $T_{VV}$ for $N=1$, using the 
obvious generalization of \eqref{eq:ptsplit} to the operators $(\partial_V^N \Phi)^2$. 

Since $\partial v/\partial V = - (\kappa_\ca V)^{-1}$ with $V=V_\ca$ the outgoing 
Kruskal-type coordinate regular in a neighborhood of $\CH^R$, we can work instead with the variable $v$, so we are interested in 
$\langle (\partial_v^N \Phi)^2 \rangle_\rU-\langle (\partial_v^N \Phi)^2 \rangle_\rF$. 
As in \cite{Sela18}, for computational purposes only (not changing the state!), 
it is useful to consider the ``Boulware'' modes\footnote{These are analogous to \eqref{eq:LeftMover_2d} in 
2 dimensions.}, defined by the wave equation and asymptotic data\footnote{This has to be understood in the sense of wave packets, \cf \eqref{eq:f_omega_ell_limit} for how to translate this into asymptotic conditions for a scattering problem.}
\begin{subequations}
\label{eq:Psi_Boulware_4d}
\begin{align}
\label{eq:Psi_in_I_4d}
 \psi^{\rin, \rI}_{\omega \ell m} & = \begin{cases} (2 \pi)^{-\frac{3}{2}} (2 |\omega|)^{-\frac{1}{2}} Y_{\ell m}(\theta, \phi) r_\co^{-1} \Theta(-V) e^{- i \omega v} & \text{ on } \dSH, \\ 0 & \text{ on } \EH, \end{cases} \\
\label{eq:Psi_up_I_4d}
 \psi^{\rup, \rI}_{\omega \ell m} & = \begin{cases} 0 & \text{ on } \dSH \\ (2 \pi)^{-\frac{3}{2}} (2 |\omega|)^{-\frac{1}{2}} Y_{\ell m}(\theta, \phi) r_\ev^{-1} \Theta(-U) e^{- i \omega u} & \text{ on } \EH, \end{cases} \\
\label{eq:Psi_up_III_4d}
 \psi^{\rup, \rII}_{\omega \ell m} & = \begin{cases}  0 & \text{ on } \dSH \\ (2 \pi)^{-\frac{3}{2}} (2 |\omega|)^{-\frac{1}{2}} Y_{\ell m}(\theta, \phi) r_\ev^{-1} \Theta(U) e^{- i \omega u} & \text{ on } \EH, \end{cases} \\
\label{eq:Psi_out_III_4d}
 \psi^{\rout, \rII}_{\omega \ell m} & = \begin{cases} (2 \pi)^{-\frac{3}{2}} (2 |\omega|)^{-\frac{1}{2}} Y_{\ell m}(\theta, \phi) r_\ca^{-1} \Theta(-V) e^{- i \omega v} & \text{ on } \CH^L \\ 0 & \text{ on } \CH^+, \end{cases}
\end{align}
\end{subequations}
with $\Theta$ the step function, and where we used the notation \eqref{eq:def:H_c_H}. Furthermore, $V=V_{\co/\ca}$ in the first/last equation. For later convenience, we defined these modes also for $\omega < 0$. We omitted the modes $\psi^{\rin, \rIII}$ and $\psi^{\rout, \rIV}$, which are not relevant for the calculations that we are going to perform. In the next section~\ref{sec:modes2}, we shall find the exact solutions to the above characteristic initial value problems in mode form, i.e.\ in the general form 
\beq
\label{eq:modeansatz}
\psi_{\omega \ell m}(t,r,\theta,\phi) = (2 \pi)^{-\frac{3}{2}} (2 | \omega |)^{-\frac{1}{2}} Y_{\ell m}(\theta, \phi) r^{-1} R_{\omega \ell}(r) e^{- i \omega t}, 
\eeq
where the radial function $R$ has to solve
\beq
\label{eq:psi_omega_ell_wave_eqn}
 - \del_{r_*} \del_{r_*} R_{\omega \ell} + (V_\ell - \omega^2) R_{\omega \ell} = 0,
\eeq
with a smooth effective potential
\beq
 V_\ell = f \left( \frac{\ell (\ell + 1)}{r^2} + \frac{f'}{r} + \mu^2 \right).
\eeq
The particular solutions corresponding to various asymptotic conditions are denoted by superscripts, such as in 
$\psi^{\rin, \rI}_{\omega \ell m}$, with corresponding radial function $ R^{\rin, \rI}_{\omega \ell}$, etc. For example, 
$ R^{\rin, \rI}_{\omega \ell}$ should satisfy the asymptotic condition
\beq
\label{eq:f_omega_ell_limit}
 R^{\rin, \rI}_{\omega \ell}(r) = \begin{cases} e^{- i \omega r_*} + \sR^{\rin,\rI}_{\omega \ell} e^{i \omega r_*} & r_* \to \infty, \\ \sT^{\rin,\rI}_{\omega \ell} e^{- i \omega r_*} & r_* \to - \infty. \end{cases}
\eeq
Since the effective potential is smooth and, when expressed in terms of $r_*$, decays, with all its derivatives, faster than any power, we must have \cite{YafaevII}, Prop.~5.2.9,
\begin{align}
 \sR^{\rin,\rI}_{\omega \ell} & = \begin{cases} \order(\omega^{-\infty}) & \omega \to \infty, \\ 1 + \order(\ell^{-\infty}) & \ell \to \infty, \end{cases} &
 | \sT^{\rin,\rI}_{\omega \ell} |^2 & = \begin{cases} 1 + \order(\omega^{-\infty}) & \omega \to \infty, \\ \order(\ell^{-\infty}) & \ell \to \infty. \end{cases},
\end{align}
where $\order(\omega^{-\infty})$ indicates a decay faster than any power. 
This provides asymptotic data on $\EH^R$ and $\EH^L$ (the latter vanish), 
which can be transported to $\CH^L$ by again solving \eqref{eq:psi_omega_ell_wave_eqn}, now with the asymptotic condition
\beq
\label{eq:f_omega_ell_limit_2}
 R^{\rin,\rII}_{\omega \ell}(r) = \sT^{\rin,\rI}_{\omega \ell} \begin{cases}  e^{- i \omega r_*} & r_* \to - \infty, \\ \sT^{\rin,\rII}_{\omega \ell} e^{- i \omega r_*} + \sR^{\rin,\rII}_{\omega \ell} e^{i \omega r_*} & r_* \to \infty. \end{cases}
\eeq
Note that the new coefficients fulfill the relation
\beq
\label{eq:tau_rho_relation}
 | \sT^{\rin,\rII}_{\omega \ell} |^2 - | \sR^{\rin,\rII}_{\omega \ell} |^2 = 1,
\eeq
which in particular means that they need not be bounded. Nevertheless, we again have
\begin{align}
 \sR^{\rin,\rII}_{\omega \ell} & = \order(\omega^{-\infty}), &
 | \sT^{\rin,\rII}_{\omega \ell} |^2 & = 1 + \order(\omega^{-\infty}).
\end{align}
We thus obtain, for the restriction of the ``Boulware'' in-mode to $\CH^L$,
\beq
 \label{eq:Psi_in_I_4d_at_-}
 \psi^{\rin, \rI}_{\omega \ell m} |_{\CH^L} = \sT^{\rin,\rII}_{\omega \ell} \sT^{\rin,\rI}_{\omega \ell} \psi^{\rout, \rII}_{\omega \ell m}.
\eeq 
For the up-modes $\psi^{\rup, \rI}_{\omega \ell m}$, we proceed similarly, i.e., we first scatter them to $\EH^R$ by solving again \eqref{eq:psi_omega_ell_wave_eqn}, now with the asymptotic condition
\beq
\label{eq:AsymptoticCondition_I_up}
 R_{\omega \ell}^{\rup,\rI}(r) = \begin{cases} e^{i \omega r_*} + \sR^{\rup,\rI}_{\omega \ell} e^{-i \omega r_*} & r_* \to - \infty, \\ \sT^{\rup,\rI}_{\omega \ell} e^{i \omega r_*} & r_* \to \infty. \end{cases}
\eeq
One then scatters to $\CH^L$ as for the in-modes and obtains
\beq
\label{eq:Psi_up_I_4d_at_-}
 \psi^{\rup, \rI}_{\omega \ell m} |_{\CH^L} = \sT^{\rin,\rII}_{\omega \ell} \sR^{\rup,\rI}_{\omega \ell} \psi^{\rout, \rII}_{\omega \ell m},
\eeq
For the up-modes $\psi^{\rup, \rII}_{\omega \ell m}$, only a single scattering is necessary. The radial function 
$ R^{\rup,\rII}_{\omega \ell}$ is defined by the asymptotic condition
\beq
\label{eq:f_omega_ell_limit_up}
 R^{\rup,\rII}_{\omega \ell}(r) = \begin{cases} e^{i \omega r_*} & r_* \to - \infty, \\ \sT^{\rup,\rII}_{\omega \ell} e^{i \omega r_*} + \sR^{\rup,\rII}_{\omega \ell} e^{- i \omega r_*} & r_* \to \infty. \end{cases}
\eeq
This gives
\beq
\label{eq:Psi_up_III_4d_at_-}
 \psi^{\rup, \rII}_{\omega \ell m} |_{\CH^L} = \sR^{\rup,\rII}_{\omega \ell} \psi^{\rout, \rII}_{\omega \ell m}.
\eeq
Comparing \eqref{eq:f_omega_ell_limit_up} with \eqref{eq:f_omega_ell_limit_2} yields
\begin{align}
\label{eq:up_in_relation}
 \sR^{\rup,\rII}_{\omega \ell} & = \overline{\sR^{\rin,\rII}_{\omega \ell}}, &
 \sT^{\rup,\rII}_{\omega \ell} & = \overline{\sT^{\rin,\rII}_{\omega \ell}}.
\end{align}

The symmetrized two-point function, i.e., the Hadamard function, of the Unruh state defined by the modes \eqref{eq:Psi_MU_4d} can be expressed in terms of the Boulware modes \eqref{eq:Psi_Boulware_4d} as, \cf \cite{Lanir:2017oia} for a similar calculation (but note the change of sign of $u$ in the interior region \wrt that reference),
\begin{align}
\label{eq:G_MU}
 & \langle \{ \Phi(x_1), \Phi(x_2) \} \rangle_\rU \\
& = \sum_{\ell m} \int_0^\infty \ud \omega  \left[ \coth \frac{\pi \omega}{\kappa_\co} \left\{ \psi^{\rin, \rI}_{\omega \ell m}(x_1), \overline{\psi^{\rin, \rI}_{\omega \ell m}(x_2)} \right\} + 2 \csch \frac{\pi \omega}{\kappa_\ev} \Re \left\{ \psi^{\rup, \rI}_{\omega \ell m}(x_1), \overline{\psi^{\rup, \rII}_{\omega \ell m}(x_2)} \right\} \right. \nn \\ 
 & \qquad \left. + \coth \frac{\pi \omega}{\kappa_\ev} \left( \left\{ \psi^{\rup, \rI}_{\omega \ell m}(x_1), \overline{\psi^{\rup, \rI}_{\omega \ell m}(x_2)} \right\} + \left\{ \psi^{\rup, \rII}_{-\omega \ell m}(x_1), \overline{\psi^{\rup, \rII}_{-\omega \ell m}(x_2)} \right\} \right) \right], \nn
\end{align}
where we omitted the contribution from $\psi^{\rin, \rIII}_{\omega \ell m}$, which is absent in $\rI \cup \rII$, and used the notation
\beq
 \{ A(x_1), B(x_2) \} = \frac{1}{2} \left( A(x_1) B(x_2) + A(x_2) B(x_1) \right).
\eeq

We now turn to the comparison state. 
The specification of this state in terms of modes on a null hypersurface has the property that, after selecting an angular momentum component,
\beq
 \Phi_{\ell m}(U, V) = r \int \overline{Y_{\ell m}(\ws)} \Phi(U,V,\ws) \ud^2 \Omega,
\eeq
and taking at least one derivative \wrt $V=V_\ca$, one can evaluate the two-point function on the null hypersurface, i.e.,
\beq
  \lim_{U \to \infty} \langle \del^n_V \Phi_{\ell_1 m_1}^*(U, V_1) \del^n_{V_2} \Phi_{\ell_2 m_2}(U, V_2) \rangle_\rF = \delta_{\ell_1 \ell_2} \delta_{m_1 m_2} 
  \frac{1}{16 \pi^3} \int_0^\infty \ud k \ k^{2n -1} e^{- i k (V_1 - V_2 - i 0)}.
\eeq
When both points are on $\CH^L$, one may reexpress this as, \cf \cite{Lanir:2017oia}  for example,
\beq
 \lim_{U \to \infty} \langle \del^n_{v} \Phi_{\ell_1 m_1}^*(U, v_1) \del^n_{v} \Phi_{\ell_2 m_2}(U, v_2) \rangle_\rF = \delta_{\ell_1 \ell_2} \delta_{m_1 m_2} \frac{1}{16 \pi^3} \int_0^\infty \ud \omega \ \omega^{2n -1} \coth \frac{\pi \omega}{\kappa_\ca} e^{- i \omega (v_1 - v_2 - i 0)}.
\eeq

With \eqref{eq:Psi_in_I_4d_at_-}, \eqref{eq:Psi_up_I_4d_at_-}, \eqref{eq:Psi_up_III_4d_at_-} and \eqref{eq:G_MU}, we thus obtain for the difference of the expectation value in the Unruh and the comparison state, in a coinciding point limit,
\begin{multline}
\label{eq:del_V_Phi_ell_m_del_V_Phi_ell_m}
\lim_{U \to \infty, v_1 \to v_2} \left[ \langle \del_v^n \Phi_{\ell_1 m_1}^*(U, v_1) \del_v^{n} \Phi_{\ell_2 m_2}(U, v_2) \rangle_\rU - \langle \del_v^n \Phi_{\ell_1 m_1}^*(U, v_1) \del_v^{n} \Phi_{\ell_2 m_2}(U, v_2) \rangle_\rF \right] \\
 = \delta_{\ell_1 \ell_2} \delta_{m_1 m_2} \frac{1}{16 \pi^3}  \int_0^\infty \ud \omega \ \omega^{2n-1}   n_{\ell_1}(\omega) , 
\end{multline}
where
 \beq
 \begin{split}
\label{eq:del_V_Phi_ell_m_del_V_Phi_ell_m1}
 n_{\ell}(\omega)
 =& \ | \sT^{\rin,\rI}_{\omega \ell} |^2 | \sT^{\rin,\rII}_{\omega \ell} |^2 \coth \frac{\pi \omega}{\kappa_\co} + \left( | \sR^{\rup,\rI}_{\omega \ell} |^2 | \sT^{\rin,\rII}_{\omega \ell}|^2 + | \sR^{\rup,\rII}_{\omega \ell} |^2 \right) \coth \frac{\pi \omega}{\kappa_\ev} \\
& + 2 \csch \frac{\pi \omega}{\kappa_\ev} \Re (\sR^{\rup,\rI}_{\omega \ell} \sT^{\rin,\rII}_{\omega \ell} \overline{\sR^{\rup,\rII}_{\omega \ell}}) - \coth \frac{\pi \omega}{\kappa_\ca} .
 \end{split}
\eeq
From the asymptotic behaviour of the transmission and reflection coefficients, it follows that 
\eqref{eq:del_V_Phi_ell_m_del_V_Phi_ell_m1} falls off faster than any power in $\omega$ for $|\omega| \to \infty$, so that the integral 
\eqref{eq:del_V_Phi_ell_m_del_V_Phi_ell_m} is UV finite for all $n \geq 1$. However, one may worry that there is an IR divergence for $n = 1$, due to the fact that generically $\sT^{\rin,\rII}_{\omega \ell}$ and hence, by \eqref{eq:tau_rho_relation}, also $\sT^{\rin,\rII}_{\omega \ell}$ have a simple pole at $\omega = 0$, \cite{KehleShlapentokhRothman}, Prop.~6.2, so that the second and the third term on the \rhs are individually IR divergent for $n=1$ (in the first term, $| \sT^{\rin,\rI}_{\omega \ell} |^2$ vanishes rapidly as $\omega \to 0$).\footnote{The nature of this pole can also be seen explicitly from the expressions given for the scattering coefficients in sec. \ref{sec:modes2}.} 
However, the potential IR divergences due to these terms cancel: We first notice that $\sR^{\rup,\rI}_{\omega \ell} = - 1 + i C_\ell \omega + \order(\omega^2)$ as $\omega \to 0$, where $C_\ell \in \R$, \cf \cite{Sela18}.  The statement then follows from \eqref{eq:up_in_relation} and the fact, derivable from inspection of the proof of \cite{KehleShlapentokhRothman}, Prop.~6.2, that, in case of a divergent $\sT^{\rin,\rII}_{\omega \ell}$,
\begin{align}
\label{eq:tau_rho_omega_to_0}
 \sT^{\rin,\rII}_{\omega \ell} & = \frac{i C'_\ell}{\omega} + \order(\omega^0), &
 \sR^{\rin,\rII}_{\omega \ell} & = - \frac{i C'_\ell}{\omega} + \order(\omega^0),
\end{align}
where $C'_\ell \in \R$.\footnote{This follows from the fact that the coefficients $A$, $B$ in (6.5) of \cite{KehleShlapentokhRothman} are real.} It is then straightforward to see that the potential IR divergences in \eqref{eq:del_V_Phi_ell_m_del_V_Phi_ell_m} cancel. 

We also note that due to $| \sT^{\rin, \rI}_{\omega \ell} |^2 = | \sT^{\rup, \rI}_{\omega \ell} |^2$, it suffices to determine $\sR^{\rup,\rI}_{\omega \ell}$, $\sT^{\rup,\rII}_{\omega \ell}$, and $\sR^{\rup,\rII}_{\omega \ell}$ in order to evaluate \eqref{eq:del_V_Phi_ell_m_del_V_Phi_ell_m1}.

If the integral on the \rhs of \eqref{eq:del_V_Phi_ell_m_del_V_Phi_ell_m} gives a non-zero value for any $n$, then the Unruh state is singular at $\CH^R$: By stationarity and continuity, the values on $\CH^R$ and $\CH^L$ coincide. Converting the resulting tensor to adapted coordinates $V=V_\ca$ instead of $v$, one obtains 
\beq
\langle \del_V^n \Phi_{\ell m}^* \del_V^n \Phi_{\ell m} \rangle_\rU - \langle \del_V^n \Phi_{\ell m}^* \del_V^n \Phi_{\ell m} \rangle_\rF \sim c_{n\ell} |V|^{-2n},
\eeq
where $c_{n\ell}$ is $\kappa^{2n}_\ca$ times the value of 
the integral on the right hand side of \eqref{eq:del_V_Phi_ell_m_del_V_Phi_ell_m}. 
That the integral vanishes for all $\ell$ and all $n \geq 1$ seems quite improbable, and this expectation is confirmed in sec.~\ref{sec:Numerics}. Because the comparison state 
$\rF$ is Hadamard in an open neighborhood on $\CH^R$ by thm.~\ref{thm:hadaU}, any singular behavior must be due to the Unruh state $\rU$ and not the comparison state $\rF$. Therefore, in view of property H1) of subsection \ref{sec:4.1}, we conclude:

\medskip
\noindent
{\bf Conclusion 1:}
Unless the quantity on the right side of \eqref{eq:del_V_Phi_ell_m_del_V_Phi_ell_m} vanishes for all $\ell$ and all $n \geq 1$, the Unruh state
cannot be a Hadamard state in a two sided open neighborhood of $\CH^R$ (of course it is Hadamard away from $\CH^R$ by thm. \ref{thm:hadaU}).

\medskip 

While for the angular momentum components $\Phi_{\ell m}$ the limits $U \to \infty$ and $v_1 \to v_2$ on the left hand side of \eqref{eq:del_V_Phi_ell_m_del_V_Phi_ell_m} can be controlled and shown to converge to the right hand side, this is not straightforward for $\Phi$, required for the evaluation of the difference of expectation values of local Wick powers such as the stress tensor $T_{vv}$ on $\CH^L$. Noting that we can write
\beq
 \del_v \Phi(x_1) \del_v \Phi(x_2) = \frac{1}{r_1 r_2} \sum_{\ell_1, \ell_2} \sum_{m_1, m_2} Y_{\ell_1 m_1}^*(\Omega_1) Y_{\ell_2 m_2}(\Omega_2) \del_v \Phi_{\ell_1 m_1}^*(U_1, V_1) \del_v \Phi_{\ell_2 m_2}(U_2, V_2),
\eeq
and assuming that we may interchange the limits $U \to \infty$ and $v_1 \to v_2$ with the summations, we obtain
\beq
\label{eq:stint}
 \langle T_{vv} \rangle_\rU - \langle T_{vv} \rangle_\rF = \frac{1}{16 \pi^3 r_\ca^2} \sum_\ell \frac{2 \ell + 1}{4 \pi} \int_0^\infty \ud \omega \ \omega n_{\ell}(\omega) ,
\eeq
with $n_{\ell}(\omega)$ as in \eqref{eq:del_V_Phi_ell_m_del_V_Phi_ell_m1}. As discussed above, the integral on the right hand side converges for all $\ell$. However, the convergence of the sum over $\ell$ is not obvious. Numerical evidence, see below, indicates that it does converge, which we take as an indication that the above interchange of limits is justified. In view of the regularity of $\rF$ on $\CH^R$, the relationship between $v$ and $V_\ca$, and prop.~\ref{thm:hadadiff}, we then also conclude:

\medskip
\noindent
{\bf Conclusion 2:} If the sum on right side of \eqref{eq:stint} converges to a nonzero value, then for any initially a Hadamard state $\Psi$ in a neighborhood of the Cauchy surface $\Sigma$ as in fig. \ref{fig:1} and for any RNdS spacetime with $\beta>1/2$ (see \eqref{eq:betadef}), then the expected stress tensor behaves in region $\rII$ near 
$\CH^R$ as 
\beq
\label{eq:MainResult1}
\langle T_{VV} \rangle_\Psi \sim C |V|^{-2} + t_{VV} \, .
\eeq
Here $C$ is given by $\kappa_\ca^2$ times the right side of \eqref{eq:stint} and only depends on the parameters of the black hole but not on $\Psi$.  
$t_{VV}$, which depends on $\Psi$, has the regularity prop. \ref{thm:hadadiff}, i.e. $t_{VV} \in L^p$ as a function of $V$ locally near $\CH^R$ with 
$1/p = 2-2\beta+2\epsilon$ and with $\beta-1/2-\epsilon$ fractional derivatives w.r.t. $V$ in $L^{1/(3/2 - \beta +\epsilon)}$, where $\epsilon>0$ is arbitrarily small.

\medskip

Based on the results of the next section \ref{sec:modes2}, we shall numerically
evaluate the right sides of \eqref{eq:del_V_Phi_ell_m_del_V_Phi_ell_m1} and \eqref{eq:stint} in the case of the conformally coupled scalar field and find them to be nonvanishing. In particular, this provides evidence that \eqref{eq:MainResult1} holds with $C \neq 0$ for generic values of the black hole parameters.

\medskip
\noindent
{\bf Remark:}
These results should be compared with the expression \eqref{eq:ModeSum2dBoulware} in the 2 dimensional case. In particular, in \eqref{eq:del_V_Phi_ell_m_del_V_Phi_ell_m1} one explicitly sees the effect of backscattering in 
4 dimensions, as the up modes also contribute. Analogous expressions for the evaluation of $T_{vv}$ in the Unruh state on RN (without the subtraction of the contribution of the comparison state) can be found in \cite{Zilberman:2019buh}.

\subsection{Computation of the scattering coefficients}

\label{sec:modes2}

We next describe how to determine the scattering coefficients $\sR^{\rup,\rI}$, $\sT^{\rin,\rII}$, $\sR^{\rin,\rII}$ needed for the evaluation of \eqref{eq:del_V_Phi_ell_m_del_V_Phi_ell_m1}. When expressed in terms of $r$, the radial equation \eqref{eq:psi_omega_ell_wave_eqn} has regular singular points at $r = r_\ca, r_\ev, r_\co, r_o, \infty$ (recall \eqref{eq:def_r_o} for the definition of $r_o$). The transformation 
\beq
 x = \frac{(r_\ca - r_o) (r - r_\ev)}{(r_\ca - r_\ev) (r - r_o)},
\eeq
maps $r_\ca$, $r_\ev$, $r_\co$, $r_o$, $\infty$ to
\begin{align}
 x_\ca & = 1, & 
 x_\ev & = 0, & 
 x_\co & = \frac{(r_\ca - r_o) (r_\co - r_\ev)}{(r_\ca - r_\ev) (r_\co - r_o)}, & 
 x_o & = \infty, &
 x_\infty & = \frac{r_\ca - r_o}{r_\ca - r_\ev}
\end{align}
which are the singular points of the radial equation written in terms of $x$. 
In the conformally coupled case, $\mu^2 = \frac{1}{6} R = \frac{2}{3}\Lambda$, it is possible to factor out the singularity at $x_\infty$ by a suitable transformation \cite{SuzukiEtAl1999}
of the radial equation\footnote{One can in principle also treat the case of different values of $\mu$ by a variant of the method below using results by \cite{Schmidt1979}.}. On account of this simplification, we restrict consideration to the conformally coupled case from this point on. 

We define
\begin{align}
 b_\ev & = + i \frac{\omega}{2 \kappa_\ev}, &
 b_\ca & = - i \frac{\omega}{2 \kappa_\ca}, &
 b_\co & = - i \frac{\omega}{2 \kappa_\co}, &
 b_o & = + i \frac{\omega}{2 \kappa_o},
\end{align}
and make the ansatz
\beq
 (r^{-1} R)(r(x)) = (\pm x)^{\sigma_\ev b_\ev} (1-x)^{\sigma_\ca b_\ca} \left( \frac{x - x_\co}{1 - x_\co} \right)^{\sigma_\co b_\co} \frac{x - x_\infty}{1 - x_\infty} H(x),
\eeq
where $\sigma_X = \pm 1$ and the sign in the first factor is chosen below depending on whether one considers solutions in $\rI$ or $\rII$.
It is then found \cite{SuzukiEtAl1999} that $H(x)$ satisfies Heun's equation \cite{Ronveaux} with the singular points $x_\ca, x_\ev, x_\co, \infty$. The parameters of this 
equation can be written in terms of these and $x_\infty$. 

It is known that solutions to  Heun's equation which are analytic in a domain of the complex plane
containing two of the four singular points can be constructed as series of various special functions, see 
\cite{Schmidt1979} and \cite{Leaver1986JMP,Leaver1986} for the origins of this method. For us, it will be useful to follow a variant of this method where the special functions 
are hypergeometric functions, \cf \cite{AbramowitzStegun}, Chapter~15. Of particular interest for us are the domains containing $\{x_\ev, x_\ca\}$ respectively $\{x_\co, x_\infty \}$. In the first case one can define solutions fulfilling ``in'' or ``up'' boundary conditions at $r_\ev$ and compute the scattering coefficients $\sT^{\rup,\rII}$, $\sR^{\rup,\rII}$. In the second case, one can define solutions fulfilling specified boundary conditions at $r_\co$. As the domains overlap, one can compute the scattering coefficients $\sR^{\rup,\rI}$.

The solution to the radial equation \eqref{eq:psi_omega_ell_wave_eqn} with ``up'' boundary condition at $r_\ev$ in region $\rII$ can be written, up to normalization, as
\begin{multline}
\label{eq:R_up}
 r^{-1} R^{\rup, \rII}_{\omega \ell,\nu}(r) = x^{b_\ev} (1-x)^{b_\ca} \left( \frac{x - x_\co}{1 - x_\co} \right)^{b_\co} \frac{x - x_\infty}{1 - x_\infty} \\
 \times \sum_{n = - \infty}^\infty a_n^\nu F(-n-\nu+b_\ev+b_\ca, n+\nu+b_\ev+b_\ca+1; 1 + 2 b_\ev; x),
\end{multline}
with $F={}_2 F_1$ the Gauss hypergeometric function. The so-called characteristic exponent $\nu$ is, a priori, an undetermined parameter unrelated to the parameters appearing in the radial-- resp. Heun equation. 
The ansatz can be shown  \cite{SuzukiEtAl1999} to yield a formal solution for {\em any} $\nu$, provided the coefficients $a^\nu_n$ solve the three term recursion relation
\beq
\label{eq:3_term_recursion}
 \alpha^\nu_n a^\nu_{n+1} + \beta^\nu_n a^\nu_n + \gamma^\nu_n a^\nu_{n-1} = 0
\eeq
with initial condition
\beq
 a_0^\nu = 1
\eeq
and coefficients
\begin{subequations}
\begin{align}
 \alpha_n^\nu & = - \frac{(n + \nu - b_\ev - b_\ca + 1) (n + \nu - b_\co + b_o + 1) (n + \nu - b_\co - b_o + 1) (n + \nu - b_\ev + b_\ca + 1)}{2 (n + \nu + 1)(2n + 2\nu + 3)}, \\
 \beta_n^\nu & = \frac{(b_\ev + b_\ca) (b_\ev - b_\ca) (b_\co - b_o) (b_\co + b_o)}{2 (n+\nu)(n+\nu+1)} + \left( \frac{1}{2} - x_\co \right) (n+\nu) (n+\nu+1) \nn \\
 & \qquad + \frac{1}{2} \left( (2 b_\co + 1) (b_\ev - b_\ca) - (2 b_\ca + 1) (b_\ev + b_\ca) + (b_\ev + b_\ca + b_\co + 1)^2 - b_o^2 \right) \\
 & \qquad + x_\co \left( (b_\ev + b_\ca)^2 + b_\ev + b_\ca \right) + v, \nn  \\
 \gamma_n^\nu & = - \frac{(n + \nu + b_\ev + b_\ca) (n + \nu + b_\co - b_o) (n + \nu + b_\co + b_o) (n + \nu + b_\ev - b_\ca)}{2 (n + \nu)(2n + 2\nu - 1)},
\end{align}
\end{subequations}
where
\begin{align}
 v & = 2 \omega^2 \frac{r_\ev^3 (r_\ev r_\ca - 2 r_\ca r_\co + r_\ev r_\co) (r_o^2 - r_\ev r_\ca - r_\ev r_\co - r_\ca r_\co)^2}{(r_\ev-r_\ca)^3 (r_\ev-r_\co)^2 (r_\ev-r_o) (r_\co-r_o)} - 2 \frac{r_o^2}{(r_\ev - r_\ca)(r_\co - r_o)} \nn \\
 & \quad - x_\infty - ( 1 + x_\co ) b_\ev - x_\co b_\ca - b_\co - 2 b_\ev (x_\co b_\ca + b_\co) - \ell (\ell+1) \frac{r_o^2 - r_\ev r_\ca - r_\ev r_\co - r_\ca r_\co}{(r_\ev - r_\ca)(r_\co - r_o)}.
\end{align}
However, for generic values of the characteristic exponent $\nu$, the solution is only formal in the sense that the series in \eqref{eq:R_up} fails to converge. 
To get a convergent (bilateral) series, $\nu$ has to be determined such that the recursion relation \eqref{eq:3_term_recursion} has a solution which is ``minimal'' in the 
sense of \cite{Gautschi}, both for the ascending $n \to + \infty$ and for the descending $n \to - \infty$ series. By the general theory of three-term relations, a solution which is 
minimal in both directions can be determined as follows. Using one of the algorithms for finding the minimal solution to a three-term recursion relation described in \cite{Gautschi}, one determines
\begin{align}
 \rho_n^\nu & = \frac{a_n^\nu}{a_{n-1}^\nu}, & \lambda_n^\nu & = \frac{a_n^\nu}{a_{n+1}^\nu},
\end{align}
the former by recursion starting at $n = + \infty$, the latter by recursion starting at $n = - \infty$ (for numerical investigations, at some integer of large modulus). The critical exponent $\nu$ is then determined by solving the transcendental equation
\beq
\label{eq:Transcendental}
 \rho_n^\nu \lambda_{n-1}^\nu = 1.
\eeq
Choosing this minimal solution, it follows by thm.~2.2 of \cite{Gautschi} and the well-known asymptotics of 
$F={}_2 F_1$  that the series for \eqref{eq:R_up} converges in some ellipse in the complex plane with focal points $\{x_\ev, x_\ca\}$ \cite{SuzukiEtAl1999}.

It is obvious from the series \eqref{eq:R_up} that the critical exponent is only determined up to an integer. Furthermore, one can show \cite{SuzukiEtAl1999} that if $\nu$ is a solution to \eqref{eq:Transcendental}, then so is $- \nu - 1$, with the corresponding coefficients related by
\beq
\label{eq:a_reciprocity}
 a^{-\nu-1}_{-n} = a^\nu_n.
\eeq

Although the solution of the transcendental equation \eqref{eq:Transcendental} does not seem possible in closed form, one can in practice get solutions as power series expansions in suitably small 
parameters. For example, one can obtain the following (not necessarily convergent) asymptotic expansion for large $r_\co$ (near RN) and small $\omega$,
\beq
\label{eq:nuAsymptotically}
\nu \sim \ell \left[
1 + \sum_{m,n \ge 0}' \nu_{n,m} \left( \frac{\omega}{\ell} \right)^n \left( \frac{1}{r_\co} \right)^m \right], 
\eeq
where the prime indicates that $(m,n)=(0,0)$ is excluded. The coefficients depend on $\ell$ and $r_\ca, r_\ev$, remain bounded for $\ell \to \infty$,  and can be found 
by direct substitution into the defining equation for $\nu$. The lowest order in $n+m$ coefficients turn out to be
\begin{subequations}
\begin{align}
 \nu_{2,0} & = \frac{ \ell(11 r_\ca^2 + 14 r_\ca r_\ev + 11 r_\ev^2) - \ell^2 (\ell+1) (15 r_\ca^2 + 18 r_\ca r_\ev + 15 r_\ev^2) }{2 (2 \ell + 1) ( 4 \ell (\ell+1) - 3)}, \\
 \nu_{0,2} & = \frac{(15 \ell^4 + 30 \ell^3 + 9 \ell^2 - 6 \ell - 4) (r_\ca^2 + r_\ev^2) + (18 \ell^4 + 36 \ell^3 + 10 \ell^2 - 8 \ell - 4) r_\ca r_\ev}{2\ell  (2 \ell + 1) ( 4 \ell (\ell+1) - 3)}, \\
 \nu_{0,3} & = \frac{-(15 \ell^4 + 30 \ell^3 + 9 \ell^2 - 6 \ell - 4) (r_\ca^2 + r_\ev^2)^2  - (33 \ell^4 + 66 \ell^3 + 19 \ell^2 - 14 \ell - 8) r_\ca r_\ev (r_\ca + r_\ev)}{2 \ell (8 \ell^3 + 12 \ell^2 - 2 \ell -3)}.
\end{align}
\end{subequations}
To determine the normalization of the solution \eqref{eq:R_up}, we note that, by $F(a,b;c;0) = 1$, it behaves as
\beq
 r^{-1} R^{\rup, \rII}_{\omega \ell, \nu}(x) \sim e^{i \omega (r_* + D')} \left(  \frac{-x_\co}{1-x_\co} \right)^{b_\co} \frac{-x_\infty}{1-x_\infty} \sum_{n=-\infty}^\infty a_n^\nu
\eeq
near $r_\ev$, with a constant $D'$, related to the integration constant $D$ in \eqref{eq:r*IntegrationConstant}.
To determine the asymptotic behavior at the Cauchy horizon, we use \cite{AbramowitzStegun}, 15.3.6,
\begin{multline}
\label{eq:F_identity}
 F(a, b; c; x) = \frac{\Gamma(c) \Gamma(c-a-b)}{\Gamma(c-a) \Gamma(c-b)} F(a, b; a+b-c+1; 1-x) \\ + (1-x)^{c-a-b} \frac{\Gamma(c) \Gamma(a+b-c)}{\Gamma(a) \Gamma(b)} F(c-a, c-b; c-a-b+1; 1-x).
\end{multline}
It follows that near $r_\ca$, the solution \eqref{eq:R_up} behaves as
\begin{multline}
 r^{-1} R^{\rup, \rII}_{\omega \ell, \nu} \sim e^{i \omega (r_* + D'')} \sum_{n=-\infty}^\infty a^\nu_n \frac{\Gamma(1 + 2 b_\ev) \Gamma(- 2 b_\ca)}{\Gamma(b_\ev - b_\ca +n+\nu+1) \Gamma(b_\ev - b_\ca-n-\nu)} \\
 + e^{-i \omega (r_* + D'')} \sum_{n=-\infty}^\infty a^\nu_n \frac{\Gamma(1 + 2 b_\ev) \Gamma(2 b_\ca)}{\Gamma(b_\ev + b_\ca-n-\nu) \Gamma(b_\ev + b_\ca+n+\nu+1)} ,
\end{multline}
with another integration constant $D''$. From this, one can read off
\begin{subequations}
\begin{align}
 \sT^{\rup,\rII}_{\omega \ell} & = e^{-i \omega (D'-D'')} \frac{r_\ca}{r_\ev} \left( \frac{1-x_\co}{-x_\co} \right)^{b_\co} \frac{1-x_\infty}{-x_\infty} \Gamma(1+2 b_\ev) \Gamma(-2 b_\ca) \\
 & \qquad \qquad \times \frac{\sum_{n=-\infty}^\infty a^\nu_n [\Gamma(b_\ev - b_\ca+n+\nu+1) \Gamma(b_\ev - b_\ca-n-\nu)]^{-1}}{\sum_{n=-\infty}^\infty a_n^\nu}, \nn \\
 \sR^{\rup,\rII}_{\omega \ell} & = e^{-i \omega (D'+D'')} \frac{r_\ca}{r_\ev} \left( \frac{1-x_\co}{-x_\co} \right)^{b_\co} \frac{1-x_\infty}{-x_\infty} \Gamma(1 + 2 b_\ev) \Gamma(2 b_\ca) \\ 
 & \qquad \qquad \times \frac{\sum_{n=-\infty}^\infty a^\nu_n [\Gamma(b_\ev + b_\ca-n-\nu) \Gamma(b_\ev + b_\ca+n+\nu+1)]^{-1}}{\sum_{n=-\infty}^\infty a_n^\nu}. \nn
\end{align}
\end{subequations}
From the factors $\Gamma(\pm 2 b_\ca)$, one easily verifies the asymptotic behavior \eqref{eq:tau_rho_omega_to_0}. 
Note that in both expressions the fraction of the sums converges to 
\beq
- \frac{1}{\pi} \sin \pi \nu \frac{\sum_{n=-\infty}^\infty (-1)^n a_n^\nu}{\sum_{n=-\infty}^\infty a_n^\nu}
\eeq
in the limit as $\omega \to 0$, so that the divergence at $\omega \to 0$ cannot be enhanced but possibly cancelled if $\nu \to \Z$ for $\omega \to 0$, which, as can be seen from \eqref{eq:nuAsymptotically}, happens in the RN limit, consistent with results by \cite{KehleShlapentokhRothman}.

In order to determine the coefficient $\sR^{\rup,\rI}$, we first define the up mode $R^{\rup, \rI}_{\omega \ell, \nu}$ in region $\rI$ by replacing in \eqref{eq:R_up} the first factor by $(-x)^{b_\ev}$. Then, by \cite{AbramowitzStegun}, 15.3.8, and \eqref{eq:a_reciprocity},
one can write
\beq
 R^{\rup, \rI}_{\omega \ell, \nu}(x) = R^{\ev, \rI}_{\omega \ell, \nu}(z) + R^{\ev, \rI}_{\omega \ell, -\nu-1}(z),
\eeq
where
\begin{multline}
 r^{-1} R^{\ev, \rI}_{\omega \ell, \nu}(z) = (z-1)^{b_\ev} z^{b_\ca} \left( 1 - \frac{z}{z_\co} \right)^{b_\co} \left( 1 - \frac{z}{z_\infty} \right) \Gamma(1+2 b_\ev) \\
 \times \sum_{n=-\infty}^{\infty} a^\nu_n \frac{\Gamma(2n+2\nu+1)}{\Gamma(n+\nu+b_\ev+b_\ca+1) \Gamma(n+\nu+b_\ev-b_\ca+1)} z^{n+\nu-b_\ev-b_\ca} \\
 \times F(b_\ev+b_\ca-n-\nu,b_\ev-b_\ca-n-\nu;-2n-2\nu;1/z).
\end{multline}
Here $z = 1 - x$, and analogously for $z_\co$, $z_\infty$.
By \eqref{eq:F_identity}, the behavior near the event horizon is given by 
\begin{align}
 r^{-1} R^{\ev, \rI}_{\omega \ell, \nu} & \sim e^{i \omega (r_* + D''')} \left( \frac{-x_\co}{1-x_\co} \right)^{b_\co} \frac{-x_\infty}{1-x_\infty} \frac{\sin [\pi (\nu + b_\ev + b_\ca)] \sin [\pi (\nu + b_\ev - b_\ca)] }{\sin 2 \pi b_\ev \sin 2 \pi \nu}
 \sum_{n=-\infty}^\infty a_n^\nu \nn \\
 & \qquad + e^{-i \omega (r_* + D''')} \left( \frac{-x_\co}{1-x_\co} \right)^{b_\co} \frac{-x_\infty}{1-x_\infty} \Gamma(1+2 b_\ev) \Gamma(2 b_\ev) \Gamma(1+2\nu) \Gamma(-2\nu) \nn \\
 & \qquad \quad \times \sum_{n=-\infty}^\infty a_n^\nu \left[ \Gamma(1+n+\nu+b_\ev+b_\ca) \Gamma(1+n+\nu+b_\ev-b_\ca) \right. \nn \\
\label{eq:R_ev_near_r_ev}
 & \qquad \qquad \quad \times \left. \Gamma(-n-\nu+b_\ev+b_\ca) \Gamma(-n-\nu+b_\ev-b_\ca) \right]^{-1},
\end{align} 
with some integration constant $D'''$, where we also used
\beq
\label{eq:GammaReflection}
 \Gamma(z) \Gamma(1-z) = \pi \csc \pi z.
\eeq
An analogous solution in region $\rI$ which is regular at the cosmological horizon $r_\co$ is given by
\begin{multline}
 r^{-1} R^{\co, \rI}_{\omega \ell, \nu}(z) = (z-1)^{b_\ev} z^{b_\ca} \left( 1 - \frac{z}{z_\co} \right)^{b_\co} \left( 1 - \frac{z}{z_\infty} \right) \frac{1}{\Gamma(b_\ev+b_\ca-b_\co-b_o+1)} \\
 \times \sum_{n=-\infty}^{\infty} a^\nu_n \frac{\Gamma(n+\nu-b_\co+b_o+1) \Gamma(n+\nu-b_\co-b_o+1)}{\Gamma(2n+2\nu+2)} \left( \frac{z}{z_\co} \right)^{n+\nu-b_\ev-b_\ca} \\
 \times F(n+\nu+b_\ev-b_\ca+1,n+\nu+b_\ev+b_\ca+1;2n+2\nu+2;z/z_\co).
\end{multline}
It follows from the coefficients in our three-term relation, by thm.~2.2 of \cite{Gautschi} and the well-known asymptotics of 
$F={}_2 F_1$  that the series for \eqref{eq:R_up} converges in some ellipse in the complex $z_\co / z$ plane with focal points $\{0, 1\}$, containing in particular $\{ x_\co, x_\infty \}$ in the complex $x$ plane.
By comparison of coefficients one then finds that \cite{SuzukiEtAl1999}
\beq
\label{eq:K_relation}
 R^{\ev, \rI}_{\omega \ell, \nu} = K^\nu_{\omega \ell} R^{\co, \rI}_{\omega \ell, \nu}
\eeq
with
\beq
\begin{split}
 & K^\nu_{\omega \ell} = z_\co^{p+\nu-b_\ev-b_\ca} \times \\
& \frac{\Gamma(1+2 b_\ev) \Gamma(b_\ev+b_\ca-b_\co-b_o+1)}{\Gamma(p+\nu-b_\ev-b_\ca+1) \Gamma(p+\nu-b_\ev+b_\ca+1) \Gamma(p+\nu+b_\co-b_o+1) \Gamma(p+\nu+b_\co+b_o+1)} \times \\
& \sum_{n=p}^\infty (-1)^{n-p} a_n^\nu \frac{\Gamma(n+\nu-b_\ev-b_\ca+1) \Gamma(n+\nu-b_\ev+b_\ca+1) \Gamma(p+n+2\nu+1)}{\Gamma(n+\nu+b_\ev+b_\ca+1) \Gamma(n+\nu+b_\ev-b_\ca+1) (n-p)!} \times \\
& \bigg( \sum_{n=-\infty}^p a_n^\nu \frac{\Gamma(n+\nu-b_\co+b_o+1) \Gamma(n+\nu-b_\co-b_o+1)}{\Gamma(n+\nu+b_\co-b_o+1) \Gamma(n+\nu+b_\co+b_o+1) \Gamma(p+n+2\nu+2) (p-n)!}\bigg)^{-1}.
\end{split}
\eeq
Here $p$ is an arbitrary integer, on which the result does not depend (keeping this parameter provides a check for the subsequent formulas). The solutions $R^{\co, \rI}_{\omega \ell, \nu}$ and $R^{\co, \rI}_{\omega \ell, - \nu - 1}$ are linearly independent. Using \eqref{eq:a_reciprocity}, \eqref{eq:F_identity}, and the identity \eqref{eq:GammaReflection}, one can show that the linear combination
\begin{multline}
 R^{\rout, \rI}_{\omega \ell, \nu} = \sin [\pi (b_\co + b_o - \nu)] \sin[\pi (b_\co - b_o - \nu)] R^{\co, \mathrm{I}}_{\omega \ell, \nu} \\
 -  \sin [\pi (b_\co + b_o + \nu)] \sin[\pi (b_\co - b_o + \nu)] R^{\co, \mathrm{I}}_{\omega \ell, - \nu - 1}
\end{multline}
is given by
\begin{multline}
 r^{-1} R^{\rout, \rI}_{\omega \ell, \nu}(z) = (z-1)^{b_\ev} z^{b_\ca} \left( 1 - \frac{z}{z_\co} \right)^{b_\co} \left( 1 - \frac{z}{z_\infty} \right) \frac{1}{\Gamma(b_\ev+b_\ca-b_\co-b_o+1)}  \frac{\pi \sin 2 \pi \nu}{\Gamma(1+2 b_\co)} \\
 \times \sum_{n=-\infty}^{\infty} a^\nu_n \left( \frac{z}{z_\co} \right)^{n+\nu-b_\ev-b_\ca}  F(n+\nu+b_\co-b_o+1,n+\nu+b_\co+b_o+1;1+2 b_\co;1-z/z_\co).
\end{multline}
Its asymptotic form near $r = r_\co$ is given by $R^{\rout, \rI}_{\omega \ell, \nu} \sim D_\omega e^{i \omega r_*}$ with some irrelevant factor $D_\omega$. Thus, it represents an outgoing solution near the cosmological horizon. Using \eqref{eq:K_relation}, we may reexpress $R^{\rout, \rI}_{\omega \ell, \nu}$ as
\beq
\label{eq:R_out_R_ev}
 R^{\rout, \rI}_{\omega \ell, \nu} = \sin [\pi (b_\co + b_o - \nu)] \sin[\pi (b_\co - b_o - \nu)] \left( K^\nu_{\omega \ell} \right)^{-1} R^{\ev, \rI}_{\omega \ell, \nu} - \{ \nu \leftrightarrow -\nu -1 \}.
\eeq
From \eqref{eq:R_ev_near_r_ev}, we can thus read off
\begin{align}
 \sR^{\rup,\rI}_{\omega \ell} & = - e^{- 2 i \omega D'} \pi \sin (2 \pi b_\ev) \Gamma(1 + 2 b_\ev) \Gamma(2 b_\ca) \\
 & \times \bigg[  \sin [\pi (b_\co + b_o - \nu)] \sin[\pi (b_\co - b_o - \nu)] \left( K^\nu_{\omega \ell} \right)^{-1} + \{ \nu \leftrightarrow -\nu-1 \} \bigg] \nn \\   
 & \times \sum_{n=-\infty}^\infty a_n^\nu \left[ \Gamma(1+n+\nu+b_\ev+b_\ca) \Gamma(1+n+\nu+b_\ev-b_\ca) \right. \nn \\
 & \qquad \qquad \qquad \times \left. \Gamma(b_\ev+b_\ca-n-\nu) \Gamma(b_\ev-b_\ca-n-\nu) \right]^{-1} \nn \\
 & \times \bigg[ \sin [\pi (b_\co + b_o - \nu)] \sin[\pi (b_\co - b_o - \nu)] \sin [\pi (b_\ev+b_\ca+\nu)] \sin[\pi(b_\ev-b_\ca+\nu)] \left( K^\nu_{\omega \ell} \right)^{-1} \nn \\
 & \qquad + \{ \nu \leftrightarrow \nu-1 \} \bigg]^{-1} \bigg[ \sum_{n=-\infty}^\infty a_n^\nu \bigg]^{-1} . \nn
\end{align}
Regarding the integration constants $D'$, $D''$, one easily checks that these cancel in the next-to-last term in \eqref{eq:del_V_Phi_ell_m_del_V_Phi_ell_m1}, which is the only term sensitive to phases.

\subsection{Numerical results}
\label{sec:Numerics}

From the discussion in the previous sections, two open questions remain: 1) Is the ``density of states'' $n_{\ell}(\omega)$ as given in \eqref{eq:del_V_Phi_ell_m_del_V_Phi_ell_m1} non-vanishing,\footnote{As discussed in the next section, the analog of $n_{\ell}(\omega)$ in the case of the BTZ black hole vanishes identically.} with generically non-vanishing integrals of the form \eqref{eq:del_V_Phi_ell_m_del_V_Phi_ell_m}? 2) Does the sum \eqref{eq:stint} over angular momenta converge and give a generically non-vanishing result, corresponding to a nonzero $C$ in \eqref{eq:MainResult1}, resp.\ \eqref{eq:MainResult}? Without performing an exhaustive parameter scan, we provide numerical evidence that the answer to both questions is affirmative for the case of a conformally coupled scalar field.

As the near extremal case seems to be the most interesting one, due to large violations of sCC in the classical case, we pick an example from that regime. Figure~\ref{fig:L0} shows $\omega n_{0}(\omega)$ for the parameters $r_\co = 100$, $r_\ev = 2$, $r_\ca = 1.95$, corresponding to $\kappa_\co = 0.0094$, $\kappa_\ev = 0.0062$, $\kappa_\ca = 0.0066$. Shown in the figure are results obtained with the method based on Heun functions described in sect.~\ref{sec:modes2}, combined with results obtained by direct numerical integration of the radial equation, which is the method employed in \cite{Zilberman:2019buh} for the RN case. The consideration of the results obtained by direct integration of the radial equation not only provides a consistency check for the results obtained with the Heun function method, but is at present also necessary as the search for a solution of \eqref{eq:Transcendental}, i.e.\ for the critical exponent $\nu$, becomes difficult for large frequencies $\omega$ for which the asymptotic expansion \eqref{eq:nuAsymptotically} does not provide a useful estimate for the starting point of a numerical search algorithm.\footnote{However, we do not doubt that with more effort the range of applicability of the Heun function method can be extended to greater frequencies.}
On the other hand, the Heun method seems to have advantages for $\omega \to 0$ where some scattering coefficients diverge.

\begin{figure}
\centering
\includegraphics[scale=0.8]{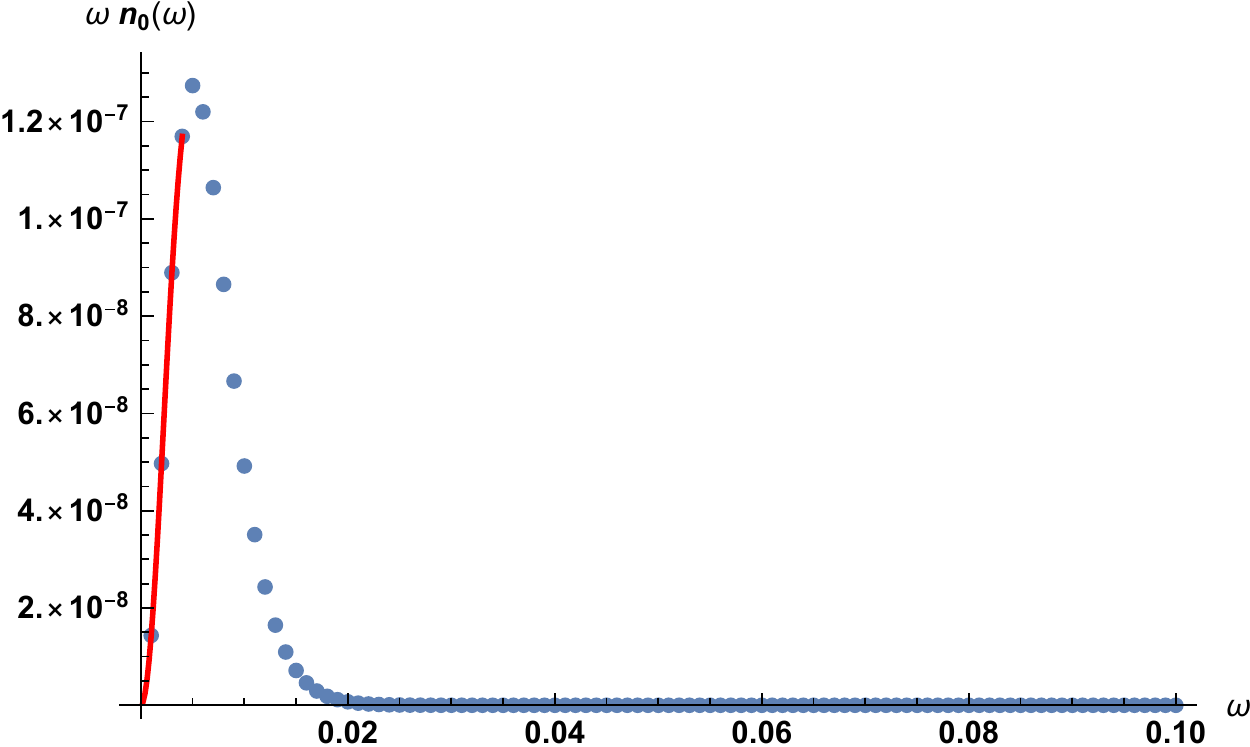}
\caption{Numerical results for $\omega n_0(\omega)$ for the parameters $r_\co = 100$, $r_\ev = 2$, $r_\ca = 1.95$. In red the results from the Heun function method, in blue results obtained by direct numerical integration of the radial equation.}
\label{fig:L0}
\end{figure}

At any rate, clearly $\omega n_{0}(\omega)$ is not only non-vanishing, but also has a non-vanishing integral, answering the first of the above questions. As for the second question, one finds that $\omega n_{1}(\omega)$ is already of order $10^{-15}$, i.e.\ suppressed by eight orders of magnitude. The contributions of higher angular momenta are already drowned in numerical noise. The same behavior, i.e.\ non-vanishing integrands $\omega n_\ell(\omega)$ with non-vanishing integrals and a rapid decrease in angular momentum were also found for the other parameters $r_X$ at which we numerically evaluated $\omega n_\ell (\omega)$ (both in the near-extremal and the ``normal'' regime $\kappa_\ca > \kappa_\co$). We thus conclude that the coefficient $C$ in \eqref{eq:MainResult} is indeed generically non-vanishing.

\section{Comparison with the BTZ black hole}
\label{sec:BTZ}

In \cite{Dias:2019ery} Dias et al. have studied the stress tensor of a scalar field obeying $(\square - \mu^2)\Phi=0$ in the 3-dimensional BTZ black hole spacetime \cite{BTZ92, BTZ93}. In local coordinates it has the metric 
\beq
  g = - f \ud t^2 + f^{-1} \ud r^2 + r^2 ( \ud \phi - \Omega \ud t )^2,
\eeq
with
\begin{align}
 f & = \frac{(r^2 - r_+^2)(r^2 - r_-^2)}{r^2}, &
 \Omega & = \frac{r_+ r_-}{r^2},
\end{align}
and Penrose diagram as in Figure~\ref{fig:BTZ}. (For simplicity, we have set the AdS radius $L = 1$; we refer to \cite{Dias:2019ery} for details on how to construct the diagram.) Even though the BTZ black hole is a solution with a negative cosmological constant, these authors showed that the regularity of the classical stress tensor $T_{\mu\nu}$ is governed by a parameter $\beta$ similar to 
\eqref{eq:betadef}. However, in contrast with RNdS, they showed that $\alpha$ is not the spectral gap \eqref{eq:aldef} of all quasi-normal modes in region I, but only of 
the ``counter-rotating'' ones. As a consequence, it can be seen \cite{Dias:2019ery} that $\beta \to \infty$ in the BTZ case as the spacetime approaches extremality $r_- \to r_+$. 
Hence, classically, the stress tensor becomes as regular as we wish across the Cauchy horizon, implying an even stronger violation of sCC than in the case of RNdS at the classical level. 
The authors \cite{Dias:2019ery} have also considered the expected quantum stress tensor in the Hartle-Hawking state and have presented arguments that it similarly becomes arbitrarily regular at the Cauchy horizon  when we approach extremality, in marked contrast to our results in RNdS. 

\begin{figure}
\centering
\begin{tikzpicture}[scale=1.0]
 \draw (0,0) -- (-2,2) node[midway, below, sloped]{$\EH^-$};
 \draw (-2,2) -- (-4,4) node[midway, below, sloped]{$\EH^L$};
 \draw[double] (0,0) -- (0,4) node[midway, right]{$I_R$};
 \draw (-2,2) -- (0,4) node[midway, above, sloped]{$\EH^R$};
 \draw (0,4) -- (-2,6) node[midway, above, sloped]{$\CH^R$};
 \draw (-4,4) -- (-2,6) node[midway, above, sloped]{$\CH^L$};

 \draw[snake it] (0,4) -- (0,8);
 \draw (-2,6) -- (0,8) node[midway, above, sloped]{$\CH^+$};
 \draw (-1,2) node{${\rm I}$};
 \draw (-2,4) node{${\rm II}$};
 \draw (-1,6) node{${\rm IV}$};

 \draw[fill=white] (-4,4) circle (2pt);
 \draw[fill=white] (0,8) circle (2pt);
 \draw[fill=black] (-2,6) circle (2pt);
 \draw[fill=white] (0,0) circle (2pt);
 \draw[fill=black] (-2,2) circle (2pt);
 \draw[fill=white] (0,4) circle (2pt);
  
\end{tikzpicture}
\caption{Sketch of the relevant part of the BTZ spacetime. Here the double line indicates the asymptotic timelike boundary.}
\label{fig:BTZ}
\end{figure}
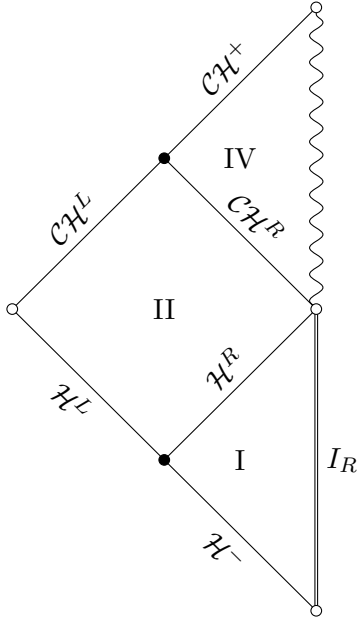

As \cite{Dias:2019ery} have pointed out, these results are in effect due to a cancellation that appears to be special to the BTZ case. Nevertheless, since the behavior of the quantum stress-energy tensor near $CH^R$ for the BTZ black hole found by \cite{Dias:2019ery} is very different from the behavior we have found in the RNdS case, it is instructive to understand how this difference arises from the perspective of the type of mode calculation given in the previous section. Therefore, we will now analyze the BTZ case by the methods of the previous section. The upshot of the analysis that we will now give is that the constant $C$ appearing in \eqref{eq:MainResult} is zero in the BTZ case, consistent with the findings of \cite{Dias:2019ery}.

Following \cite{Dias:2019ery}, we introduce the variable 
\beq
 z = \frac{r^2 - r_-^2}{r_+^2 - r_-^2},
\eeq
and write generic mode solutions as
\beq
 \psi_{\omega m}(t,z,\phi) = e^{- i \omega t} e^{i m \phi} R_{\omega m}(z).
\eeq
Solutions are given in terms of hypergeometric functions involving the parameters
\begin{subequations}
\begin{align}
 a & = \frac{1}{2} \left( \Delta - i \frac{\omega - m \Omega_-}{\kappa_-} - i \frac{\omega - m \Omega_+}{\kappa_+} \right), \\
 b & = \frac{1}{2} \left( 2 - \Delta - i \frac{\omega - m \Omega_-}{\kappa_-} - i \frac{\omega - m \Omega_+}{\kappa_+} \right), \\
 c & = 1 - i \frac{\omega - m \Omega_-}{\kappa_-} .
\end{align}
\end{subequations}
Here $\Delta = 1 + \sqrt{1 + \mu^2}$, and the surface gravities $\kappa_\pm$ and the angular velocities $\Omega_\pm$ of event and Cauchy horizon, which can be expressed in terms of the horizon radii $r_\pm$ as
\begin{align}
 \kappa_\pm & = \frac{r_+^2 - r_-^2}{r_\pm}, & \Omega_\pm & = \frac{r_\mp}{r_\pm} = \frac{\kappa_\pm}{\kappa_\mp}.
\end{align}
The solutions with special behavior near the Cauchy horizon $r_-$ are given by
\begin{subequations}
\begin{align}
 R^{\rout, -}_{\omega m} (z)& = z^{- \frac{1}{2} (1-c)} (1-z)^{\frac{1}{2} (a+b-c)} F(a, b; c; z), \\
 R^{\rin, -}_{\omega m} (z)& = z^{+ \frac{1}{2} (1-c)} (1-z)^{\frac{1}{2} (a+b-c)} F(a-c+1, b-c+1; 2- c; z).
\end{align}
\end{subequations}
The solutions with special behavior near the event horizon $r_+$ are given by
\begin{subequations}
\begin{align}
 R^{\rout, +}_{\omega m} (z)& = z^{- \frac{1}{2} (1-c)} |1-z|^{-\frac{1}{2} (a+b-c)} F(c-b, c-a; -a-b+c+1; 1-z), \\
 R^{\rin, +}_{\omega m} (z)& = z^{- \frac{1}{2} (1-c)} |1-z|^{+\frac{1}{2} (a+b-c)} F(a, b; a+b-c+1; 1-z).
\end{align}
\end{subequations}
Finally, the solution with the fast decay at infinity is given by
\beq
 R^{0, \infty}_{\omega m} (z)= z^{- \frac{1}{2} (2a-c+1)} (z-1)^{\frac{1}{2} (a+b-c)} F(a, a-c+1; a-b+1; 1/z).
\eeq
From the linear transformation formulas for hypergeometric functions, it follows that
\begin{subequations}
\begin{align}
 R^{\rout, +}_{\omega m} & = \mathscr{A}_{\omega m} R^{\rout, -}_{\omega m} + \mathscr{B}_{\omega m} R^{\rin, -}_{\omega m}, \\
 R^{\rin, +}_{\omega m} & = \tilde{\mathscr{A}}_{\omega m} R^{\rin, -}_{\omega m} + \tilde{\mathscr{B}}_{\omega m} R^{\rout, -}_{\omega m}, \\
 \tilde{\mathscr{T}}_{\omega m} R^{0, \infty}_{\omega m} & = R^{\rout, +}_{\omega m} + \tilde{\mathscr{R}}_{\omega m} R^{\rin, +}_{\omega m},
\end{align}
\end{subequations}
where
\begin{subequations}
\begin{align}
 \mathscr A_{\omega m} & = \frac{\Gamma(1-c) \Gamma(1-a-b+c)}{\Gamma(1-a) \Gamma(1-b)}, &
 \mathscr B_{\omega m} & = \frac{\Gamma(c-1) \Gamma(1-a-b+c)}{\Gamma(c-a) \Gamma(c-b)}, \\
 \tilde{\mathscr A}_{\omega m} & = \frac{\Gamma(c-1) \Gamma(a+b-c+1)}{\Gamma(a) \Gamma(b)}, &
 \tilde{\mathscr B}_{\omega m} & = \frac{\Gamma(1-c) \Gamma(a+b-c+1)}{\Gamma(a-c+1) \Gamma(b-c+1)}, \\
 \tilde{\mathscr T}_{\omega m} & = \frac{\Gamma(a) \Gamma(a-c+1)}{\Gamma(a-b+1) \Gamma(a+b-c)}, &
 \tilde{\mathscr R}_{\omega m} & = \frac{\Gamma(a) \Gamma(a-c+1) \Gamma(c-a-b)}{\Gamma(1-b) \Gamma(c-b) \Gamma(a+b-c)}.
\end{align}
\end{subequations}
We also recall the definition 
\beq
 r_* = \frac{1}{2 \kappa_+} \log \left| \frac{r-r_+}{r+r_+} \left( \frac{r+r_-}{r-r_-} \right)^{\Omega_+} \right|
\eeq
of the tortoise coordinate and introduce the corotating angles
\beq
 \phi_\pm = \phi - \Omega_\pm t.
\eeq

As before, we use boundary data to define the relevant states. To define the comparison state we use the modes with boundary data
\beq
 \psi^{\rout}_{\omega m} = \begin{cases} (2 \pi)^{-1} (2 | \omega |)^{-\frac{1}{2}} r_-^{-\frac{1}{2}} e^{i m \phi_-} e^{- i \omega v} & \text{ on } \CH^L, \\ 0 & \text{ on } \CH^R.\end{cases} 
\eeq
One easily checks that
\beq
 \psi^{\rin, -}_{\omega m} = (2 \pi) (2 | \omega - m \Omega_- |)^{\frac{1}{2}} r_-^{\frac{1}{2}} \left( \frac{4 r_-}{\kappa_-} \right)^{i \frac{\omega - m \Omega_-}{2 \kappa_-}} \left( \frac{r_+ - r_-}{r_+ + r_-} \right)^{i \frac{\omega - m \Omega_-}{2 \kappa_+}} \psi^\rout_{\omega-m \Omega_-, m}.
\eeq
To define the Hartle-Hawking (HH) state state (an analog of our Unruh state), we use the modes with boundary data
\begin{subequations}
\begin{align}
 \psi^{\rup, \rI}_{\omega m} & = \begin{cases} (2 \pi)^{-1} (2 | \omega |)^{-\frac{1}{2}} r_+^{-\frac{1}{2}} e^{i m \phi_+} e^{- i \omega u} & \text{ on } \EH^-, \\ \order(z^{-\Delta/2}) & \text{ at } \infty, \end{cases} \\
 \psi^{\rup, \rII}_{\omega m} & = \begin{cases} (2 \pi)^{-1} (2 | \omega |)^{-\frac{1}{2}} r_+^{-\frac{1}{2}} e^{i m \phi_+} e^{- i \omega u} & \text{ on } \EH^L, \\ 0 & \text{ on } \EH^R. \end{cases}
\end{align}
\end{subequations}
In $\rI$ and $\rII$, respectively, these are related to the mode solutions introduced above by
\begin{subequations}
\begin{align}
 \psi^{\rout, +}_{\omega m} + \tilde{\mathscr{R}}_{\omega m} \psi^{\rin, +}_{\omega m} & = (2 \pi) (2 | \omega - m \Omega_+ |)^{\frac{1}{2}} r_+^{\frac{1}{2}} \left( \frac{4 r_+}{\kappa_+} \right)^{+ i \frac{\omega - m \Omega_+}{2 \kappa_+}} \left( \frac{r_+ - r_-}{r_+ + r_-} \right)^{+ i \frac{\omega - m \Omega_+}{2 \kappa_-}} \psi^{\rup, \rI}_{\omega - m \Omega_+, m}, \\
 \psi^{\rout, +}_{\omega m} & = (2 \pi) (2 | \omega - m \Omega_+ |)^{\frac{1}{2}} r_+^{\frac{1}{2}} \left( \frac{4 r_+}{\kappa_+} \right)^{+ i \frac{\omega - m \Omega_+}{2 \kappa_+}} \left( \frac{r_+ - r_-}{r_+ + r_-} \right)^{+ i \frac{\omega - m \Omega_+}{2 \kappa_-}} \psi^{\rup, \rII}_{\omega - m \Omega_+, m}.
\end{align}
\end{subequations}
Similar to \cite{Lanir:2017oia}, and analogously to \eqref{eq:G_MU}, the symmetrized two-point function in the Hartle-Hawking state can be expressed as
\begin{multline}
 \langle \{ \Phi(x_1), \Phi(x_2) \} \rangle_{\rHH} = \sum_{m=-\infty}^\infty \int_{-\infty}^\infty \ud \omega \left[  \frac{\sign \omega}{\sinh \frac{\pi \omega}{\kappa_+}} \Re \left\{ \psi^{\rup, \rI}_{\omega, m}(x_1), \overline{\psi^{\rup, \rII}_{\omega, m}}(x_2) \right\}   \right. \\ \left. + \frac{\sign \omega}{1- e^{- \frac{2\pi \omega}{\kappa_+}}} \left( \left\{ \psi^{\rup, 
 \rI}_{\omega, m}(x_1), \overline{\psi^{\rup, \rI}_{\omega, m}}(x_2) \right\} + \left\{ \psi^{\rup, \rII}_{-\omega, m}(x_1), \overline{\psi^{\rup, \rII}_{-\omega, m}}(x_2) \right\} \right)\right]
\end{multline}
Restricting to region $\rII$, using the above relations and neglecting the $\psi^{\rout,-}$ modes, which are irrelevant for the restriction to $\CH^L$ that we want to perform in the end, one obtains
\begin{align}
 \langle \{ \Phi(x_1), \Phi(x_2) \} \rangle_{\rHH} & = \frac{r_-}{r_+} \sum_{m=-\infty}^\infty \int_{-\infty}^\infty \ud \omega \frac{| \omega |}{| \omega - m \Omega_+ + m \Omega_- |} \\
 & \times \Bigg[ \frac{\sign (\omega - m \Omega_+ + m \Omega_-) }{1- e^{- \frac{2\pi (\omega - m \Omega_+ + m \Omega_-)}{\kappa_+}}} | \tilde{\mathscr{A}}_{\omega + m \Omega_-, m} |^2 \left\{ \psi^{\rout}_{\omega, m}(x_1), \overline{\psi^{\rout}_{\omega, m}}(x_2) \right\} \nn \\
 & \quad + \frac{\sign (\omega - m \Omega_+ + m \Omega_-) }{1- e^{- \frac{2\pi (\omega - m \Omega_+ + m \Omega_-)}{\kappa_+}}} | \mathscr{B}_{\omega + m \Omega_-, m} |^2 \left\{ \psi^{\rout}_{-\omega, m}(x_1), \overline{\psi^{\rout}_{-\omega, m}}(x_2) \right\} \nn \\
 & \quad + \frac{\sign (\omega - m \Omega_+ + m \Omega_-)}{\sinh \frac{\pi (\omega - m \Omega_+ + m \Omega_-)}{\kappa_+}} \nn \\
 & \quad \qquad \times \Re \left\{ \tilde{\mathscr{R}}_{\omega+m \Omega_-, m} \tilde{\mathscr{A}}_{\omega + m \Omega_-, m} \psi^{\rout}_{\omega, m}(x_1), \overline{\mathscr{B}_{\omega+m \Omega_-, m}} \overline{\psi^{\rout}_{\omega, m}}(x_2) \right\} \Bigg]. \nn
\end{align}
For the difference of the Hartle-Hawking and the comparison state, differentiated $n$ times \wrt $v$ and restricted to $\CH^L$, one thus obtains formally, after symmetrization under the simultaneous transformation $\omega \to - \omega$, $m \to - m$,
\beq
\bigg[ \langle \del_v^n \Phi(x_1) \del_v\Phi(x_2) \rangle_{\rHH} - \langle \del_v^n \Phi(x_1) \del_v^{n} \Phi(x_2) \rangle_\rF \bigg]_{
x_1 \to x_2 \to \CH^L}
= \frac{1}{8 \pi^2 r_-^2} \sum_{m = - \infty}^\infty \int_0^\infty \ud \omega \ \omega^{2n-1} n_{m}(\omega),
\eeq
where 
\begin{align}
\label{eq:nbtz}
 n_{m}(\omega) & = \frac{r_-}{r_+} \frac{\omega}{\omega - m \Omega_+ + m \Omega_-} \left[ \coth \frac{\pi (\omega - m \Omega_+ + m \Omega_-)}{\kappa_+} \left( | \tilde{\mathscr{A}}_{\omega + m \Omega_-, m} |^2 + | \mathscr{B}_{\omega + m \Omega_-, m} |^2 \right) \right. \\
 & \left. \quad + 2 \csch \frac{\pi (\omega - m \Omega_+ + m \Omega_-)}{\kappa_+} \Re \left( \tilde{\mathscr{R}}_{\omega+m \Omega_-, m} \tilde{\mathscr{A}}_{\omega + m \Omega_-, m} \overline{\mathscr{B}_{\omega+m \Omega_-, m}} \right) \right] - \coth \frac{\pi \omega}{\kappa_-}. \nn
\end{align}

Equation \eqref{eq:nbtz}  is the analog of \eqref{eq:del_V_Phi_ell_m_del_V_Phi_ell_m}. In the present case, however, one can use standard identities for products of $\Gamma$ functions and trigonometric identities to show that $n_{m}(\omega)$ identically vanishes! This means in particular that $\langle T_{vv} \rangle_{\rHH}-\langle T_{vv} \rangle_{\rF}
\propto  \sum_{m} \int \ud \omega \ \omega n_{m}(\omega)$ vanishes 
at $\CH^L$, and since both states are stationary, i.e. invariant under the flow of $\partial_t$, the same must be true at $\CH^R$. Hence, the constant $C$ appearing in \eqref{eq:MainResult} is zero in the BTZ case, whereas we have presented evidence that $C \neq 0$ generically in the case of RNdS.
Thus, our findings are mathematically consistent with those of \cite{Dias:2019ery}. We believe that the vanishing of $n_{m}(\omega)$ is a very special property of the BTZ black hole spacetime.

\section{Conclusions}
\label{sec:Conclusion}

In this work, we have analyzed the expected stress tensor of a real, linear Klein-Gordon quantum field near the Cauchy horizon inside the RNdS black hole. 
We have presented arguments that, due entirely to quantum effects, the $VV$-component of the stress tensor diverges as $V^{-2}$ near the Cauchy horizon (defined by $V=0$)
with a coefficient that does not depend on the initial state, provided only that the state is regular (Hadamard) near an initial Cauchy surface extending beyond the cosmological horizon. Numerical 
evidence indicates that this coefficient is not zero at least for the black hole parameter values which we have considered. Since the coefficient would at any rate be expected to be basically an analytic function of the black hole parameters, it should be zero only on a set of measure zero of finely tuned parameters.   

Although we have studied in this paper the case of RNdS spacetime, our methods should apply with relatively minor modifications to the the Kerr-Newman-dS black holes, and we strongly expect the same conclusions. More generally, we expect the behavior \eqref{eq:MainResult} to hold generically in eternal black hole spacetime whose Cauchy horizon is a bifurcate Killing horizon.

Our analysis does not preclude the existence of spacetimes for which $C = 0$ and for which the behavior of the quantum stress-energy is less singular than \eqref{eq:MainResult}. Examples of such spacetimes are the 2 dimensional RNdS spacetime with $\kappa_\co = \kappa_\ca$ and the 3 dimensional BTZ black hole spacetime analyzed by \cite{Dias:2019ery}. 
Nevertheless, we shall argue elsewhere that even for such spacetimes in which the quantum stress-energy is not (badly) singular as one approaches the Cauchy horizon, a Hadamard state can not be extended beyond the Cauchy horizon, which typically manifests itself in the divergence of some local Wick power at the horizon. 
Thus, it appears that quantum field effects enforce strong cosmic censorship---or, at least, move it to a regime where full quantum gravity will be needed to determine to what extent it holds.

\medskip
{\bf Acknowledgements:} S.H.\ is grateful to the Max-Planck Society for supporting the collaboration between MPI-MiS and Leipzig U., grant Proj.~Bez.\ M.FE.A.MATN0003. The research of R.M.W. was supported by NSF grant PHY18-04216 to the University of Chicago. S.H.\ thanks A. Ishibashi and G.~t'Hooft for discussions concerning black hole interiors and P.~Zimmerman for discussions concerning Heun's equation and F.~Otto. J.Z.\ thanks J.~Derezinski for discussions on scattering problems. We thank Tom Enders (MPI-MiS) for assistance with figures.


\end{document}